\documentclass[12pt]{iopart}
\expandafter\let\csname equation*\endcsname\relax
\expandafter\let\csname endequation*\endcsname\relax
\usepackage[utf8]{inputenc}
\usepackage[T1]{fontenc}
\usepackage{iopams}
\usepackage{amsmath,amssymb,amsfonts,amsthm}
\usepackage{graphicx,mathtools,pbox,bbold}
 \usepackage{subfig}
 \usepackage{times,math dots,bbm,tabularx, bm,dsfont}
\usepackage[usenames]{color}
\usepackage{todonotes}
\usepackage{float} %% to put a figure exactly where we want
\usepackage{tikz,tikz-cd}
\usetikzlibrary{decorations.markings}

\newtheorem{theorem}{Theorem}%[section]
\newtheorem{proposition}[theorem]{Proposition}

\usepackage[toc,page]{appendix}
\usepackage{hyperref} 
\renewcommand{\figurename}{Figure }

\begin{document}

	 \title[]{Modified toric code models with flux attachment from Hopf algebra gauge theory}

     \author{A. Conlon $^{1,2}$,  D. Pellegrino$^{2}$ , J.K. Slingerland $^{1,2}$. }

    \address{ $^1$ Dublin Institute for Advanced Studies, School of Theoretical Physics, 10 Burlington Rd, Dublin, Ireland, \newline $^2$ Department of Theoretical Physics, Maynooth University, Ireland. }

    \ead{aaron.conlon@stp.dias.}

\date{\today}

\begin{abstract}
	Kitaev's toric code is constructed using a finite gauge group from gauge theory. Such gauge theories can be generalized with the gauge group generalized to any finite-dimensional semisimple Hopf algebra. This also leads to generalizations of the toric code. 
	Here we consider the simple case where the gauge group is unchanged but furnished with a non-trivial quasitriangular structure (R-matrix), which modifies the construction of the gauge theory. 
	This leads to some interesting phenomena; for example, the space of functions on the group becomes a non-commutative algebra. 
	We also obtain simple Hamiltonian models generalizing the toric code, which are of the same overall topological type as the toric code, except that the various species of particles created by string operators in the model are permuted in a way that depends on the R-matrix. In the case of $\mathbb{Z}_{N}$ gauge theory, we find that the introduction of a non-trivial R-matrix amounts to flux attachment. 
\end{abstract}
\newline
% \noindent{\submitto{\NJP}}
\maketitle

% We calculate the excitations and exchange statistics in Hopf algebra gauge theory for Hopf algebras constructed over a finite abelian group.
% We demonstrate how the introduction of the braided tensor product changes the identification of particle exchange statistics in Hopf algebra lattice gauge theories when compared with Kitaev models.
% % We focus on $\mathcal{A}^{*}$, the algebra of functions on connections.
% %  % this algebra is modified with a braided tensor product.
%   In particular, we show that for $\mathbb{C}\mathbb{Z}_{N}$ valued connections, the functions acquire non-commutative behaviour.
%  We describe an exact mapping of our model to the  $D(\mathbb{Z}_{N})$ Kitaev models.

\section{Introduction}

It was proposed in \cite{Kitaev2003}, that anyonic excitations could be utilised to implement quantum computation in a fault-tolerant manner.
The model introduced there is often called the Kitaev model or the toric code, since it can be described as a stabiliser code on a toroidal geometry~\cite{QuditSurfaceCode}. 
It is essentially a discrete gauge theory, based on a finite group $G$, where gauge invariance is encoded through an energy penalty for gauge violations, so that the ground states of the model are gauge invariant, while excited states can have electric charge, in which case they transform non-trivially under gauge transformations. 
In recent years a renewed interest in these models has developed with the advent of quantum simulators~ \cite{Minimal,ToricSim}. In fact, there have recently been experimental measurements of anyonic braiding statistics in systems of qubits hosting the toric code ground state~\cite{PhysRevLett.121.030502,PhysRevLett.117.110501,Experimental_Toric_Code}.

On the more mathematical side, these models have many interesting features. The distinct types of excitations are labelled by irreducible representations of $D(G)$, the Drinfeld double, or the quantum double of the gauge symmetry $G$. 
 See \cite{Bais_Propitius}, for an overview of how this structure arises in discrete gauge theories.
Furthermore, although originally defined over a finite group, as pointed out in \cite{Kitaev2003}, the construction can be generalised to any finite-dimensional semisimple Hopf algebra \cite{KitaevModels_FiniteGroup,Buerschaper2013}. 
Many of the tools developed in Kitaev models such as the formalism of ribbon operators have already been extended to these generalised Kitaev models~\cite{GenRibbonOp}. 
Although many of the developments in toric code models are based on closed manifolds, the framework has been developed to study boundaries~\cite{QuantumCodesLatticeBoundary,ExactBoundaryKitaev,D3FusionModularInvar}, and even explicit constructions of boundary theories for generalised Kitaev models have been introduced  ~\cite{Boundary_Generalized_Kitaev_Hopf_Gauge,AlgBoundary}.

Another notable development in the area has been the introduction of axioms for Hopf algebra gauge theory, which generalises discrete group gauge theory~\cite{Catherine1} and the equivalence with Kitaev models, \cite{Catherine2}.
In particular, a relation is established between the topological invariants of both frameworks, i.e. the space of gauge invariant flat connections in Hopf algebra gauge theory is equivalent to the space of protected states in Kitaev models. 
 
In this work, we shall closely follow~\cite{Catherine1} and examine the excitations in Hopf algebra gauge theory and the relation to Kitaev models.
Our focus will be on taking a simple class of Hopf algebras and calculating the relevant quantities, like $\mathcal{R}$-symbols, from a lattice gauge theory construction.
In particular, our gauge theory will be given by a quasitriangular semisimple Hopf algebra constructed on $\mathbb{C}\mathbb{Z}_{N}$.
With the trivial choice of quasitriangular structure, this simply reproduces the usual $\mathbb{Z}_{N}$ gauge theory and $\mathbb{Z}_{N}$ toric code.
However, $\mathbb{C}\mathbb{Z}_{N}$ allows for a choice of non-trivial quasitriangular structure. We will show how this choice affects the braiding of excitations in the model and leads to a notion of flux attachment.

% An outline of the paper is as follows.
In Section \ref{sec::Background}, we provide a brief overview of some of the relevant material in \cite{Catherine1}.
In Section \ref{sec::LatticeHopf} we will discuss our model and analyse the excitations of the operators defined in the preceding sections. 
We will show how the introduction of a non-trivial R matrix on the Hopf algebra can lead to a different identification of the quasiparticle content when compared with the usual Kitaev models. 
Finally, in Section \ref{sec::Conclusions}, we draw conclusions from our results. 
% and indicate possible directions for further research.

% %%%%%%%%%%%%%%%%%%%%%%%%%%%%%%%%%%%%%%%%%%%%%%%%%%%% 
\section{Background on Hopf algebra gauge theory}
\label{sec::Background}
A comprehensive summary of the differences caused by generalising the input gauge object from a group to $K$, a finite-dimensional semisimple quasitriangular Hopf algebra, can be found in \cite{Catherine1} and we report this schematically in Figure \ref{fig::differenceTable}.

\begin{figure}[H]
\begin{center}
\begin{tabular}{|c|c |c |c|}
 \hline
  &  Lattice gauge theory   &  Hopf algebra gauge theory \\
 \hline
  gauge object & group $G$ & Hopf algebra $K$ \\
  connection & $G$-colouring of graph &  $K$-colouring of graph\\
 collections of connections & $G^{E}$ & $K^{E}$ \\
  gauge transformation on connection & $G$-set & $K^{V}$- module structure on $K^{E}$ \\
gauge transformation on functions &dual  $G$ set 
& \vtop{\hbox{\strut $K^{V}$- module algebra}\hbox{\strut  structure on $K^{*E}$ }} 
\\
 \hline
\end{tabular}
\end{center}
 \caption{This table shows how many of the standard ingredients of lattice gauge theory are generalised to Hopf algebra gauge theory. $E$ and $V$ indicate the collection of edges and vertexes of the graph respectively. 
 By $K^{E}$ ($K^{V}$) we denote the tensor product of all of the Hopf algebras associated with each edge (vertex).
 }
 \label{fig::differenceTable}
 \end{figure}

In this section we shall follow \cite{Catherine1}, where it is explained how to construct a Hopf algebra gauge theory on a  ribbon graph from the description of the gauge theory on smaller, local constituents of the graph, called local vertex neighbourhoods. 
Firstly a basis for the Hilbert space in discrete group gauge theory is given by assigning group elements to the edges of the graph, often called ``colouring'' the graph by $G$. The situation is analogous for Hopf algebra gauge theory except now the graph is coloured by elements of the Hopf algebra.  

A ciliated ribbon graph is used as a discrete model for space, this is a directed graph with a cyclic ordering of the edges at each vertex. The edges are enumerated anti-clockwise starting from the cilium, as we display in Figure \ref{fig::ciliatedVertex}.
The cilia are not themselves edges of the graph, but can instead be thought of as a bookkeeping device which keeps track of the ordering of edges at a vertex, so that one can associate an oriented surface with a boundary to a ribbon graph. 
This is done by replacing edges with ribbons, vertices with disks and glueing ribbons and disks according to the cyclic ordering defined by the cilia. 

\begin{figure}[t]
\centering
   \subfloat[]{
   \begin{tikzpicture}[>=latex,scale=0.7,>=latex,very thick,decoration={
       markings,
       mark= at position 0.65 with {\arrow{>}}}]
   	\draw[postaction={decorate}] ({0},{-3})--({0},{0});
   	\draw[postaction={decorate}] ({0},{0})--({0},{3});
   	\draw[fill] ({0},{0}) circle [radius=0.1];
   	\draw[postaction={decorate}] ({0},{0})--({-3},{0});
   	\draw[postaction={decorate}] ({0},{0})--({3},{0});
   	\draw[postaction={decorate}] ({2.4)},-{2.4})--({0},{0});
   	\draw[postaction={decorate}] ({2.4},{2.4})--({0},{0});
   	\draw[postaction={decorate}] ({0},{0})--({-2.4},{2.4});
   	\draw[very thick,dotted] ({-0.75*sqrt(2)},{-0.75*sqrt(2)}) -- ({0},{0});   
       \node [left] at (0,-3) {$1$};
       \node[left] at (2.25,-2.505) {$2$};
       \node[below] at (3,0) {$3$};
       \node[below] at (2.505,2.4) {$4$};
       \node[right] at (0.5,2.5) {$\ldots$};
       \node[right] at (-2.25,2.475) {$n-1$};
       \node[above] at (-3,0) {$n$};
   \end{tikzpicture}
   	\label{fig::ciliatedVertex}
	}  
   \qquad\qquad
   \subfloat[]{
   \begin{tikzpicture}[scale=0.3,>=latex,decoration={
       markings,
       mark=at position 0.65 with {\arrow{>}}}
       ] 
       \foreach \x in {0,8}
       {
   	 {
   	    \draw[postaction={decorate},very thick] ({\x},{8})--({\x+4},{8});
   	    \draw[fill] ({\x},{8}) circle [radius=0.1];
   	    \draw[very thick,dotted] ({-sqrt(2)/2+\x},{-sqrt(2)/2+8}) -- (\x,8);
   	}
   	\foreach \y in {4,12} % spacing
   	{
   	    \draw[postaction={decorate},very thick] ({\x+4},{\y})--({\x},{\y});
   	    \draw[fill] ({\x},{\y}) circle [radius=0.1];
   	    \draw[very thick,dotted] ({-sqrt(2)/2+\x},{-sqrt(2)/2+\y}) -- (\x,\y);
   	}
       }
       \foreach \x in {4,12}
       {
   	{
   	    \draw[postaction={decorate},very thick]
   	    ({\x+4},{8})--({\x},{8});
   	    \draw[fill] ({\x},{8}) circle [radius=0.1];
   	    \draw[very thick,dotted] ({-sqrt(2)/2+\x},{-sqrt(2)/2+8}) -- (\x,8);
   	}
   	\foreach \y in {4,12} % spacing
   	{
   	    \draw[postaction={decorate},very thick] 
   	    ({\x},{\y})--({\x+4},{\y});
   	    \draw[very thick,dotted] ({-sqrt(2)/2+\x},{-sqrt(2)/2+\y}) -- (\x,\y);
   	}
       }
      
       \foreach \y in {4,12}
       {
   	\foreach \x in {0,8} % spacing
   	{
   	    \draw[postaction={decorate},very thick] ({\x},{\y})--({\x},{\y+4});
   	    %\draw[fill] ({\x},{\y}) circle [radius=0.1];
   	}
   	\foreach \x in {4,12} % spacing
   	{
   	    \draw[postaction={decorate},very thick] ({\x},{\y+4})--({\x},{\y});
   	    \draw[fill] ({\x},{\y}) circle [radius=0.1];
   	}
       }
      \foreach \y in {0,8} % keep zero
       {
   	\foreach \x in {0,8} % spacing
   	{
   	    \draw[postaction={decorate},very thick]
   	    ({\x},{\y+4})--({\x},{\y});
   	}
   	\foreach \x in {4,12} % spacing
   	{
   	    \draw[postaction={decorate},very thick] 
   	    ({\x},{\y})--({\x},{\y+4});
   	}
       }

   	\foreach \x in {-4}
       {
   	{
   	    \draw[postaction={decorate},very thick] 
   	    ({\x+4},{8})--({\x},{8});
   	}
   	\foreach \y in {4,12} % spacing
   	{
   	    \draw[postaction={decorate},very thick] ({\x},{\y})--({\x+4},{\y});
   	}	
   	}	
   	\node [above] at (8.9,8) {$v_{2}$};
   	\node [above] at (4.9,8) {$v_{1}$};
   	\node [above] at (6,5) {$p_{1}$};
   	\node [above] at (2,5) {$p_{2}$};
   \end{tikzpicture}
   \label{fig::squareLattice}
   }
    \caption{In (a) we show an example of a ciliated vertex, the edges are numbered counterclockwise starting from the cilium, which we denote by a dotted line.
    \\In (b) we show a piece of the bulk square lattice built as a ribbon graph. In general, the lattice can be considered to have periodic boundary conditions. Note that with our choice of cilia each of the elementary plaquettes has a single cilium pointing inwards and each vertex has two incoming and two outgoing edges.
	}
\end{figure} 

Since not all the readers might be familiar with Hopf algebras we will now start with a brief recap of the main properties and definitions that we will use throughout the article. Explicit actions for the canonical Hopf algebra structure on a group algebra and for the algebra of functions on 
a finite group can be found in \ref{appendix::Quasi_Hopf_Group}. A Hopf algebra $K$ is an associative unital algebra where, together with the product which defines the algebra, we have the following algebra homomorphisms:
\begin{equation}
    \begin{split}
		\Delta: K \rightarrow K\otimes K, \qquad \qquad
		 % \\      &
		\eta: \mathbb{C}\rightarrow K, \qquad\qquad\ \ \epsilon: K\rightarrow \mathbb{C}.
    \end{split}
\end{equation}
The algebra homomorphism $\Delta$ is called coproduct and can be considered as a sort of ``reverse'' of the product map. It is used to split an element into two tensored components\footnote{Most of the reader should have already encountered the coproduct, as it is used to create tensor products of representation. E.g. for determining the representation of angular momentum for two particles, the coproduct is given by $$\Delta(j)=j\otimes 1+1\otimes j\ .$$}. 
%\AAAC{This map allows us to split elements...}
The action of $\Delta$ on an element $h\in K$ is generally denoted through the Sweedler notation:
\begin{equation}
    \Delta(h) = \sum_{(h)} h^{(1)} \otimes h^{(2)}\ ,
\end{equation}
which is a symbolic representation of the fact that a coproduct acting on an element can generally contain a sum of tensor products of different elements of the algebra. The summation is often left implicit and we will follow this convention in the following.
The coproduct is the structure used to define the tensor product of two representations, i.e.
\begin{equation}
	\pi^{1} \otimes \pi^{2}: h \rightarrow ( \pi^{1} \otimes \pi^{2})(\Delta(h)).
\end{equation}
This is of particular importance in toric code models as this structure is used to form multiparticle states. 
One further ingredient is necessary for the definition of a Hopf algebra, that is the antipode, $S: K\rightarrow K$ such that
\begin{equation}
	S(ab) = S(b)S(a)\ .
\end{equation}
This is a generalisation of the map that inverts elements of the group and for semisimple Hopf algebra, we have $S^{2}= \text{id}$.
The map $\eta$ is given by the following:
\begin{equation}
    \eta(c)= c \, 1_{K}\ ,
\end{equation}
where $1_{K}$ indicates the identity element of the algebra and $c$ is an element of the field, which we take to be $\mathbb{C}$. 

The $\epsilon$ map is called counit and can be considered as the equivalent of the unit in the algebra for the coproduct. The coproduct and counit satisfy the following compatibility conditions:
\begin{equation}
    \begin{split}
        &\epsilon(h^{(1)}) h^{(2)}=h^{(1)}\epsilon(h^{(2)})=h,\\
        &S(h^{(1)}) h^{(2)}= h^{(1)} S(h^{(2)})=\epsilon(h) 1\ .
    \end{split}
\end{equation}

Note that when the Hopf algebra is obtained from a group algebra there is a canonical Hopf algebra structure given in the following form
\begin{equation}
    \Delta(h) = h\otimes h \qquad\qquad S(h)=h^{-1}, \qquad \epsilon(h) = 1_{\mathbb{C}} \ ,
\end{equation}
with the multiplication of elements given by the multiplication in the group algebra. 

In order for a Hopf algebra to be quasitriangular, there must be an invertible element of $K\otimes K$ that we will represent as
\begin{equation}
    R = \sum R'\otimes R''\ .
	\label{eq::General_R}
\end{equation}
The notation is simply meant to indicate that $R$ can be made of sums of tensor products of elements in the algebra. 
The $R$ matrix satisfies the following properties
\begin{equation}
    \begin{split}
    &R\,(h^{(1)}\otimes h^{(2)})\, R^{-1}=h^{(2)}\otimes h^{(1)},\qquad \forall\ h\in K,\\
    &(\Delta\otimes id)(R) = R_{13} R_{23},\qquad (id\otimes \Delta)(R) = R_{13} R_{12} 
    \end{split}
\end{equation}
where $R_{13}=R'\otimes 1\otimes R''$, $R_{12}=R'\otimes R''\otimes 1$ and $R_{23}=1\otimes R'\otimes R''$; with the summation left implicit.

% In \ref{appendix::Quasi_Hopf_Group} we provide standard definitions of the Hopf algebra maps in the basis of a finite group.
For a more detailed and formal discussion on this topic, we refer to the many books and articles on the subject, such as \cite{Kassel, Majid_1995}.  
 We will now describe aspects of Hopf algebra gauge theory that are relevant in the present context.

\subsection{Gauge transformations}
\label{sec::GaugeTransformation}
In this section, we will draw a comparison between key concepts of discrete gauge theory and the analogous ones in Hopf algebra gauge theory.  For details on discrete gauge theory, we refer the reader to \cite{Mark_Propitus_Twisted, Bais_Propitius, Wen}. 

We start by recalling that for finite groups, a gauge transformation by an element $h$, acts at a vertex $v$ of the graph on the elements of the group assigned to the edges incident to that vertex. 
This action we denote by $\tilde{A}^{h}_{v}$ and it is defined by the following rules:
\begin{equation}
	\tilde{A}^{h}_{v}(k) = h \cdot k \qquad\qquad 	\tilde{A}^{h}_{v}(k) = k \cdot h^{-1} 
\end{equation}
for incoming and outgoing edges respectively. The tilde is to differentiate this operator from its dual counterpart, which we will introduce in the following.
For Hopf algebra gauge theories, if the local vertex neighbourhood has more than one incident edge, we need to split the element $h$ to act on each of the connected edges, in a way that 
is compatible with the multiplication in the algebra and linear over $\mathbb{C}$. 
This is where the coproduct is used. 
Therefore, by direct analogy, we can define the action of gauge transformations at each edge by 

\begin{equation}
    \label{eq::GtransformRules}
    \tilde{A}^{h}_{v}(k_{i}) = h^{(i)} \, k_{i} \qquad\qquad \tilde{A}^{h}_{v}(k_{i}) = k_{i} \, S(h^{(i)})\qquad i=1,\ldots, n
\end{equation}
for incoming and outgoing edges respectively. 
For clarity, an example of this transformation is displayed in Figure \ref{fig::gaugeTransform}. 
The trivial action, for edges that are not connected to the vertex $v$, is given by
\begin{equation}
    \label{eq::GtransformRulesTrivial}
    \tilde{A}^{h}_{v}(k) = \epsilon(h) k \qquad \forall\ \ k\neq k_1,\ k_2,\ \ldots,\ k_n\ .
\end{equation}
For local vertex neighbourhoods with closed edges (loops), the action of gauge transformation can be written as
\begin{equation}
    \tilde{A}^{h}_{v} (k) = (h^{(1)} k S(h^{(2)}))	\qquad \text{or} \qquad
    \tilde{A}^{h}_{v} (k) = (h^{(2)} k S(h^{(1)}))\ ,	
\end{equation}
depending on whether the cilium is pointing inwards or outward to a counterclockwise loop. 
Similar relations hold for clockwise loops.
\\We can therefore see that both the directedness and the cyclic ordering of the edges around a vertex play a role in Hopf algebra gauge theory and the above discussion should convince the reader of the necessity of the cilium in this framework. 

This concludes our discussion of gauge transformations for Hopf algebra gauge theory. As shown in \cite{Catherine1},
 the properties of the theory take a more appealing and simple algebraic form when working in the dual space $K^{*}$, the algebra of linear functionals over $K$. 

\begin{figure}[t]
\centering
	% {\input{Paper_Pictures/GaugeTransform}
	\begin{tikzpicture}[>=latex,scale=0.7,>=latex,very thick,decoration={
	    markings,
	    mark= at position 0.65 with {\arrow{>}}}]
		\draw[postaction={decorate}] ({0},{-3})--({0},{0});
		\draw[postaction={decorate}] ({0},{0})--({0},{3});
		\draw[fill] ({0},{0}) circle [radius=0.1];
		\draw[postaction={decorate}] ({0},{0})--({-3},{0});
		\draw[postaction={decorate}] ({0},{0})--({3},{0});
		\draw[postaction={decorate}] ({2.4)},-{2.4})--({0},{0});
		\draw[postaction={decorate}] ({2.4},{2.4})--({0},{0});
		\draw[postaction={decorate}] ({0},{0})--({-2.4},{2.4});
		\draw[very thick,dotted] ({-0.75*sqrt(2)},{-0.75*sqrt(2)}) -- ({0},{0});   
	    \node [left] at (0,-2.5) {$h^{(1)}k_{1}$};
	    \node[left] at (2.4,-2.505) {$h^{(2)}k_{2}$};
	    \node[below] at (3,0) {$k_3S(h^{(3)})$};
	    \node[below] at (3,2.4) {$h^{(4)}k_{4}$};
	    \node[right] at (0.5,2.5) {$\ldots$};
	    \node[right] at (-4,3) {$k_{n-1}S(h^{(n-1)})$};
	    \node[above] at (-2.5,0) {$k_{n}S(h^{(n)})$};
	\end{tikzpicture}
    \caption{We display the gauge transformation for a local vertex neighbourhood by an element $h \in K$, the elements $k_{i}$ are assigned to $K^{E}$, the Hopf algebra on the edges. 
	}
\label{fig::gaugeTransform}
\end{figure}
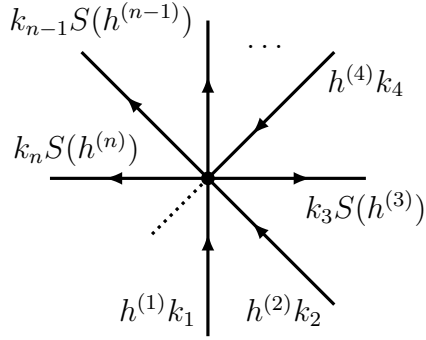 
For this reason, we will now describe the gauge transformations on $K^{*}$, which can be found by dualizing (\ref{eq::GtransformRules}). This dualization is performed using two main properties, that come from the compatibility between the Hopf algebra with its dual
 \begin{equation}
 \langle \alpha, k_1\cdot k_2\rangle=\langle\alpha^{(1)}\otimes\alpha^{(2)}, k_1\otimes k_2\rangle\qquad\langle \alpha_1\cdot\alpha_2, k\rangle=\langle\alpha_1\otimes\alpha_2, k^{(1)}\otimes k^{(2)}\rangle\ .
 \end{equation}
Consider therefore a local vertex neighbourhood with $n$ edges, such as the one shown in Figure \ref{fig::ciliatedVertex}, 
then the gauge transformation by $h \in K^{v}$ on  $\alpha_1\otimes\alpha_2\cdots\otimes\alpha_n \in {K^{*}}^{\otimes n}$, is given by
\begin{equation}
    \label{eq::gaugeTransformation}
  	 A^{h}_v(\alpha_1\otimes\ldots\otimes\alpha_n):= \langle S^{\tau_{1}}
	 (\alpha_1^{(1+\tau_1)})\cdots S^{\tau_{n}}(\alpha_n^{(1+\tau_n)}), 	h\rangle(\alpha_1^{(2-\tau_1)}\otimes\ldots\otimes\alpha_n^{(2-\tau_n)})\ ,
\end{equation}
where $\tau_i=0,\ 1$ if the $i$'th edge is incoming to the vertex or outgoing respectively (we are assuming $S^0=id$). Note that this constitutes a right $K^{V}$ module action on ${K^{*}}^{\otimes n}$, which is denoted $(\alpha_1\otimes\ldots\otimes\alpha_n)\triangleleft^{*}h$ in \cite{Catherine1} and further details can be found therin, in particular see Corollary 3.12. 
It can now be proved that with these definitions we have
\begin{equation}
\label{eq::commVertex}
    A^{h_{1}}_{v}A^{h_{2}}_{v}=A^{h_{2}h_{1}}_{v}\qquad A^{h_{1}}_{v_1}A^{h_{2}}_{v_2}=A^{h_{2}}_{v_2}A^{h_{1}}_{v_1}\qquad v_{1}\neq v_{2},  \, \forall h_{1},h_{2} \in K \ .
\end{equation}
Analogously to discrete gauge theory we say that a function is gauge invariant at a vertex $v$ when
\begin{equation}
     A^{h}_v(\alpha_1\otimes\ldots\otimes\alpha_n)= \epsilon(h)(\alpha_1\otimes\ldots\otimes\alpha_n)\qquad \forall\ h \in K\ . 
\end{equation}
It can be seen that given the Haar integral $l$ of the Hopf algebra $K$, the projector $P^{\text{inv}}_{v}$ into the space of gauge invariant functions at a vertex $v$ is given by the gauge action
\begin{equation}
    \label{eq::invProjector}
	 P_{v}^{\text{inv}} = A_v^l\ .
\end{equation}
For a finite group $G$, the Haar integral is given by 
\begin{equation}
    l = \sum_{g\in G} g
\end{equation}
and by substituting this formula in (\ref{eq::invProjector}) the familiar projector into gauge invariant states of the toric code can be found,\cite{Buerschaper2013,Catherine2}.

\subsection{Braided tensor product and holonomy}
\label{sec::BraidedTensor}

We can now introduce the second fundamental ingredient of Hopf algebra gauge theory, the braided tensor product, which plays an important role in the construction of the plaquette operator.
As in quantum field theory, the holonomy is obtained by considering parallel transport around loops. Consider therefore a plaquette $p$ turning counterclockwise, with all cilia at vertexes pointing outward, except for one, as shown in Figure \ref{fig::holonomy}. The holonomy on $p$, for discrete gauge theory, is then the product of elements assigned along the edges of the path starting from the unique cilium pointing into the plaquette
\begin{equation}
	\label{eq::HopfHol}
   \text{Hol}(p)=(k_{n}k_{n-1}\cdots k_{2}k_{1})^{-1}\ .
\end{equation}
To simplify the discussion we have supposed that the orientation of the plaquette's edges agrees with that of a path turning counterclockwise. 
If the path goes against the directedness of an edge, the element assigned to that edge is first acted upon with the antipode $S$.
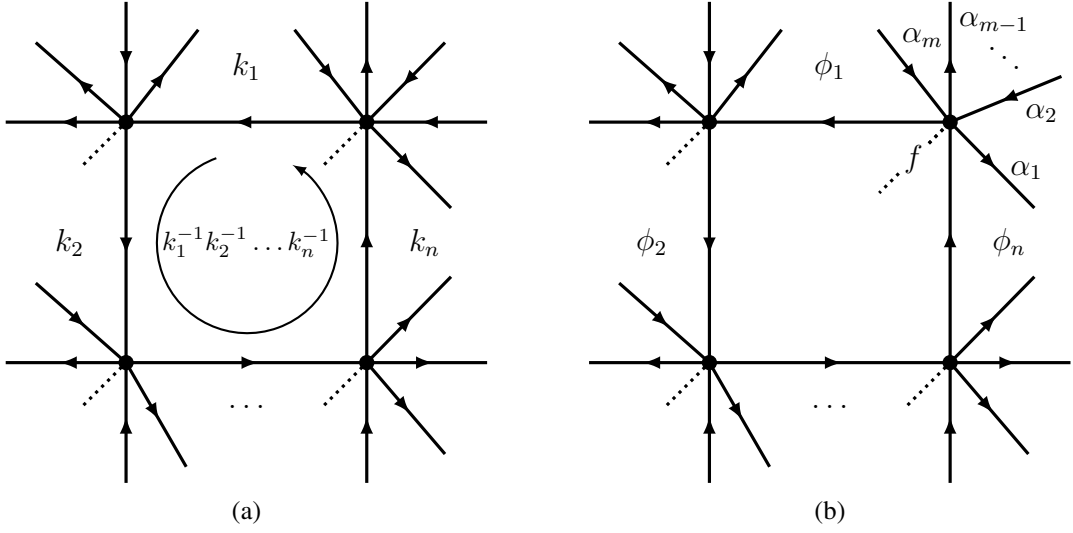
\begin{figure}[h]
\centering
    \begin{center}
        \subfloat[]{
		\begin{tikzpicture}[>=latex,scale=0.8,>=latex,very thick,decoration={
		    markings,
		    mark= at position 0.55 with {\arrow{>}}}]
			\draw[postaction={decorate}] ({0},{-2})--({0},{0}); % Abl vertical down 
			\draw[postaction={decorate}] ({0},{4})--({0},{0}); % plaquette vertical left, lower
			\draw[fill] ({0},{0}) circle [radius=0.1];
			\draw[fill] ({0},{4}) circle [radius=0.1]; 
			\draw[postaction={decorate}] ({0},{0})--({-2},{0}); % Abl horizontal left 
			\draw[postaction={decorate}] ({0},{0})--({1},{-sqrt(3)});
			\draw[postaction={decorate}] ({-1.5},{sqrt(1.75)})--({0},{0});	
			\draw[postaction={decorate}] ({0},{4})--({-1.5},{4+sqrt(1.75)});
			\draw[postaction={decorate}] ({0},{4})--({1.2},{4+sqrt(2.36)});
			\draw[postaction={decorate}] (2.8,{4+sqrt(2.36)})--(4,4);	
			\draw[postaction={decorate}] (5.3,{4+sqrt(1.75)})--(4,4);
			\draw[postaction={decorate}] (4,4)--(5.4,{4-sqrt(2.04)});
		    \draw[postaction={decorate}] (4,0)--(5.4,{sqrt(2.04)});   
		    \draw[postaction={decorate}] (4,0)--(5.3,{-sqrt(2.31)});	
			\draw[postaction={decorate}] ({0},{6})--({0},{4});
			\draw[postaction={decorate}] ({0},{0})--({4},{0});  % plaquette horizontal left
			\draw[postaction={decorate}] ({0},{4})--({-2},{4}); 
			\draw[postaction={decorate}] ({4},{4})--({0},{4}); % plaquette horizontal top left 
			\draw[fill] ({4},{4}) circle [radius=0.1]; 	
			\draw[postaction={decorate}] ({4},{0})--({6},{0}); % Bbr right horizontal
			\draw[fill] ({4},{0}) circle [radius=0.1];
			\draw[postaction={decorate}] ({4},{-2})--({4},{0}); % Bbr vertical down 
			\draw[postaction={decorate}] ({4},{0})--({4},{4});  % plaquette rhs vertical 	
			\draw[postaction={decorate}] ({6},{4})--({4},{4}); %. Atr horizontal right 
			\draw[postaction={decorate}] ({4},{4})--({4},{6}); % Atr vertical up 	
			\draw[very thick,dotted] ({-sqrt(2)/2},{-sqrt(2)/2}) -- ({0},{0});
			\draw[very thick,dotted] ({-sqrt(2)/2},{-sqrt(2)/2+4}) -- ({0},{4});
			\draw[very thick,dotted] ({-sqrt(2)/2+4},{-sqrt(2)/2+4}) -- ({4},{4});
			\draw[very thick,dotted] ({-sqrt(2)/2+4},{-sqrt(2)/2}) -- ({4},{0});
		    \draw[thick, ->] (1.5,3.4) arc (110:420:1.5cm);		
			\node at (2,2) {\footnotesize$k_{1}^{-1}k_{2}^{-1}\ldots k_{n}^{-1}$};
			\node [left] at (-0.5,2) {$k_2$};
			\node [above] at (2,4.5) {$k_1$};
			\node [right] at (4.5,2) {$k_n$};
			\node [below] at (2,-0.5) {$\ldots$};	
		\end{tikzpicture}
		\label{fig::holonomy}
		}
        \hspace{1cm}
        \subfloat[]{
		\begin{tikzpicture}[>=latex,scale=0.8,>=latex,very thick,decoration={
		    markings,
		    mark= at position 0.55 with {\arrow{>}}}]
			\draw[postaction={decorate}] ({0},{-2})--({0},{0}); % Abl vertical down 
			\draw[postaction={decorate}] ({0},{4})--({0},{0}); % plaquette vertical left, lower
			\draw[fill] ({0},{0}) circle [radius=0.1];
			\draw[fill] ({0},{4}) circle [radius=0.1]; 
			\draw[postaction={decorate}] ({0},{0})--({-2},{0}); % Abl horizontal left 
			\draw[postaction={decorate}] ({0},{0})--({1},{-sqrt(3)});
			\draw[postaction={decorate}] ({-1.5},{sqrt(1.75)})--({0},{0});
			\draw[postaction={decorate}] ({0},{4})--({-1.5},{4+sqrt(1.75)});
			\draw[postaction={decorate}] ({0},{4})--({1.2},{4+sqrt(2.36)});
			\draw[postaction={decorate}] (2.8,{4+sqrt(2.36)})--(4,4);
			\draw[postaction={decorate}] (4,4)--(5.4,{4-sqrt(2.04)});
		    \draw[postaction={decorate}] (4,0)--(5.4,{sqrt(2.04)});  
		    \draw[postaction={decorate}] (4,0)--(5.3,{-sqrt(2.31)});
			\draw[postaction={decorate}] ({0},{6})--({0},{4});
			\draw[postaction={decorate}] ({0},{0})--({4},{0});  % plaquette horizontal left
			\draw[postaction={decorate}] ({0},{4})--({-2},{4}); 
			\draw[postaction={decorate}] ({4},{4})--({0},{4}); % plaquette horizontal top left 
			\draw[fill] ({4},{4}) circle [radius=0.1]; 	
			\draw[postaction={decorate}] ({4},{0})--({6},{0}); % Bbr right horizontal
			\draw[fill] ({4},{0}) circle [radius=0.1];
			\draw[postaction={decorate}] ({4},{-2})--({4},{0}); % Bbr vertical down 
			\draw[postaction={decorate}] ({4},{0})--({4},{4});  % plaquette rhs vertical 
			\draw[postaction={decorate}] ({5.85},{4+sqrt(0.5775)})--({4},{4}); %. Atr horizontal right 
			\draw[postaction={decorate}] ({4},{4})--({4},{6}); % Atr vertical up 	
			\draw[very thick,dotted] ({-sqrt(2)/2},{-sqrt(2)/2}) -- ({0},{0});
			\draw[very thick,dotted] ({-sqrt(2)/2},{-sqrt(2)/2+4}) -- ({0},{4});
			\draw[very thick,dotted] ({-sqrt(2)/4+4},{-sqrt(2)/4+4}) -- ({4},{4});
			\node at ({-sqrt(2)/2+4.1},{-sqrt(2)/2+4.1}) {$f$};
			\draw[very thick,dotted] ({-sqrt(2)/4+3.5},{-sqrt(2)/4+3.5}) -- ({-sqrt(2)/1.2+4},{-sqrt(2)/1.2+4});
			\draw[very thick,dotted] ({-sqrt(2)/4+4},{-sqrt(2)/4+4}) -- ({4},{4});
			\draw[very thick,dotted] ({-sqrt(2)/2+4},{-sqrt(2)/2}) -- ({4},{0});
			\node [left] at (-0.5,2) {$\phi_2$};
			\node [above] at (2,4.5) {$\phi_1$};
			\node [right] at (4.5,2) {$\phi_n$};
			\node [right] at (4.8,3.2) {$\alpha_1$};
			\node [right] at (5.05,4.15) {$\alpha_2$};
			\node [below] at (2,-0.5) {$\ldots$};
			\node at (4.9,{3.9+sqrt(1.75)}) {$\ddots$};
			\node [right] at (3.95,5.65) {$\alpha_{m-1}$};
			\node [right] at (3.0,5.4) {$\alpha_{m}$};
		\end{tikzpicture}
		\label{fig::holonomyDual}
		}
    \end{center}
    \caption{In Figure (a) we display the holonomy of elements $k_{1}\dots k_{n}$. In Figure (b) we display the dual holonomy on the space of functions.  
	}
\end{figure} 
When the product of the group elements around the plaquette is equal to the identity, we say that the connection is flat at that plaquette.
\\The plaquette operator associated with an element $f$ of the algebra of functions for plaquette $p$ is then given by
\begin{equation}
	\label{eq::GroupHol}
	B^f_{p}(k_1\otimes k_2\otimes\ldots k_n) = \langle f, (k_{n}k_{n-1}\cdots k_{2}k_{1})^{-1}\rangle(k_1\otimes k_2\otimes\ldots k_n) 
\end{equation}
if in place of $f$ we take the Haar integral $\lambda=\delta_{e}$ of $K^{*}$, this operator then becomes the projector into the space of flat connections on $p$ 
\begin{equation}
    P^{\text{flat}}_{p}=B^{\delta_e}_p\ .
\end{equation}
Note that generally the plaquette operator $B_p^{\delta_a}$ is denoted by $B_p^{e}$ (i.e. see~\cite{Kitaev2003} and \cite{Buerschaper2013}). 
The notion of holonomy can now be readily adapted to our present context.
In particular, for a counterclockwise directed plaquette, we can write
\begin{equation}
    \text{Hol}_p(k_{1}, k_{2}, \ldots, k_{n})=S(k_{n}k_{n-1}\cdots k_{2}k_{1})\ .
\end{equation}
If now we go by analogy with discrete gauge theory, the most ``obvious'' extension of the plaquette operator would be
\begin{equation}
    \tilde{B}^f_{p}(k_1\otimes k_2\otimes\ldots k_n) = \langle f, S(k_{n}^{(1)}k_{n-1}^{(1)}\cdots k_{2}^{(1)}k_{1}^{(1)})\rangle(k_1^{(2)}\otimes k_2^{(2)}\otimes\ldots k_n^{(2)})\
\end{equation}
or similarly, in the dual space
\begin{equation}
    \label{eq::originalPlaquette}
    \begin{split}
        B^{f}_{p}(\phi_1\otimes \phi_2\otimes\ldots \phi_n)&=(S(f^{(1)})\phi_{1}\otimes S(f^{(2)})\phi_{2}\otimes\ldots\otimes S(f^{(n)})\phi_{n})\ .
    \end{split}
\end{equation}
This can also be written as 
\begin{equation}
    B^{f}_{p}(\phi_1\otimes \phi_2\otimes\ldots \phi_n)=\text{Hol}^*_p(f)\cdot(\phi_{1}\otimes\phi_{2}\otimes\ldots\otimes\phi_{n})\ ,
    \label{eq::Holdot}
\end{equation}
where $\text{Hol}^*_p(f)$ is the dual holonomy:
\begin{equation}
    \text{Hol}^*_p(f) = (S(f^{(1)})\otimes S(f^{(2)})\otimes\ldots\otimes S(f^{(n)}))
\end{equation}
and $\cdot$ is the canonical product structure on the tensored elements\footnote{This algebra structure is simply given by
\begin{equation*}
    (a\otimes b)\cdot(c\otimes d)= (a\cdot c \otimes b\cdot d)
\end{equation*}
}. Similarly to what happened before, for plaquettes whose edge orientations do not agree with a path turning counterclockwise, the antipode is employed (which would effectively remove the above antipode as $S^2=id$ for semisimple Hopf algebras).

It is now evident that the plaquette operator is closely related to the definition of tensor product algebra. Note however that this product is not necessarily compatible with the action of gauge transformations \cite{Catherine1, Braided_Hisenberg_Group}. Consider a vertex with two incoming edges. Then if we consider $\alpha \in K^{*}$ is assigned to one edge and $\beta\in K^{*}$ to the other.
The action of the vertex operator, for a general finite-dimensional semisimple Hopf algebra, is not compatible with the multiplication of functions, i.e.
\begin{equation}
    A^{h}_{v}(\alpha\otimes \beta)= A^{h}_{v}((\alpha\otimes 1)\cdot(1\otimes \beta))\neq (A^{h^{(2)}}_{v}(\alpha\otimes 1))\cdot(A^{h^{(1)}}_{v}(1\otimes \beta))\ .
\end{equation}
The disequality, generally, holds with an equal sign only when we are working with a group algebra or, more broadly speaking, with a cocommutative Hopf algebra. 
This fact poses a significant issue for the consistency of the theory, as it means that holonomies of functions on multiple edges will be changed in ``unexpected'' ways by gauge transformations. In particular, the set of functions with trivial holonomies will not be invariant under gauge transformations, so it won't be possible to consider the space of gauge invariant functions on flat connections, which is the topologically protected space. This is why the definition of the plaquette operator, or actually, more fundamentally, the product of the functions, needs to be modified. 
\\When the Hopf algebra is quasitriangular and therefore admits an $R$-matrix, there is a natural deformation of the tensor product that allows one to recover compatibility of the tensor product algebra with gauge transformations. 
This product takes the name of \textit{braided tensor product} \cite{Braided_Groups_Majid,Catherine1, BraidedMajid, MajidQuantumPrimer,Braided_Hisenberg_Group}.  
Given $\alpha,\ \beta,\ \alpha',\ \beta'$ in the algebra of functions and the $R$ matrix of a quasitriangular Hopf algebra around a vertex $v$ with two edges, this product can be written as
\begin{equation}
	\label{eq::BraidedTensorProducModuleAlgebra}
	(\alpha \otimes \beta) * (\alpha' \otimes \beta') = (\alpha \,  \cdot A^{R'}_{v}(\alpha')) \otimes (A^{R''}_{v}(\beta)  \cdot \beta')\ 
\end{equation}
where $A^{R'}_{v}$ ($A^{R''}_{v}$) denotes the vertex operator acting with the first (second) tensor entry of the $R$ matrix (see (\ref{eq::General_R})) and we are considering the ordering of the tensored elements to agree with the ordering imposed by the cilia. 
More explicitly, using (\ref{eq::gaugeTransformation}), this formula can be written as
\begin{equation}
    (\alpha \otimes \beta) * (\alpha' \otimes \beta')= \langle {\alpha'}^{(1+\tau_1)}\otimes\beta^{(1+\tau_2)}, (S^{\tau_1}\otimes S^{\tau_2})(R)\rangle(\alpha\cdot{\alpha'}^{(2-\tau_1)}\otimes \beta^{(2-\tau_1)}\cdot\beta')\ ,
\end{equation}
where the above result can be regarded as a consequence of the commutation between $\beta$ and $\alpha'$, which are sitting on edges of decreasing order (with respect to the cilia ordering).

In general, for a local vertex neighbourhood with more than two edges embedded into the lattice, the structure of the product depends not only on the relative position of the elements being multiplied but also on the edge orientations. This is to make sure that the braided tensor structure stays consistent across the lattice \cite{Catherine1}\footnote{Essentially due to the fact that an edge that is incoming for a given vertex is outgoing for a neighbouring one.}. 
Given two elements $\alpha$ and $\beta$, with $i,\ j$ being their edges position as defined by the cilia, the algebra structure on a bivalent
 local vertex neighbourhood is given by
\begin{equation}
\begin{aligned}
	& (\alpha)_{i} * (\beta)_{i} = \langle \beta^{(2)} \otimes \alpha^{(2)}, R \rangle ( \alpha^{(1)}\beta^{(1)})_{i}
	 && \tau_i = 1 \\
	& (\alpha)_{i} *  (\beta)_{i} = (\alpha\beta)_{i} 
	&&  \tau_i = 0 \\
	& (\alpha)_{i} * (\beta)_{j}  = (\alpha \otimes \beta)_{ij} 
	&& i < j \\
	& (\alpha)_{i} *  (\beta)_{j}  =  \langle \beta^{(1+\tau_j)} \otimes \alpha^{(1+\tau_i)},(S^{\tau_i}\otimes S^{\tau_j})(R) \rangle ( \alpha^{(2-\tau_i)} \otimes \beta^{(2-\tau_j)})_{ij}
	&& i > j\ ,
\end{aligned}
\label{eq::BraidedTensProduct}
\end{equation}
as before $\tau_i=0,\ 1$ if the edge is incoming into the vertex or outgoing respectively, and the subscripts indicate which element is associated with which edge.
\\This local algebraic structure can be extended to the full ribbon graph by stitching together the different local vertex neighbourhood algebras. The details on this procedure can be found in \cite{Catherine1} and we provide some description of it in \ref{appendix::AlgProofs}.

From here we can now write down the plaquette operator for Hopf Algebra gauge theory
\begin{equation}
    \mathcal{B}^{f}_{p}(\phi_1\otimes \phi_2\otimes\ldots \phi_n)=\text{Hol}^{*}(p)*(\phi_1\otimes \phi_2\otimes\ldots \phi_n)\ .
\end{equation}
In \ref{appendix::AlgProofs} we show that this operator, when acting on a plaquette such as the one given in Figure \ref{fig::holonomyDual}, can be written as 
\begin{equation}
\label{eq::plaquetteGeneralFormula}
    \begin{split}
        & \mathcal{B}^{f}_{p}(\phi_{1},\ \phi_{2},\ \ldots,\ \phi_{n-1},\ \phi_{n},\ \alpha_{1},\ \alpha_{2}, \ \ldots,\ \alpha_{m}) =\\
        & \, \ \ \langle\phi_{n}^{(1)}\otimes f^{(1)},R\rangle\left[ \, \prod_{j=1}^{m}\langle\alpha_{j}^{(1+\tau_j)}\otimes f^{(j+1)},(S^{\tau_j}\otimes id)(R)\rangle\right]\langle\phi_{1}^{(2)}\otimes f^{(m+2)}, R^{-1}\rangle\\
        &\ \ B^{f^{(m+3)}}_{p}\left(\phi_{1}^{(1)},\ \phi_{2},\ \ldots,\ \phi_{n-1},\ \phi_{n}^{(2)},\ \alpha_{1}^{(2-\tau_{1})}, \ \ldots,\ \alpha_{m}^{(2-\tau_{m})}\right)\ ,
    \end{split}
\end{equation}
with $B^f_p$ as in (\ref{eq::originalPlaquette}). This formula may seem complicated, but it is simply saying that the most ``obvious'' plaquette operator and the one that is compatible with gauge transformation differ only for something that happens at the starting vertex of the plaquette. A gauge transformation acting at any combination of vertices of the lattice (i.e. some product of gauge transformations on individual vertices) cannot change the holonomy, unless it acts at the base point of the plaquette and then just in a canonical way for a gauge transformation acting on a single edge (a loop), with the group element on that edge being the monodromy of the plaquette. Note that equation (\ref{eq::plaquetteGeneralFormula}) looks like as an element $f$, placed at an extra incoming edge at the cilia, is multiplying the functions on the lattice according to equation (\ref{eq::BraidedTensProduct}). 

With this definition, it is now possible to make sense of the concept of flat connection for Hopf algebra gauge theory. In particular, given the Haar integral $\lambda$, of the algebra of functions, the following defines a projector (see \cite{Catherine1}) into the space of flat connections:
\begin{equation}
    P^{\text{flat}}_{p}=\mathcal{B}^{\lambda}_p\ .
\end{equation}
With this definition plaquette operators have non-trivial commutation relations, and
 in particular
\begin{equation}
    \label{eq::plaquetteCommutations}
    \begin{split}
    &\mathcal{B}^{f}_p\mathcal{B}^{g}_p=\langle g^{(1)}\otimes f^{(2)}, R\rangle\langle g^{(3)}\otimes f^{(1)}, R^{-1}\rangle \, \mathcal{B}^{f^{(3)}g^{(2)}}_{p}\\
    &\mathcal{B}^{f}_p \, \mathcal{B}^{g}_q=\mathcal{B}^{g}_q \, \mathcal{B}^{f}_p\qquad p\neq q\ .
    \end{split}
\end{equation}
Plaquette operators and vertex operators commute at all vertices which do not contain the singular cilium pointing into the plaquette.
However at the vertex containing the plaquette's cilia we have
\begin{equation}
    \label{eq::commVertexPlaquette}
        A_{v}^{h}\mathcal{B}_p^{f}=\langle f^{(3)}, S(h^{(1)})\rangle\langle f^{(1)}, h^{(2)}\rangle  \mathcal{B}_{p}^{f(2)} A_{v}^{h^{(3)}}\ .
\end{equation}
The proofs for the above relations can be found in \ref{appendix::AlgProofs}. 
These properties can be directly compared with similar ones obtained using $B_{p}^{f}$ ~\cite{Buerschaper2013,Catherine2}, 
\begin{equation}
    \label{eq::commBA}
    B^{f}_{p}B^{g}_p=B_{p}^{gf}\qquad A_{v}^{h}B_p^{f}=\langle f^{(3)}, S(h^{(1)})\rangle\langle f^{(1)}, h^{(3)}\rangle  B_{p}^{f(2)} A_{v}^{h^{(2)}}\ ,
\end{equation}
which is saying that plaquette and vertex operators of traditional Kitaev models form a representation of the Drinfeld double of the gauge symmetry~\cite{Kitaev2003}. Note that for cocommutative Hopf algebras, Equation (\ref{eq::commVertexPlaquette}) agrees with Equation (\ref{eq::commBA}).
In particular, this also means that the particle content of the two theories is the same, as we will see in the following.

We conclude this section by mentioning that both in the original Kitaev model and in the more general Hopf algebra gauge theory, projectors onto flat connections and onto the space of gauge invariant functions commute with each other:
\begin{equation}
\begin{split}
     \qquad P^{\text{flat}}_{p_{1}} \, P^{\text{flat}}_{p_{2}} = P^{\text{flat}}_{p_{2}} \, P^{\text{flat}}_{p_{1}} &\qquad P_{v_{1}}^{\text{inv}}\, P_{v_{2}}^{\text{inv}} = P_{v_{2}}^{\text{inv}} \, P_{v_{1}}^{\text{inv}}\qquad \forall\  p_{1},p_{2},v_{1},v_{2},\\ &\hspace{-30pt} P^{\text{flat}}_{p}\, P_{v}^{\text{inv}}=P_{v}^{\text{inv}}\, P^{\text{flat}}_{p}\qquad \forall\ p, v\ .
\end{split}
\end{equation}
These results can be proved from the properties of the Haar integrals and the relations given above. In particular, this means that we have a good set of commuting projectors to define a stabilizer code (e.g. on a torus).

\section{Lattice Hopf algebra gauge theory for $\mathbb{CZ}_N$}
\label{sec::LatticeHopf}
In this section, we will construct Hamiltonians for Hopf algebra gauge theories on a 
square lattice like the one shown in Figure \ref{fig::squareLattice}.
We will follow Kitaev's idea and encode the space of gauge invariant functions and flat connection into the ground state of some Hamiltonian, \cite{Kitaev2003}.

Since all the local projectors of the Hopf algebra gauge theory commute with each other, we can define such a Hamiltonian, which is a sum of mutually commuting projection operators, by
\begin{equation}
    \label{eq::projectorsHamiltonia}
     H = -\sum_{v}P_{v}^{\text{inv}}-\sum_{p}P_{p}^{\text{flat}}\ .
\end{equation}
The space of gauge invariant and flat connections is then identified with the space of ground states of this Hamiltonian.
Gauge invariance is implemented as an energy penalty for violations. 
Considering states that break gauge invariance and flatness give rise to anyons \cite{Kitaev2003, Buerschaper2013,KitaevModels_FiniteGroup}.

In the present context, we will consider the theory constructed over the quasitriangular Hopf algebra $\mathbb{C}\mathbb{Z}_N$, the group algebra constructed over $\mathbb{Z}_{N}$, the cyclic group of order $N$. 
We denote the basis of $K = \mathbb{CZ}_{N}$ by $a^{p}$, where $a$ is the generator of $\mathbb{Z}_{N}$ and a basis of characters for $K^{*}$ by $\chi_{i}$, with $i = 0,\dots N-1$.
The action of the characters on $\mathbb{Z}_{N}$ is defined as
\begin{equation}
    \label{eq::omegaDef}
   \chi_{i}(a^{p})=\omega^{ip} \qquad\omega=e^{\frac{2\pi i}{N}}\ .
\end{equation}

As we discussed in Section \ref{sec::Background}, the novelty of this approach comes with the introduction of the braided tensor product with a non-trivial $R$ matrix in $K$. 
We are therefore interested to see how a non-trivial $R$-matrix changes the gauge theory and Kitaev model in the context where the Hopf algebra is constructed over a group algebra. In this case, the gauge theory with the trivial $R$-matrix would be exactly the conventional lattice gauge theory and the Kitaev model would be the conventional toric code.
%The easiest set-up, in this sense, comes from using a quasitriangular Hopf algebra constructed from a group algebra and that is the main reason why we concentrate on $\mathbb{CZ}_N$. 
In particular, for  $\mathbb{CZ}_N$, there are up to $N$ choices of quasitriangular structure. 
We label these by $k=0,1,\ldots,N-1$ in the $R$ matrix, which is defined by 
% Each group with a normal Abelian subgroup admits a non-trivial $R$ matrix \cite{DavydovQuasiStruct}.
\begin{equation}
    R=\begin{cases}
    \frac{1}{N}\sum_{p,q=0}^{N-1}\omega^{-kpq} \hspace{2pt} a^{p}\otimes a^{q} & \text{gcd}(k,N)=1\\
    \ 1\otimes 1 & k=0
	% \qquad k=0,1,\ldots,N-1\ ,
	\end{cases}
	\label{eq::ZnRmatrixAppendix}
\end{equation}
with $\omega$ an $N$'th root of unity as in (\ref{eq::omegaDef}). We can now start building representations of vertex and plaquette operators.

Building a matrix representation for the Hamiltonian (\ref{eq::projectorsHamiltonia}) is equivalent to considering the regular representation
% \footnote{There might be other representations beyond the regular representation. However, given the above description, a representation in terms of the algebra of functionals $(K^*)^E$ results to be simplest, appropriate Hilbert space.}
 of vertex operators and plaquette operator over $(\mathbb{CZ}_N)^{*}$ and we identify
\begin{equation}
\label{eq::basisDefinition}
    \chi_i\equiv |i\rangle\ .
\end{equation}
Using this basis and using (\ref{eq::gaugeTransformation}, \ref{eq::plaquetteGeneralFormula}), it is straightforward to find the regular representation of vertex operators and plaquette operators. In particular, these can be given in terms of the matrices 
\begin{equation}
    \sigma=\begin{pmatrix}1 & 0 & 0 &\cdots & 0\\
    0 & \omega & 0&\cdots& 0\\
    \vdots& &\ddots & &\vdots\\
    0&0&\cdots&\omega^{N-2}&0\\
    0&0&\cdots&0&\omega^{N-1}
    \end{pmatrix}\ , 
	\qquad
	\tau = \begin{pmatrix}
	0&1&0&\cdots&0\\
	0&0&1&\cdots&0\\
	\vdots&&&\ddots 
	&\vdots\\
	0&0&\cdots&0&1\\
	1&0&\cdots&0&0
	\end{pmatrix}\ ,
\label{eq::Tau}
\end{equation}
which are a generalization of the Pauli matrices and reduce to the familiar $\sigma^z,\ \sigma^x$ when $N=2$. The commutation relation of the $\tau$ and $\sigma$ matrices is the following:
\begin{align}
	&\tau^{m} \sigma^{l} = \omega^{ml} \sigma^{l}\tau^{m} \nonumber\\
	& \sigma^{l} \tau^{m} = \omega^{-ml} \tau^{m} \sigma^{l}\ .
\end{align}

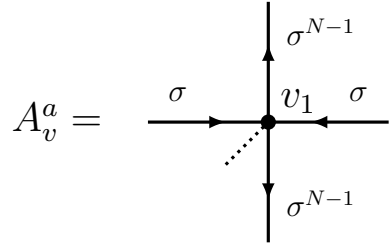
\begin{figure}[t]
	 \centering
	 \begin{tikzpicture}[>=latex,scale=0.8,>=latex,very thick,decoration={
	     markings,
	     mark= at position 0.65 with {\arrow{>}}}]
	 	\draw[postaction={decorate}] ({0},{0})--({0},{-2}); % A downwards vertical
	 	\draw[postaction={decorate}] ({0},{0})--({0},{2});
	 	\draw[fill] ({0},{0}) circle [radius=0.1];
	 	\draw[postaction={decorate}] ({-2},{0})--({0},{0});
	 	\draw[postaction={decorate}] ({2},{0})--({0},{0});	
	 	\draw[very thick,dotted] ({-sqrt(2)/2},{-sqrt(2)/2}) -- ({0},{0});
	     \node [right] at (0.1,-1.35) {$\sigma^{N-1}$};
	     \node [above] at (1.5,0.15) {$\sigma$};
	     \node [above] at (-1.5,0.15) {$\sigma$};
	     \node [right] at (0.1,1.5) {$\sigma^{N-1}$};   
	 	\node [scale = 1.3] at ({-3.51},{0}) {$A^{a}_{v}=$};
	 	\node [scale = 1.3] at ({0.5},{0.35}) {$v_{1}$};	
	 \end{tikzpicture}
	 \caption{Here we display the action of the vertex operator associated with gauge 
	 transformation by $ a\in\mathbb{C} \mathbb{Z}_{N}$ for the local vertex neighbourhood $v_{1}$. 
	 The matrix $\sigma$ is assigned to incoming edges and $\sigma^{-1} = \sigma^{N-1}$ is assigned to 	an outgoing edge.  
	 This is the usual vertex operator for conventional $\mathbb{Z}_{N}$ toric code. 
	 }
    \label{fig::vertexZN}
\end{figure}

Let's therefore start by considering the representation of a vertex operator that acts on the vertex $v_1$ depicted in Figure \ref{fig::vertexZN}, under (\ref{eq::gaugeTransformation}). It can be proved that this action, on the basis given in (\ref{eq::basisDefinition}), corresponds to the one provided in the same Figure. Using (\ref{eq::gaugeTransformation}) for a general vertex $v$ we get
\begin{equation}
    \label{eq::propertyVertexZN}
    A_{v}^{a^{i}}=(A_{v}^{a})^{i},\qquad (A_{v}^{a})^{N}=1, \qquad(A_{v}^{a})^{N-1}=(A_{v}^{a})^{\dagger}, 
\end{equation}
and reversing the edge orientations corresponds to taking the Hermitian adjoint of the above operator. This means that when considering the lattice shown in Figure \ref{fig::squareLattice}, the direction of the arrows around the vertex is unimportant (note that there are two types of vertices), as the form of the projector into gauge invariant space is independent of edges orientations at that vertex. 
Given $l= \frac{1}{N}\sum_{i=0}^{N-1} a^{i}$ we have 
\begin{equation}
    \label{eq::gaugeProjectorZN}
    P_{v}^{\text{inv}} = \frac{1}{N}\left(1+(A_{v}^{a})+(A_{v}^{a})^{2}+\ldots +(A_{v}^{a})^{N-1}\right)\ .
\end{equation}

We can now turn our attention to the plaquette operator. The matrix form of the plaquette operator is given by considering the representation of (\ref{eq::plaquetteGeneralFormula}) on the basis (\ref{eq::basisDefinition}). Consider therefore the plaquette shown in Figure \ref{fig::plaquetteSubdivision}. It is straightforward to see that the
   representation of the plaquette operator associated with $\chi_{1}$, is then the one depicted in the same Figure, where we display the action only for those edges on which the operator acts non-trivially.
\begin{figure}[t]
\begin{center}
    \begin{tikzpicture}[>=latex,scale=0.7,>=latex,very thick,decoration={
        markings,
        mark= at position 0.55 with {\arrow{>}}}]
    	\draw[postaction={decorate}] ({0},{-2})--({0},{0}); % Abl vertical down 
    	\draw[postaction={decorate}] ({0},{4})--({0},{0}); % plaquette vertical left, lower
    	\draw[fill] ({0},{0}) circle [radius=0.1];
    	\draw[fill] ({0},{4}) circle [radius=0.1]; 
    	\draw[postaction={decorate}] ({0},{0})--({-2},{0}); % Abl horizontal left 
    	\draw[postaction={decorate}] ({0},{4})--({0},{6});
    	\draw[postaction={decorate}] ({0},{0})--({4},{0});  % plaquette horizontal left
    	\draw[postaction={decorate}] ({-2},{4})--({0},{4}); 
    	\draw[postaction={decorate}] ({4},{4})--({0},{4}); % plaquette horizontal top left 
    	\draw[fill] ({4},{4}) circle [radius=0.1]; 	
    	\draw[postaction={decorate}] ({6},{0})--({4},{0}); % Bbr right horizontal
    	\draw[fill] ({4},{0}) circle [radius=0.1];
    	\draw[postaction={decorate}] ({4},{0})--({4},{-2}); % Bbr vertical down 
    	\draw[postaction={decorate}] ({4},{0})-- ({4},{4});  % plaquette rhs vertical 	
    	\draw[postaction={decorate}] ({4},{4})--({6},{4}); %. Atr horizontal right 
    	\draw[postaction={decorate}] ({4},{6})--({4},{4}); % Atr vertical up 
    	\draw[very thick,dotted] ({-sqrt(2)/2},{-sqrt(2)/2}) -- ({0},{0});
    	\draw[very thick,dotted] ({-sqrt(2)/2},{-sqrt(2)/2+4}) -- ({0},{4});
    	\draw[very thick,dotted] ({-sqrt(2)/2+4},{-sqrt(2)/2+4}) -- ({4},{4});
    	\draw[very thick,dotted] ({-sqrt(2)/2+4},{-sqrt(2)/2}) -- ({4},{0});
     	\node  [scale =1] at (6.0,4.6) {$\sigma^{-\widehat{k}}$}; % top right, vertical 
     	\node  [scale =1] at (4.7,5.6) {$ \sigma^{\widehat{k}}$};
    	\node  [scale =1] at (2.,-0.7) {$\tau^{N-1}$}; % bottom horizontal 
    	\node  [scale =1] at (-1.00,1.65) {$\tau^{N-1}$};	
    	\node  [scale =1.2] at (-3.50,2.0) {$\mathcal{B}^{\chi_1}_{p}=$};
    	\node  [scale =1] at (5.3,1.8) {$\tau^{N-1} \sigma^{\widehat{k}}$}; % right vertical	
    	\node  [scale = 0.9] at (2.15,4.8) {$\tau^{N-1} \sigma^{-\widehat{k}}$};	
    	\node [scale=1] at (4.6,4.6) {$v_{2}$};
    	\node  [scale =1] at (2, 2) {$p_1$};	
    \end{tikzpicture}
\end{center}
\caption{Here we display the decomposition of $\text{Hol}^{*}_{p_{1}}(\chi_{1})$ 
	 for the plaquette $p_{1}$. We display the resulting matrices from the regular representation of the braided tensor product. The important observation is the presence of $\sigma$ matrices at the local vertex neighbourhood with the cilia pointing inward to $p_{1}$, the power of $k$ is related to the choice of R matrix in the braided tensor product.
We can see when $k=0$, then the plaquette operator reduces to the familiar one from the Kitaev model.
}
    \label{fig::plaquetteSubdivision}
\end{figure}
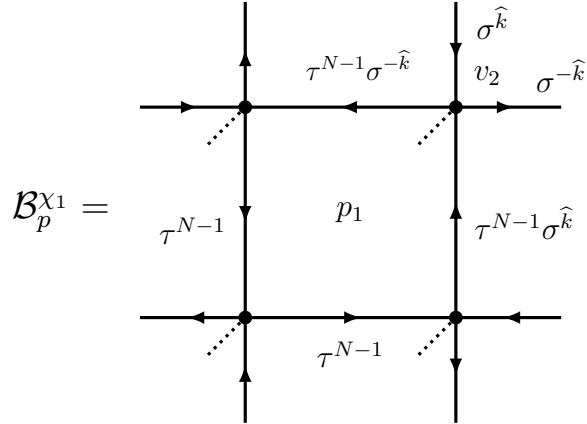

Note that given the particular choice of $R$ matrix that we made, our plaquette operator differs from the plaquette operator of conventional Kitaev models by an extra vertex operator attached to the vertex with the cilium. 
More specifically we have
\begin{equation}
   \mathcal{B}^{\chi_1}_{p_{1}}=B^{\chi_1}_{p_{1}}A^{a^{\widehat{k}}}_{v_{2}} 
   \label{eq::definitionBAtilde}
\end{equation}
where $B_{p}^{\chi_1}$ is the plaquette operator arising from Kitaev models for $\mathbb{C}\mathbb{Z}_{N}$ \cite{pachos_2012,KitaevModels_FiniteGroup}, and $\widehat{k}$ is connected to the $k$ appearing in the $R$ matrix (\ref{eq::ZnRmatrixAppendix}) in the following way
\begin{equation}
    \widehat{k}=\begin{cases}
        \frac{1}{k}\ \text{mod}\ N & \text{gcd}(k,N)=1\\
        0& \text{otherwise}\ .
    \end{cases}
	\label{eq::khat}
\end{equation}
Similarly to the vertex operators, it can be proved that the plaquette operators satisfy the following relations
\begin{equation}
    \label{eq::propertyPlaquette}
    \mathcal{B}^{\chi_{j}}=(\mathcal{B}^{\chi_1})^{j}\qquad (\mathcal{B}^{\chi_1})^{N}=1 \qquad(\mathcal{B}^{\chi_1})^{N-1}=(\mathcal{B}^{\chi_1})^{\dagger}\ .
\end{equation}
In light of these properties, the form of the projector onto the space of flat connections at a general plaquette $p$, does not depend on the orientation of that particular plaquette. 
Given $\lambda=\frac{1}{N}\sum_{i=0}^{N-1}\chi_i$, the projector onto flat connections is
\begin{equation}
    P^{\text{flat}}_{p}=\frac{1}{N}\left(1 + (\mathcal{B}^{\chi_{1}})_p + (\mathcal{B}^{\chi_{1}}_{p})^2 + \ldots + (\mathcal{B}^{\chi_{1}}_{p})^{N-1}\right)
\end{equation}
which is similar to (\ref{eq::gaugeProjectorZN}). 
Bear in mind, however, that the introduction of the $R$ matrix is breaking the electro-magnetic duality given by the interchange of the vertex and plaquette operators \cite{Kitaev2003,EMdualityLattice}. In particular, the vertex operator acts on a single vertex neigborhood in the direct lattice, while the plaquette operator does not act on a single vertex neighborhood of the dual lattice.

Given the above structure for the plaquette and vertex operators, it is not difficult to see that given any two such operators, they commute
\begin{equation}
    \left[\mathcal{B}^{\chi_1}_{p},A^{a}_{v}\right]=0\qquad \forall\ p,\ v\ .
\end{equation}

We are now in the position to prove that the ground state of the theory coincides with the original Kitaev quantum double model. Note in fact that since all the vertex and plaquette operators entering the Hamiltonian commute with each other, the ground state of the model $|GS\rangle$ is such that
\begin{align}
	A^{a}_{v}| GS \rangle &= | GS \rangle \label{eq::ConditionVertexGroundState}\\
	\mathcal{B}^{\chi_1}_{p}|GS  \rangle = |GS \rangle & \, \Rightarrow \, B^{\chi_1}_{p_{1}}|GS \rangle = | GS \rangle\ ,
	\label{eq::ConditionPlaquetteGroundState}
\end{align}
which are the exact same equations that one would obtain in the original Kitaev model. This implies that the ground states of the two theories coincide. 
Observe that this property should be true for any two models based on the same Hopf algebra, but differing in the choice of R-matrix on that Hopf algebra,
as can be deduced more generally from \cite{Catherine1}. 

This concludes the treatment of the ground state physics of Hopf algebra gauge theory based on $\mathbb{CZ}_N$. In the next section, we will consider the states with excitations.
For completeness, before undertaking this task, we will describe in some detail the braided tensor product for a local vertex neighbourhood with $\mathbb{CZ}_N$, as there are some interesting aspects to it.

Consider the braided tensor product on $\bm{\chi_{l}}$, $\bm{\chi_{r}}$, two elements on a local vertex neighbourhood similar to the one shown in Figure \ref{fig::gaugeTransform}, defined as
\begin{equation}
	\bm{\chi_{l}}=(\chi_{l_1}\otimes\chi_{l_2}\otimes\ldots\otimes\chi_{l_n})
	\qquad
	\bm{\chi_{r}}=(\chi_{r_1}\otimes\chi_{r_2}\otimes\ldots\otimes\chi_{r_n})\ .
\end{equation}
Using (\ref{eq::BraidedTensProduct}) it can be seen that the braided tensor product for these two general elements, is given by
\begin{equation}
\label{eq::braidedZN}
\begin{split}
   	\bm{\chi_{l}}*\bm{\chi_{r}}
	& =  \omega^{-\sum_{i=1}^{n}\tau_{i}\widehat{k}l_{i}r_{i}-\sum_{i<j}^{n}(-1)^{\tau_i+\tau_j}\widehat{k}l_{i}r_{j}}
	(\chi_{l_1+r_1} \otimes \chi_{l_2+r_2} 
	 \otimes \ldots \otimes \chi_{l_n+r_n})\ .
\end{split}
\end{equation}
To be concrete example, consider vertex $v_1$ depicted in Figure \ref{fig::vertexZN}, and using the following basis
\begin{equation}
       \begin{split}
           &(\chi_{1})_1=(\chi_{1}\otimes \chi_0\otimes \chi_0\otimes \chi_0)\qquad(\chi_{1})_2=(\chi_0\otimes \chi_{1}\otimes \chi_0\otimes \chi_0)\\
           &(\chi_{1})_3=(\chi_0\otimes \chi_0\otimes \chi_{1}\otimes \chi_0)\qquad(\chi_{1})_4=(\chi_0\otimes \chi_0\otimes \chi_0\otimes \chi_{1})\ .
       \end{split}
   \end{equation}
The product in (\ref{eq::braidedZN}) can then be summarized through the following table, where the element at row $i$ and column $j$ corresponds to $(\chi_{1})_i*(\chi_{1})_j$
\begin{center}
   \begin{tabular}{c|c|c|c|c}
        *& $(\chi_1)_1$ & $(\chi_1)_2$ & $(\chi_1)_3$ & $(\chi_1)_4$\\
        \hline
       $(\chi_1)_1$ &$\omega^{-k}(\chi_2)_1$&$(\chi_1)_1\cdot(\chi_1)_2$ &$(\chi_1)_1\cdot(\chi_1)_3$&$(\chi_1)_1\cdot(\chi_1)_4$\\
       $(\chi_1)_2$ &$\omega^{k}(\chi_1)_1\cdot(\chi_1)_2$&$(\chi_2)_2$&$(\chi_1)_2\cdot(\chi_1)_3$&$(\chi_1)_2\cdot(\chi_1)_4$\\
       $(\chi_1)_3$ &$\omega^{-k}(\chi_1)_1\cdot(\chi_1)_3$&$\omega^{k}(\chi_1)_2\cdot(\chi_1)_3$&$\omega^{-k}(\chi_2)_3$&$(\chi_1)_3\cdot(\chi_1)_4$\\
       $(\chi_1)_4$ &$\omega^{k}(\chi_1)_1\cdot(\chi_1)_4$&$\omega^{-k}(\chi_1)_2\cdot(\chi_1)_4$&$\omega^{k}(\chi_1)_3\cdot(\chi_1)_4$&$(\chi_2)_4$\\
   \end{tabular}
\end{center}
As before, by $\cdot$ we are representing the canonical algebra structure on the tensor product. The products of all the other elements of the algebra can be derived from this table and we can see that they are in general non-commutative. 
In discrete gauge theory, the product table on the same basis of functions would be similar, but without phases, and would therefore be isomorphic to the group algebra $(\mathbb{C}\mathbb{Z}_N)^{\otimes 4}$, in accordance with Pontryagin duality. 

The presence of these phases now leads to a natural question, that is, are they gauge invariant or can they be gauged away by some isomorphism? This question can be addressed by studying the second cohomology group of this $\mathbb{CZ}_{N}$-module, as non-trivial projective representations correspond to non-trivial cocycles \cite{CharacterTheoryProjectiveRepresentations,ProjectiveGroupAlgebra}. 
It can be proved that the phases appearing in the table are 2-cocycles, which are not 2-coboundary. So, the local algebra is a twisted group algebra or, likewise, a non-trivial projective representation of $(\mathbb{CZ}_{N})^{\otimes 4}$.

So to summarise, despite the fact that we chose a cocommutative Hopf algebra, we find that the phases introduced by the $R$ matrix render the algebra of functions non-commutative.

\subsection{Excitations in Hopf algebra lattice gauge theory}
\label{sec::Excitations}
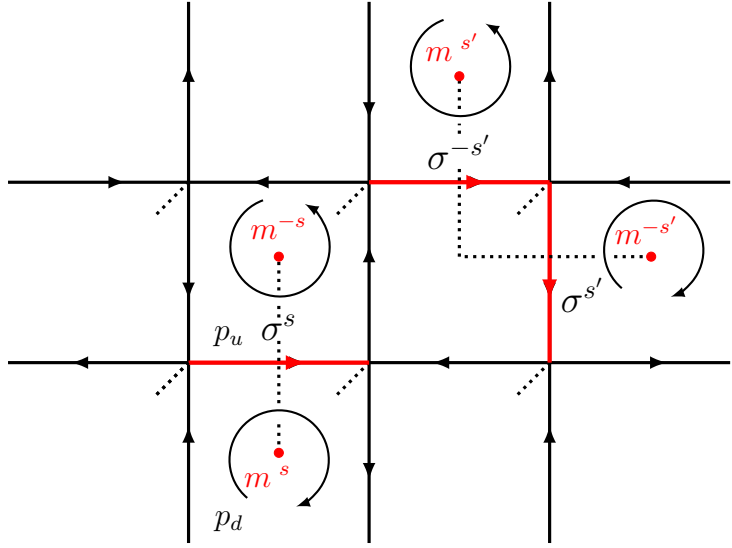
\begin{figure}[t]
    \centering
	\begin{tikzpicture}[scale=0.3,>=latex,decoration={
	    markings,
	    mark=at position 0.65 with {\arrow{>}}}
	    ]    
	\draw[very thick,dotted]({-sqrt(2)},{-sqrt(2)})--(0,0);    
	\draw[very thick,dotted]({-sqrt(2)+8},{-sqrt(2)})--(8,0);    
	\draw[postaction={decorate},very thick] (0,0)--(-8,0);    
	\draw[postaction={decorate},very thick] (16,0)--(8,0);    
	\draw[postaction={decorate},very thick] (0,-8)--(0,0);
	\draw[postaction={decorate},very thick] (0,8)--(0,0);    
	\draw[postaction={decorate},very thick] (8,0)--(8,8);   
	\draw[postaction={decorate},very thick] (8,0)--(8,-8);    
	\draw[postaction={decorate},very thick,] (0,0)--(8,0);
	\draw[postaction={decorate},very thick] (-8,8)--(0,8);
	\draw[postaction={decorate},very thick] (8,8)--(0,8);
	\draw[postaction={decorate},very thick] (8,8)--(16,8);
	\draw[postaction={decorate},very thick] (24,8)--(16,8);
	\draw[postaction={decorate},very thick] (16,0)--(24,0);
	\draw[postaction={decorate},very thick] (16,8)--(16,0); 
	\draw[postaction={decorate},very thick] (16,-8)--(16,0); 
	\draw[postaction={decorate},very thick] (16,8)--(16,16); 
	\draw[postaction={decorate},very thick] (8,16)--(8,8); 
	\draw[postaction={decorate},very thick] (0,8)--(0,16); 
	\draw[very thick,dotted]({-sqrt(2)+16},{-sqrt(2)})--(16,0);
	\draw[very thick,dotted]({-sqrt(2)+8},{-sqrt(2) +8})--(8,8);
	\draw[very thick,dotted]({-sqrt(2)+16},{-sqrt(2) +8})--(16,8);
	\draw[very thick,dotted]({-sqrt(2)+0},{-sqrt(2) +8})--(0,8); 
	\draw[postaction={decorate},ultra thick,red] (0,0)--(8,0);
	\draw[postaction={decorate},ultra thick,red] (16,8)--(16,0);
	\draw[postaction={decorate},ultra thick,red] (8,8)--(16,8);
	\node [above] at (4,0.5) {\large $\sigma^{s}$};
	\node [red, above] at (4,5) {$m^{-s}$};
	\node [below] at (1.8,2.1) {$p_u$};
	\node [red, above] at (3.5,-6.0) {$m^{\phantom{~}s}$};
	\node [below] at (1.8,-6.1) {$p_d$};
	\draw[very thick,dotted](4,-4)--(4,-2.5);
	\draw[very thick,dotted](4,-1.8)--(4,1);
	\draw[very thick,dotted](4,2)--(4,2.6);
	\draw[very thick,dotted](4,3.2)--(4,4.7);
	\draw[red, fill] (4,4.7) circle [radius=0.2];
	\draw[red, fill] (4,-4.) circle [radius=0.2];
	\draw[thick, ->] (3.3,7.2) arc (110:420:2.2cm);
	\draw[thick, ->] (2.6,-6.0) arc (590:290:2.2cm);
	\node [right] at (16,3) {\large $\sigma^{s'}$};
	\node [red, above] at (20.3,4.75) {$m^{-s'}$};
	\node [red, above] at (11.7,12.9) {$m^{\phantom{~}s'}$};
	\draw[very thick,dotted](20.5,4.7)--(18.8,4.7);
	\draw[very thick,dotted](18,4.7)--(12,4.7);
	\draw[very thick,dotted](12,4.7)--(12,8.7);
	\draw[very thick,dotted](12,11.25)--(12,12.7);
	\draw[very thick,dotted](12,10.15)--(12,10.6);
	\draw[red, fill] (12,12.7) circle [radius=0.2];
	\draw[red, fill] (20.5,4.7) circle [radius=0.2];
	\draw[thick, ->] (11.3,15.2) arc (110:420:2.2cm);
	\draw[thick, ->] (19.2,3.3) arc (590:290:2.2cm);
	\node [above] at (12,8.1) {\large $\sigma^{-s'}$};
	\draw[very thick,dotted]({-sqrt(2)},{-sqrt(2)})--(0,0);
	\end{tikzpicture}
    \caption{Here we show the action of $\sigma^{s}$ on an edge generates magnetic excitations in the neighbouring plaquettes. We also show the action of $\sigma^{s'}$ is used to translate a magnetic excitation.}
    \label{fig::MagneticExcitations}
\end{figure}

Now that we have constructed representations of our Hopf algebra lattice gauge theory operators, and defined the Hamiltonian, we can analyse excitations, which are created by violating the constraint equations given in Equation (\ref{eq::ConditionVertexGroundState},\ref{eq::ConditionPlaquetteGroundState}).
As in Kitaev models, states above the ground state can be described by quasiparticle excitations connected by strings of operators \cite{Kitaev2003,GenRibbonOp,Catherine2,KitaevModels_FiniteGroup}.
\\We will see that the fusion rules of our excitations are, in fact, the same as conventional Kitaev models, but the introduction of the braided tensor product and the non-trivial $R$ matrix will change the braiding statistics. 

We start by considering the magnetic excitations. A magnetic excitation is defined by an eigenstate of the Hamiltonian that does not satisfy the flatness constraint regarding the plaquette operators in Equation
(\ref{eq::ConditionPlaquetteGroundState}).
 % These states are obtained by applying the appropriate operators to the ground state.
Consider the state obtained by applying the operator $\sigma^{s}_{e}$ at a horizontal edge $e$ for some $s\in \mathbb{Z}_N$. We name the two plaquettes adjacent this edge $p_u$ and $p_d$ (where $u$ and $d$ stand for up and down respectively), as shown in Figure \ref{fig::MagneticExcitations}. From Equation  (\ref{eq::definitionBAtilde}), it can be seen that 
\begin{equation}
    \mathcal{B}_{p_u}\sigma^{s}_{e} \, |GS\rangle=\omega^{N-s}\sigma^{s}_{e} \, |GS\rangle\qquad \mathcal{B}_{p_d}\sigma^{s}_{e}\, |GS\rangle=\omega^{s}\sigma^{s}_{e} \, |GS\rangle\ .
\end{equation}
Since $1+\omega+\omega^2+\ldots+\omega^{N-1}=0$, this implies
\begin{equation}
    (1+\mathcal{B}_{p_i}+\mathcal{B}^{2}_{p_i}+\ldots+\mathcal{B}^{N-1}_{p_i})\sigma^{s}_{e}|GS\rangle=0\quad i=u,d\ .
\end{equation}
Similarly, it can be seen that for all vertices $v$, and all plaquettes $p\neq p_u, p_d$, we get 
\begin{equation}
    A_v\sigma^{s}_{e} \, |GS\rangle = \sigma^{s}_{e}|GS\rangle\qquad     \mathcal{B}_p\sigma^{s}_{e}\, |GS\rangle = \sigma^{s}_{e}|GS\rangle\  .
\end{equation}
These relations imply that when acting with the Hamiltonian on this state, the only non-trivial contributions come from the upper and lower adjoining plaquettes
\begin{equation}
    \begin{split}
    &H\sigma^{s}_{e}|GS\rangle =\\
    &(E_{GS}+2)|GS\rangle- \sum_{i=u,d}(1+\mathcal{B}_{p_i}+\mathcal{B}^{2}_{p_i}+\ldots+\mathcal{B}^{N-1}_{p_i})\sigma^{s}_{e}|GS\rangle=\\
    &(E_{GS}+2)\sigma^{s}_{e}|GS\rangle\ ,
    \end{split}
\end{equation}
with $E_{GS}$ the energy of the ground state. Since the energy has increased by $2$ units, we can interpret this state to be obtained by creating two particles sitting on the two plaquettes.
By convention, we say that the particle at plaquette $p_{u}$ has magnetic charge $-s$ and the one at plaquette $p_{d}$ has magnetic charge $s$. We can separate the magnetic charges
 by considering strings of sigma operators, as shown in Figure \ref{fig::MagneticExcitations}. The string must consist of alternating $\sigma^{s}$ and $\sigma^{-s}$, since excitations always come in pairs. 

As we can see, the magnetic excitations of our model are the same as those in the conventional toric code. The electric excitations, which we consider now, will provide some novelty.
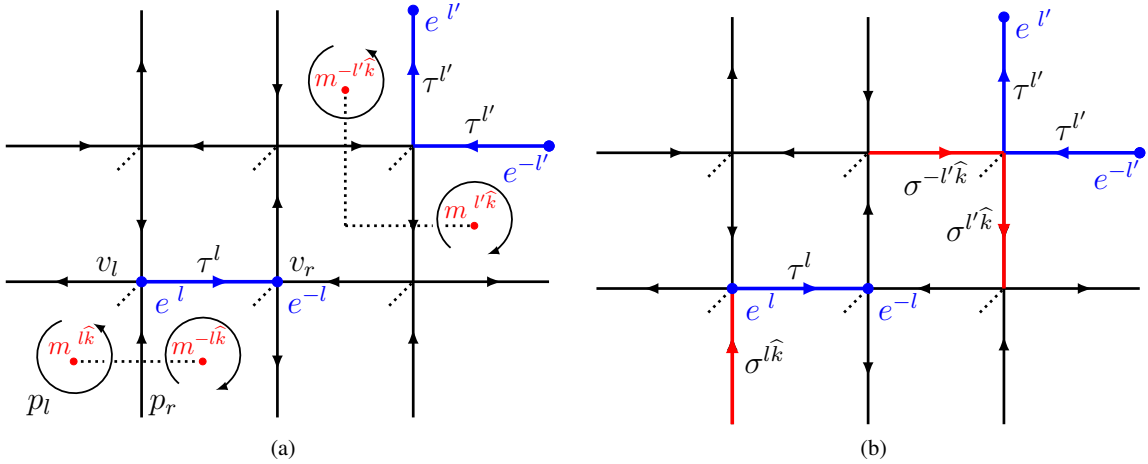
\begin{figure}[t]
	\centering
		\resizebox{0.99\linewidth}{!}
		   {  
    \subfloat[]{ 	  
	\begin{tikzpicture}[scale=0.265,>=latex,decoration={
	    markings,
	    mark=at position 0.65 with {\arrow{>}}}
	    ]     
	\draw[very thick,dotted]({-sqrt(2)},{-sqrt(2)})--(0,0);    
	\draw[very thick,dotted]({-sqrt(2)+8},{-sqrt(2)})--(8,0);    
	\draw[postaction={decorate},very thick] (0,0)--(-8,0);    
	\draw[postaction={decorate},very thick] (16,0)--(8,0);    
	\draw[postaction={decorate},very thick] (0,-8)--(0,0);
	\draw[postaction={decorate},very thick] (0,8)--(0,0);
	\draw[postaction={decorate},very thick] (8,0)--(8,8);
	\draw[postaction={decorate},blue, ultra thick] (0,0)--(8,0);
	\draw[postaction={decorate},very thick] (8,0)--(8,-8);    
	\draw[postaction={decorate},blue, ultra thick] (16,8)--(16,16);
	\draw[postaction={decorate}, very thick] (8,8)--(16,8);   
	% new lines
	\draw[postaction={decorate},very thick] (-8,8)--(0,8);
	\draw[postaction={decorate},very thick] (8,8)--(0,8);
	\draw[postaction={decorate},very thick] (16,0)--(24,0);
	\draw[postaction={decorate},very thick] (16,8)--(16,0); 
	\draw[postaction={decorate},very thick] (16,-8)--(16,0); 
	\draw[postaction={decorate},very thick] (8,16)--(8,8); 
	\draw[postaction={decorate},very thick] (0,8)--(0,16); 
	\draw[very thick,dotted]({-sqrt(2)+16},{-sqrt(2)})--(16,0);
	\draw[very thick,dotted]({-sqrt(2)+8},{-sqrt(2) +8})--(8,8);
	\draw[very thick,dotted]({-sqrt(2)+16},{-sqrt(2) +8})--(16,8);
	\draw[very thick,dotted]({-sqrt(2)+0},{-sqrt(2) +8})--(0,8);
	\draw[blue,fill] (8,0) circle [radius=0.3];
	\draw[blue,fill] (0,0) circle [radius=0.3];
	\node[above,scale=1.2] at (4,-0.1) {$\tau^{l}$};
	\node[blue, above, scale=1.2] at (1.7,-2.5) {$e^{\phantom{~}l}$};
	\node[blue, above,scale=1.2] at (9.8,-2.5) {$e^{-l}$};  
	\draw[blue,fill] (24,8) circle [radius=0.3];
	\draw[blue,fill] (16,16) circle [radius=0.3];
	\node[scale=1.2] at (17.5,12) {$\tau^{l'}$};
	\node[scale=1.2] at (20,9.6) {$\tau^{l'}$};
	\node[blue, above, scale=1.2] at (22.7,5.2) {$e^{-l'}$};
	\node[blue, above,scale=1.2] at (17.8,14.1) {$e^{\phantom{~}l'}$};
	\draw[postaction={decorate},blue, ultra thick] (24,8)--(16,8);
	\draw[blue,fill] (24,8) circle [radius=0.3];
	\draw[very thick,dotted](-4,-4.7)--(-2.2,-4.7);
	\draw[very thick,dotted](-1.6,-4.7)--(1.3,-4.7);
	\draw[very thick,dotted](1.7,-4.7)--(3.6,-4.7);
	\node [red, above] at (-4.1,-4.8) {$m^{\phantom{~}l\widehat{k}}$};
	\node[below,scale=1.2] at (-6,-6) {$p_l$};
	\node[below,scale=1.2] at (1.3,-6) {$p_r$};
	\node[above, scale=1.2] at (-2.,-0.3) {$v_l$};
	\node[above, scale=1.2] at (9.5,-0.3) {$v_r$};
	\draw[red, fill] (-4,-4.7) circle [radius=0.2];
	\draw[thick, ->] (-4.7,-2.3) arc (110:420:2.2cm);
	\node [red, above] at (3.4,-4.8) {$m^{-l\widehat{k}}$};
	\draw[red, fill] (3.6,-4.7) circle [radius=0.2];
	\draw[thick, ->] (2.2,-6.0) arc (590:290:2.2cm);
	\draw[very thick,dotted] (12,9.9)--(12,11.3);
	\draw[very thick,dotted] (12,3.3)--(12,9.5);
	\draw[very thick,dotted] (17.9,3.3)--(19.6,3.3);
	\draw[very thick,dotted](12,3.3)--(17.15,3.3);
	\node [red, above] at (12.0,11.0) {$m^{-l'\widehat{k}}$};
	\draw[red, fill] (12,11.3) circle [radius=0.2];
	\draw[thick, ->] (11.3,13.9) arc (110:420:2.2cm);
	\node [red, above] at (19.4,3.2) {$m^{\phantom{~}l'\widehat{k}}$};
	\draw[red, fill] (19.6,3.3) circle [radius=0.2];
	\draw[thick, ->] (18.2,2.0) arc (590:290:2.2cm);
	\end{tikzpicture}
	
	\label{fig::dyonicExcitations}
	}
    \hspace{0.2cm}
    \subfloat[]{
	\begin{tikzpicture}[scale=0.265,>=latex,decoration={
	    markings,
	    mark=at position 0.65 with {\arrow{>}}}
	    ]    
	\draw[very thick,dotted]({-sqrt(2)},{-sqrt(2)})--(0,0);    
	\draw[very thick,dotted]({-sqrt(2)+8},{-sqrt(2)})--(8,0);    
	\draw[postaction={decorate},very thick] (0,0)--(-8,0);    
	\draw[postaction={decorate},very thick] (16,0)--(8,0);    
	\draw[postaction={decorate},very thick] (0,-8)--(0,0);
	\draw[postaction={decorate},very thick] (0,8)--(0,0);
	\draw[postaction={decorate},very thick] (8,0)--(8,8);
	\draw[postaction={decorate},blue, ultra thick] (0,0)--(8,0);
	\draw[postaction={decorate},red, ultra thick] (0,-8)--(0,0);
	\draw[postaction={decorate},very thick] (8,0)--(8,-8);    
	\draw[postaction={decorate},blue, ultra thick] (16,8)--(16,16);
	\draw[postaction={decorate},red, ultra thick] (8,8)--(16,8);   
	\draw[postaction={decorate},very thick] (-8,8)--(0,8);
	\draw[postaction={decorate},very thick] (8,8)--(0,8);
	\draw[postaction={decorate},very thick] (16,0)--(24,0);
	\draw[postaction={decorate},very thick] (16,8)--(16,0); 
	\draw[postaction={decorate},very thick] (16,-8)--(16,0); 
	\draw[postaction={decorate},very thick] (8,16)--(8,8); 
	\draw[postaction={decorate},very thick] (0,8)--(0,16); 
	\draw[very thick,dotted]({-sqrt(2)+16},{-sqrt(2)})--(16,0);
	\draw[very thick,dotted]({-sqrt(2)+8},{-sqrt(2) +8})--(8,8);
	\draw[very thick,dotted]({-sqrt(2)+16},{-sqrt(2) +8})--(16,8);
	\draw[very thick,dotted]({-sqrt(2)+0},{-sqrt(2) +8})--(0,8);
	\draw[blue,fill] (8,0) circle [radius=0.3];
	\draw[blue,fill] (0,0) circle [radius=0.3];
	\node[above,scale=1.2] at (4,-0.1) {$\tau^{l}$};
	\node[scale=1.2] at (1.9,-4) {$\sigma^{l\widehat{k}}$};
	\node[blue, above, scale=1.2] at (1.7,-2.5) {$e^{\phantom{~}l}$};
	\node[blue, above,scale=1.2] at (9.8,-2.5) {$e^{-l}$};  
	\draw[blue,fill] (24,8) circle [radius=0.3];
	\draw[blue,fill] (16,16) circle [radius=0.3];
	\node[scale=1.2] at (17.5,12) {$\tau^{l'}$};
	\node[scale=1.2] at (20,9.6) {$\tau^{l'}$};
	\node[scale=1.2] at (12,6.5) {$\sigma^{-l'\widehat{k}}$};
	\node[blue, above, scale=1.2] at (22.7,5.2) {$e^{-l'}$};
	\node[blue, above,scale=1.2] at (17.8,14.1) {$e^{\phantom{~}l'}$};
	\draw[postaction={decorate},blue, ultra thick] (24,8)--(16,8);
	\draw[blue,fill] (24,8) circle [radius=0.3];
	\draw[postaction={decorate},red, ultra thick] (16,8)--(16,0);
	\node[scale=1.2] at (14,3.8) {$\sigma^{l'\widehat{k}}$};
	\end{tikzpicture}
	\label{fig::electricExcitations}
	}
	}
\caption{In Figure (a) we show how the action of $\tau^{l}$ on an edge creates electric excitations at the adjacent vertices but also magnetic excitations in the neighbouring plaquettes, given by the cilia at the adjoining vertices to the edge. In comparison to the Kitaev model, the action of $\tau^{l}$ creates four excitations rather than two. 
 In Figure (b) we show that the action of $\sigma^{l \hat{k}}$ on a perpendicular edge can annihilate the additional magnetic excitations. }
\end{figure}
Electric excitations are defined as eigenstates that do not satisfy the gauge invariance condition (\ref{eq::ConditionVertexGroundState}). 
Similarly to the magnetic excitations, it is not possible to create isolated particles and these particles are also obtained by acting on the ground state with the appropriate operators. 
If we were to go by direct analogy with the Kitaev model, this type of particle would be created by the action of $\tau^{l}$ at some edge. 
For example, consider the state obtained by acting with $\tau^{l}$ on a horizontal edge like in Figure \ref{fig::dyonicExcitations} and we get
\begin{equation}
    A_{v_{l}}\tau^{l}_{e}|GS\rangle=\omega^{l} \, \tau^{l}_{e}\, |GS\rangle,\qquad A_{v_{r}}\tau^{l}_{e}|GS\rangle=\omega^{N-l} \, \tau^{l}_{e} \, |GS\rangle\ .
\end{equation}
Given these relations, we can interpret this state as the one obtained by creating two electric excitations sitting at vertices $v_{l}$ and $v_{r}$, which by convention we define to have electric charges $l$ and $-l$, respectively. Acting with $\tau^l$ creates more than these two excitations though. The presence of $\sigma$ matrices attached to the edges around the plaquette cilium, in fact, have non-trivial commutation with $\tau^l$:
\begin{equation}
    \mathcal{B}_{p_{l}}\tau^{l}_{e} \, |GS\rangle=\omega^{l\widehat{k}}\tau^{l}_{e} \, |GS\rangle\qquad \mathcal{B}_{p_{r}}\tau^{l}_{e}|GS\rangle=\omega^{N-l\widehat{k}}\tau^{l}_{e} \, |GS\rangle\ .
\end{equation}
It can be checked that all the other plaquettes and vertex operators around $\tau^l$ have trivial action on this state, and this means that
\begin{equation}
    H \, \tau^{l}_{e}\, |GS\rangle = (E_{GS}+4)\, \tau^{l}_{e} \, |GS\rangle\ ,
\end{equation}
so we have extra magnetic excitations with charges $l\widehat{k}$ and $-l\widehat{k}$ at plaquettes $p_{l}$ and $p_{r}$ or, equivalently, we have extra flux quanta attached to each charge. This is also depicted in Figure \ref{fig::dyonicExcitations}.

We can therefore see that, compared to Kitaev's model, the $R$ matrix and the braided tensor product induces a change in the operators that create electric excitations and we will see that this has implications for the braiding properties of these particles. The correct creation operator for electric excitations is obtained by adding the action of an extra $\sigma$ operator. 
\\In the present case, it can be checked that acting with a $\sigma^{l\widehat{k}}$ at the bottom left edge with respect to the $\tau^{l}$ operator does the trick\footnote{Note that for a $\tau^{l}$ that acts on a vertical edge we need to act with a $\sigma^{-l\widehat{k}}$ rather than a $\sigma^{l\widehat{k}}$. This is is shown in Figure \ref{fig::electricExcitations} as well. }. 
Similar to the magnetic excitations we can separate electric excitations through strings of operators, as shown in \ref{fig::electricExcitations}.
However, these strings now have perpendicular ``hairs'', made of alternating $\sigma^{l\widehat{k}}$ and $\sigma^{-l\widehat{k}}$.

It is worth pointing out that in the original Kitaev model, these types of operators would create dyonic excitations rather than elementary ones. 
The introduction of the $R$ matrix, therefore, is introducing a mapping between the different excitations when compared with conventional Kitaev models. Since the braiding statistics of these operators on the plane are primarily independent of the Hamiltonian that was used to introduce them, we can already anticipate that these particles will have the same dyonic braiding that they have in Kitaev's model.

Following \cite{pachos_2012}, non-elementary excitations that are a combination of electric and magnetic charges are indicated by $\psi^{l,s}$, where $l$ and $s$ represent the electric and magnetic charges of the dyon. 
The fusion rules for the excitations are then the same as the Kitaev model
\begin{equation}
    e^{l}\times e^{l'}=e^{l+l'},\qquad m^{s}\times m^{s'}=m^{s+s'}, \qquad e^{l}\times m^{s}=\psi^{l,s}\ .
\end{equation}
We can see this by considering the multiplication of the operators that generate them.
We will now consider the braiding between such excitations.
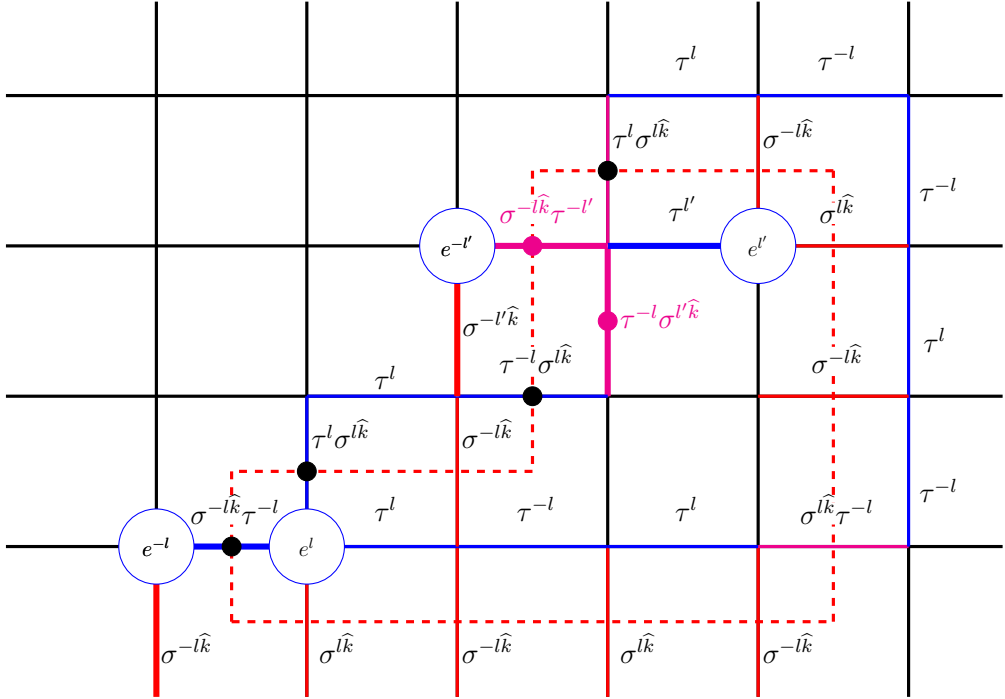
\begin{figure}[t]
		\centering
		\begin{tikzpicture}[scale=0.5,>=latex,decoration={}]
			\draw[postaction={decorate},very thick] (20,-4)--(20,0);    
			\foreach \x in {0,4,8,12,16}
			{
			    \foreach \y in {0,4,8,12}
			    {
			    \draw[postaction={decorate},very thick] ({\x},{\y})--({\x+4},{\y});
			    \draw[postaction={decorate},very thick] ({\x},{\y-4})--({\x},{\y});
			    }
			}
	
		    \foreach \y in {12}
		    {
			\foreach \x in {0,4,8,12,16,20}%{0,8,16} % spacing
			{
			    \draw[postaction={decorate},very thick] ({\x},{\y+2.5})--({\x},{\y});

			}
			}
	
			\foreach \y in {0,4,8,12}
			{
			    \draw[postaction={decorate},very thick] (-4,{\y})--(0,{\y});
	    
			}
	
			\foreach \y in {0,4,8,12}
			{
			    \draw[postaction={decorate},very thick] (20,{\y})--(22.5,{\y});
	    
			}	

		    \foreach \x in {4,8,12,16}
		    {
		        \draw[red, postaction={decorate}, very thick] ({\x},-4)--({\x},0);
		    }	
		% electric particles path
		\draw[blue,postaction={decorate},very thick] ({8},{0})--({4},{0});
		\draw[blue,postaction={decorate},very thick] ({8},{0})--({12},{0});
		\draw[blue,postaction={decorate},very thick] ({16},{0})--({12},{0});
		\draw[magenta,postaction={decorate},very thick] ({16},{0})--({20},{0});
		%\draw[blue,postaction={decorate},very thick] ({16},{0})--({20},{0});
		\draw[blue,postaction={decorate},very thick] ({20},{0})--({20},{4});
		\draw[blue,postaction={decorate},very thick] ({20},{8})--({20},{4});
		\draw[blue,postaction={decorate},very thick] ({20},{8})--({20},{12});
		\draw[blue,postaction={decorate},very thick] ({20},{12})--({16},{12});
		\draw[blue,postaction={decorate},very thick] ({12},{12})--({16},{12});
		\draw[magenta,postaction={decorate},very thick] ({12},{8})--({12},{12});
		\draw[red, postaction={decorate}, very thick] (16,4)--(20,4);
		\draw[red, postaction={decorate}, very thick] (16,8)--(20,8);
		\draw[red, postaction={decorate}, very thick] (16,8)--(16,12);
		\draw[red, postaction={decorate}, very thick] (8,0)--(8,4);
		\draw[white,postaction={decorate},very thick] ({12},{8})--({12},{4});
		\draw[magenta,postaction={decorate},style={line width=2.3pt}] ({12},{8})--({12},{4});
		\draw[white,postaction={decorate},very thick] ({16},{8})--({12},{8});
		\draw[blue,postaction={decorate},style={line width=2.3pt}] ({16},{8})--({12},{8});
		\draw[blue,postaction={decorate},very thick] ({12},{4})--({8},{4});
		\draw[blue,postaction={decorate},very thick] ({4},{4})--({8},{4});
		\draw[blue,postaction={decorate},very thick] ({4},{0})--({4},{4});
		% magnetic particles path
		\draw[red, very thick, dashed] ({2},{-2})--({18},{-2});
		\draw[red, very thick, dashed] ({18},{-2})--({18},{10});
		\draw[red, very thick, dashed] ({18},10)--({10},{10});
		\draw[red, very thick, dashed] ({10},10)--({10},{2});
		\draw[red, very thick, dashed] ({10},2)--({2},{2});
		\draw[red, very thick, dashed] ({2},2)--({2},{-2});
		\draw[white,postaction={decorate},very thick] ({8},{8})--({12},{8});
		\draw[magenta,postaction={decorate},style={line width=2.3pt}] ({8},{8})--({12},{8});
		\draw[white,postaction={decorate},very thick] ({0},{0})--({4},{0});
		\draw[blue,postaction={decorate},style={line width=2.3pt}] ({0},{0})--({4},{0});
		\draw[white,postaction={decorate},very thick] ({0},{-4})--({0},{0});
		\draw[red,postaction={decorate},style={line width=2.3pt}] ({0},{-4})--({0},{0});
		\draw[white,postaction={decorate},very thick] ({8},{4})--({8},{8});
		\draw[red,postaction={decorate},style={line width=2.3pt}] ({8},{4})--({8},{8});
		% middle dyon magentic particles
		%%%%% ----- lower left dyon  -----   %%%%%
		% electric parts
		\foreach \x  in {0,4}
		{
		\foreach \y in {0}
		{\filldraw[blue,fill=white] ({\x},{\y}) circle [radius=1];
		\node [scale=0.7] at ({0},{\y}) {${e}^{-l}$};
		\node [scale=0.7] at ({4},{\y}) {${e}^{l}$};
		}
		}
		%%%%% ----- middle dyon   -----   %%%%%
		% electric parts
		\foreach \x  in {8,16}
		{
		\foreach \y in {8}
		{
		\filldraw[blue,fill=white] ({\x},{\y}) circle [radius=1];
		\node [scale=0.7] at ({8},{\y}) {${e}^{-l'}$};
		\node [scale=0.7] at ({16},{\y}) {${e}^{l'}$};
		}
		}
		\node [scale=0.9] at ({0.8},{-2.7}) {${\sigma}^{-l\widehat{k}}$};
		\node [scale=0.9] at ({4.8},{-2.7}) {${\sigma}^{l\widehat{k}}$};
		\node [scale=0.9] at ({8.8},{-2.7}) {${\sigma}^{-l\widehat{k}}$};
		\node [scale=0.9] at ({12.8},{-2.7}) {${\sigma}^{l\widehat{k}}$};
		\node [scale=0.9] at ({16.8},{-2.7}) {${\sigma}^{-l\widehat{k}}$};
		\filldraw [fill=white, draw=white] (1.4,0.6) rectangle (2.5,1.14);
		\node [scale=0.9] at ({2.1},{1}) {${\sigma}^{-l\widehat{k}}\tau^{-l}$};
		\node [scale=0.9] at ({6.1},{1}) {$\tau^{l}$};
		\node [scale=0.9] at ({10.1},{1}) {$\tau^{-l}$};
		\node [scale=0.9] at ({14.1},{1}) {$\tau^{l}$};
		\node [scale=0.9] at ({20.8},{1.5}) {$\tau^{-l}$};
		\filldraw [fill=white, draw=white] (17.4,0.6) rectangle (18.5,1.14);
		\node [scale=0.9] at ({18.1},{1}) {${\sigma}^{l\widehat{k}}\tau^{-l}$};
		\filldraw [fill=white, draw=white] (17.4,4.6) rectangle (18.5,5.3);
		\node [scale=0.9] at ({18.1},{5}) {${\sigma}^{-l\widehat{k}}$};
		\node [scale=0.9] at ({20.7},{5.5}) {$\tau^{l}$};
		\node [scale=0.9] at ({20.8},{9.5}) {$\tau^{-l}$};
		\filldraw [fill=white, draw=white] (17.4,8.6) rectangle (18.5,9.14);
		\node [scale=0.9] at ({18.1},{9}) {${\sigma}^{l\widehat{k}}$};
		\node [scale=0.9] at ({18.1},{13}) {$\tau^{-l}$};
		\node [scale=0.9] at ({14.1},{13}) {$\tau^{l}$};
		\node [scale=0.9] at ({16.8},{11}) {$\sigma^{-l\widehat{k}}$};
		\node [scale=0.9] at ({12.9},{11}) {$\tau^{l}\sigma^{l\widehat{k}}$};
		\filldraw [fill=white, draw=white] (9.4,8.6) rectangle (10.5,9.14);
		\node [scale=0.9] at ({10.4},{9}) {$\color{magenta} {\sigma}^{-l\widehat{k}}\tau^{-l'}$};
		\node [scale=0.9] at ({14},{9}) {$\tau^{l'}$};
		\filldraw [fill=white, draw=white] (9.4,4.6) rectangle (10.5,5.3);
		\node [scale=0.9] at ({10.1},{5}) {$\tau^{-l}{\sigma}^{l\widehat{k}}$};
		\node [scale=0.9] at ({8.9},{6}) {${\sigma}^{-l'\widehat{k}}$};
		\node [scale=0.9] at ({13.4},{6.2}) {$\color{magenta}\tau^{-l}\sigma^{l'\widehat{k}}$};
		\node [scale=0.9] at ({8.8},{3}) {$\sigma^{-l\widehat{k}}$};
		\node [scale=0.9] at ({6.1},{4.5}) {$\tau^{l}$};
		\node [scale=0.9] at ({4.9},{3}) {$\tau^{l}\sigma^{l\widehat{k}}$};
		\filldraw[magenta,fill=magenta] ({12},{6}) circle [radius=0.25];
		\filldraw[magenta,fill=magenta] ({10},{8}) circle [radius=0.25];
		\filldraw[black,fill=black] ({2.0},{0.0}) circle [radius=0.25];
		\filldraw[black,fill=black] ({12.0},{10.0}) circle [radius=0.25];
		\filldraw[black,fill=black] ({10.0},{4.0}) circle [radius=0.25];
		\filldraw[black,fill=black] ({4.0},{2.0}) circle [radius=0.25];
		\end{tikzpicture}
		\caption{ The commutation relations can be found by bringing the electric particle ${e}^{l}$ around ${e}^{l'}$. 
		When there are two operators acting on an edge, they are written in the order of occurrence acting on the ground state, from right to the left. For clarity, we have added a dashed loop where the $\sigma$'s act.
		 As the loops of electric particles are trivial when acting on the ground state, the phase difference between the initial state and the final one can be obtained by commuting the loop with the action of the $l'$ electric particles. 
		The non-trivial commutations enter at the meeting points between the blue loops with the red loops and similarly at the edges coloured in magenta. Note that since edge orientation and cilia have played their role in the definition of the operators, we can avoid displaying them.}
\label{fig::BraidingPicture}
\end{figure}

\subsection{Braiding of excitations and exchange statistics}
\label{sec::Braiding}
Braiding between particles in anyon models is encoded by the \textit{$\mathcal{R}$-symbols} (not to be confused with the R matrix), which describes the exchange statistics, 
\cite{Kitaev2006}. Given two anyons $a$, $b$ that fuse to $c$, their $\mathcal{R}$-symbol is denoted $\mathcal{R}^{a,b}_{c}$. 

In conventional Kitaev quantum double models, the braiding of particles is found by considering paths in which particles loop around each other,\cite{Kitaev2003}.
In particular we can consider an electric particle $e^{l}$ looping around another electric particle $e^{l'}$ as shown in Figure \ref{fig::BraidingPicture}. 
It is understood that all the operators shown in the figure are acting upon the ground state. In order to find the braiding relations between particles, we need to move the blue and red loops across the particle creation operators, so that they can act on the ground state. 
Since the ground state is the same as the Kitaev model, and the two loops leave it unchanged, the state in which just the four particles are created is equivalent, up to a phase, to the state in which $e^{l}$ encircles $e^{l'}$. 
This phase difference can be found by commuting all the $\sigma$ operators to the right of the $\tau$ operators. The phase difference then corresponds to the monodromy between electric particles, and it is given by
\begin{equation}
	\label{eq::Monodr_Elec}
    (\mathcal{R}^{e^{l}, e^{l'}}_{e^{l+l'}})^2=  \omega^{l^{'2}\widehat{k}-l^{2}\widehat{k}-2ll'\widehat{k}} =\omega^{-2ll'\widehat{k}}\ .
\end{equation}
This implies that electric particles can have dyonic statistics, depending on the $\widehat{k}$ that comes from the $R$ matrix. 
Further, they can be mapped to the dyonic excitation of Kitaev models. 
% \footnote{The monodromy of electric particles has also been derived for $\mathbb{Z}_{N}$ gauge theory in \cite{Mark_Propitus_Twisted}.}
However, as already pointed out, dyonic excitations are not elementary in nature, as dyonic excitations are a bound state of electric and magnetic charges that can be broken up into lower energy states. 
If we choose an $R$-matrix with $k=0$, then the braiding reduces exactly to the braiding in conventional Kitaev models.
The commutation between an electric particle $e^{l}$ and a magnetic particle $m^{l'}$ is the same as usual, and the braiding statistics are given by
\begin{equation}
	\label{eq::Mix}
    (\mathcal{R}^{e^{l},m^{s}}_{\psi^{l,s}})^{2}=\omega^{ls}\ ,
\end{equation}
while magnetic particles are bosonic relative to each other
\begin{equation}
	\label{eq::Monodr_Magne}
    (\mathcal{R}^{m^{s}, m^{s'}}_{m^{s+s'}})^{2}=1\ .
\end{equation}
Note Equation (\ref{eq::Mix}) and Equation (\ref{eq::Monodr_Magne}), are not affected by the $\widehat{k}$ parameter (and therefore by the $R$ matrix).
We shall succinctly display the effect of the $\widehat{k}$ dependence using the twist factors. 
The $\mathcal{R}$-symbols are related to the twist factors $\theta_a$ by
\begin{equation}
	\label{eq::Ribbon}
	\theta_a = \sum_{c}\frac{d_{c}}{d_{a}} \, \mathcal{R}^{a,a}_c\ .
\end{equation}
For Abelian anyons, the quantum dimensions $d_{c}$, $d_{a}$ are equal to one and the summation restricts to the single fusion outcome of the two $a$ anyons,\cite{Kitaev2006,KitaevModels_FiniteGroup}.
It is now easy to compare the twists obtained in conventional Kitaev models with what we obtained:
\begin{center}
\begin{tabular}{|c|c c c|} 
 \hline
 & $\theta_{m^{s}}$ & $\theta_{e^{l}}$ & $\theta_{\psi^{l,s}}$ \\ [0.5ex]
 \hline
 Hopf algebra gauge theory & $1$ & $\omega^{-l^2\widehat{k}}$ & $\omega^{ls-l^2\widehat{k}}$ \\ 
 \hline
 Kitaev Models & $1$ & $1$ & $\omega^{ls}$ \\
 \hline
\end{tabular}
\end{center}
Observe how the process that we described is reminiscent of flux attachment \cite{Goldhaber_flux_attachment, Wilczeck_flux_attachments, Bais_flux_attachment,2021arXiv211010169R}. The elementary charges behave as if they have $-\widehat{k}$ units of flux attached to them during the braiding. In a full exchange, we get just the Aharonov-Bohm phase factors for taking each charge around the flux of the other particle (see (\ref{eq::Monodr_Elec})).

This construction can be readily generalised to a non-Abelian group with an Abelian normal subgroup. As an example, one could consider a toric code model constructed over say, $D_{3}$. Then defining an R matrix on $\mathbb{C}\mathbb{Z}_{3}\unlhd \mathbb{C}D_{3}$, we can examine permutations of order three of the fusion algebra of $D(D_{3})$. This will lead again to a mixing of the dyonic and electric sectors in the model. In particular, permutations of the fusion algebra of $D(D_{3})$ have already been studied in \cite{D3FusionModularInvar}, and are connected to modular invariants of the category of representations.   

To summarise, the exchange statistics are determined by the $k$ parameter in the $R$ matrix. 
The introduction of the braided tensor product (and the related $R$ matrix)  amounts to interchange particle $\psi^{ls}$ and $\psi^{l(s-l\widehat{k})}$ with respect to conventional Kitaev model. The exchange statistics are, therefore, formally the same as the original Kitaev model, but which excitations appear as elementary, minimal energy, depends on the specific choice of the $R$ matrix.
% but which excitations are considered elementary depends on the specific choice of the $R$ matrix.

%%%%%%%%%%%%%%%%%%%%%%%%%%%%%%%%%%%%%%%%%%%%%%%%%%%%
\section{Conclusions.}
\label{sec::Conclusions}
We have constructed examples of Hopf algebra lattice gauge theory models using the formalism developed in \cite{Catherine1}. 
We have seen that the introduction of a non-trivial choice of quasitriangular structure on $\mathbb{C}\mathbb{Z}_{N}$ and the braided
tensor product breaks the explicit electric-magnetic duality of the model's stabilizer operators and changes the identification of braiding statistics while leaving the ground state
unchanged.
We find this to be a rather interesting result; since we recover the familiar fusion statistics one would expect. However, the introduction of the braided tensor product and associated R matrix induces a permutation on the braiding statistics on representations of the quantum double. We observe that the introduction of a nontrivial $R$-matrix amounts to flux attachment permuting the excitation spectrum.
Furthermore, one could construct a Hopf algebra gauge theory with gauge symmetry given by the dihedral group, $D_{N}$, then using an R matrix defined on $\mathbb{Z}_{N}$, a normal subgroup of $D_{N}$, one can expect the permutation on the fusion algebra to be of order $N$.  
It would be interesting to generalize this work to other Hopf algebras (beyond group algebras) which allow for multiple $R$-matrices, to see if a similar interpretation of the permutation of the excitation spectrum is possible there. 

% Morita equivalence?

%%%%%%%%%%%%%%%%%%%%%%%%%%%%%%%%%%%%%%%%%%%%%%%%%%%%
\section{Acknowledgments}
The authors would like to thank A. Bullivant for several valuable discussions and comments on the manuscript, we would also like to extend gratitude to Gert Vercleyen for discussions on the automorphisms of quantum doubles. 
The authors would also like to extend gratitude to the Perimeter Institute for its hospitality during the conference on ``Hopf Algebras in Kitaev's Quantum Double Models: Mathematical Connections from Gauge Theory to Topological Quantum Computing and Categorical Quantum Mechanics'', where this work was initiated and in particular to Prince Osei for his role in organizing the conference and for useful discussions. 
 D.P. and J.K.S. acknowledge financial support from Science Foundation Ireland through Principal Investigator Awards 12/IA/1697 and 16/IA/4524.
A.C. was supported through IRC Government of Ireland Postgraduate Scholarship GOIPG/2016/722.

\section*{References}
\bibliographystyle{alpha}
\bibliography{Hopf_paper_Bibliography}

\newcommand{\etalchar}[1]{$^{#1}$}
\begin{thebibliography}{{Maj}95a}

\bibitem[Avd93]{Bais_flux_attachment}
F.~{Alexander Bais}, Peter {van Driel}, and Mark {de Wild Propitius}.
\newblock {Anyons in discrete gauge theories with Chern-Simons terms}.
\newblock {\em Nuclear Physics B}, 393(3):547--570, March 1993.

\bibitem[BB07]{QuditSurfaceCode}
Stephen~S. {Bullock} and Gavin~K. {Brennen}.
\newblock {Qudit surface codes and gauge theory with finite cyclic groups}.
\newblock {\em Journal of Physics A Mathematical General}, 40(13):3481--3505,
  March 2007.

\bibitem[BCKA10]{EMdualityLattice}
Oliver {Buerschaper}, Matthias {Christandl}, Liang {Kong}, and Miguel {Aguado}.
\newblock {Electric-magnetic duality of lattice systems with topological
  order}.
\newblock {\em arXiv e-prints}, page arXiv:1006.5823, June 2010.

\bibitem[BK98]{QuantumCodesLatticeBoundary}
S.~B. {Bravyi} and A.~Yu. {Kitaev}.
\newblock {Quantum codes on a lattice with boundary}.
\newblock {\em arXiv e-prints}, pages quant--ph/9811052, November 1998.

\bibitem[BM93]{Braided_Hisenberg_Group}
W.~K. {Baskerville} and S.~{Majid}.
\newblock {The braided Heisenberg group}.
\newblock {\em Journal of Mathematical Physics}, 34(8):3588--3606, August 1993.

\bibitem[BM08]{KitaevModels_FiniteGroup}
H.~{Bombin} and M.~A. {Martin-Delgado}.
\newblock {Family of non-Abelian Kitaev models on a lattice: Topological
  condensation and confinement}.
\newblock {\em prb}, 78(11):115421, September 2008.

\bibitem[BMCA13]{Buerschaper2013}
Oliver {Buerschaper}, Juan~Mart{\'\i}n {Mombelli}, Matthias {Christand l}, and
  Miguel {Aguado}.
\newblock {A hierarchy of topological tensor network states}.
\newblock {\em Journal of Mathematical Physics}, 54(1), January 2013.

\bibitem[BSW11]{D3FusionModularInvar}
Salman {Beigi}, Peter~W. {Shor}, and Daniel {Whalen}.
\newblock {The Quantum Double Model with Boundary: Condensations and
  Symmetries}.
\newblock {\em Communications in Mathematical Physics}, 306(3):663--694,
  September 2011.

\bibitem[{Cas}99]{ProjectiveGroupAlgebra}
R.~{Casalbuoni}.
\newblock {Projective Group Algebras}.
\newblock {\em International Journal of Modern Physics A}, 14(1):129--146,
  January 1999.

\bibitem[CCK19]{ExactBoundaryKitaev}
Yevheniia {Cheipesh}, Lorenzo {Cevolani}, and Stefan {Kehrein}.
\newblock {Exact description of the boundary theory of the Kitaev toric code
  with open boundary conditions}.
\newblock {\em prb}, 99(2):024422, January 2019.

\bibitem[CCY21]{GenRibbonOp}
Penghua {Chen}, Shawn~X. {Cui}, and Bowen {Yan}.
\newblock {Ribbon operators in the generalized Kitaev quantum double model
  based on Hopf algebras}.
\newblock {\em arXiv e-prints}, page arXiv:2105.08202, May 2021.

\bibitem[Che15]{CharacterTheoryProjectiveRepresentations}
Chuangxun Cheng.
\newblock A character theory for projective representations of finite groups.
\newblock {\em Linear Algebra and its Applications}, 469:230 -- 242, 2015.

\bibitem[CM22]{AlgBoundary}
Alexander {Cowtan} and Shahn {Majid}.
\newblock {Algebraic Aspects of Boundaries in the Kitaev Quantum Double Model}.
\newblock {\em arXiv e-prints}, page arXiv:2208.06317, August 2022.

\bibitem[{Dav}97]{DavydovQuasiStruct}
A.~A. {Davydov}.
\newblock {Quasitriangular structures on cocommutative Hopf algebras}.
\newblock {\em eprint arXiv:q-alg/970600}, pages q--alg/9706007, June 1997.

\bibitem[dB95]{Bais_Propitius}
Mark {de Wild Propitius} and F.~Alexander {Bais}.
\newblock {Discrete gauge theories}.
\newblock {\em arXiv e-prints}, pages hep--th/9511201, November 1995.

\bibitem[{de }95]{Mark_Propitus_Twisted}
Mark {de Wild Propitius}.
\newblock {\em {Topological interactions in broken gauge theories}}.
\newblock PhD thesis, -, November 1995.

\bibitem[Gol76]{Goldhaber_flux_attachment}
Alfred~S. Goldhaber.
\newblock Connection of spin and statistics for charge-monopole composites.
\newblock {\em Phys. Rev. Lett.}, 36:1122--1125, May 1976.

\bibitem[JKT22]{Boundary_Generalized_Kitaev_Hopf_Gauge}
Zhian {Jia}, Dagomir {Kaszlikowski}, and Sheng {Tan}.
\newblock {Boundary and domain wall theories of 2d generalized quantum double
  model}.
\newblock {\em arXiv e-prints}, page arXiv:2207.03970, July 2022.

\bibitem[Kas95]{Kassel}
Christian Kassel.
\newblock {\em Quantum groups / Christian Kassel.}
\newblock Graduate texts in mathematics ; 155. Springer-Verlag, New York, 1995.

\bibitem[Kit03]{Kitaev2003}
A.Yu. Kitaev.
\newblock Fault-tolerant quantum computation by anyons.
\newblock {\em Annals of Physics}, 303(1):2 -- 30, 2003.

\bibitem[Kit06]{Kitaev2006}
Alexei Kitaev.
\newblock Anyons in an exactly solved model and beyond.
\newblock {\em Annals of Physics}, 321(1):2 -- 111, 2006.
\newblock January Special Issue.

\bibitem[LB12]{Minimal}
Nicolai {Lang} and Hans~Peter {Buchler}.
\newblock {Minimal instances for toric code ground states}.
\newblock {\em pra}, 86(2):022336, August 2012.

\bibitem[Maj92]{Braided_Groups_Majid}
Shahn Majid.
\newblock Braided {Groups}.
\newblock {\em Les rencontres physiciens-math\'ematiciens de Strasbourg
  -RCP25}, 43, 1992.
\newblock talk:7.

\bibitem[{Maj}95a]{BraidedMajid}
S.~{Majid}.
\newblock Algebras and hopf algebras in braided categories.
\newblock {\em eprint arXiv:q-alg/9509023}, pages q--alg/9509023, September
  1995.

\bibitem[Maj95b]{Majid_1995}
Shahn Majid.
\newblock {\em Foundations of Quantum Group Theory}.
\newblock Cambridge University Press, 1995.

\bibitem[Maj02]{MajidQuantumPrimer}
Shahn Majid.
\newblock {\em A Quantum Groups Primer}.
\newblock London Mathematical Society Lecture Note Series. Cambridge University
  Press, 2002.

\bibitem[Meu17]{Catherine2}
Catherine Meusburger.
\newblock {Kitaev lattice models as a Hopf algebra gauge theory}.
\newblock {\em Commun. Math. Phys.}, 353(1):413--468, 2017.

\bibitem[MW21]{Catherine1}
Catherine Meusburger and Derek~K. Wise.
\newblock Hopf algebra gauge theory on a ribbon graph.
\newblock {\em Reviews in Mathematical Physics}, 33(05):2150016, 2021.

\bibitem[Pac12]{pachos_2012}
Jiannis~K. Pachos.
\newblock {\em Introduction to Topological Quantum Computation}.
\newblock Cambridge University Press, 2012.

\bibitem[PWS{\etalchar{+}}07]{ToricSim}
J.~K. {Pachos}, W.~{Wieczorek}, C.~{Schmid}, N.~{Kiesel}, R.~{Pohlner}, and
  H.~{Weinfurter}.
\newblock {Revealing anyonic features in a toric code quantum simulation}.
\newblock {\em arXiv e-prints}, page arXiv:0710.0895, October 2007.

\bibitem[{Rad}21]{2021arXiv211010169R}
Djordje {Radicevic}.
\newblock {Confinement and Flux Attachment}.
\newblock {\em arXiv e-prints}, page arXiv:2110.10169, October 2021.

\bibitem[SLS{\etalchar{+}}21]{Experimental_Toric_Code}
K.~J. {Satzinger}, Y.~J. {Liu}, A.~{Smith}, C.~{Knapp}, M.~{Newman},
  C.~{Jones}, Z.~{Chen}, C.~{Quintana}, X.~{Mi}, A.~{Dunsworth}, C.~{Gidney},
  I.~{Aleiner}, F.~{Arute}, K.~{Arya}, J.~{Atalaya}, R.~{Babbush}, J.~C.
  {Bardin}, R.~{Barends}, J.~{Basso}, A.~{Bengtsson}, A.~{Bilmes},
  M.~{Broughton}, B.~B. {Buckley}, D.~A. {Buell}, B.~{Burkett}, N.~{Bushnell},
  B.~{Chiaro}, R.~{Collins}, W.~{Courtney}, S.~{Demura}, A.~R. {Derk},
  D.~{Eppens}, C.~{Erickson}, L.~{Faoro}, E.~{Farhi}, A.~G. {Fowler},
  B.~{Foxen}, M.~{Giustina}, A.~{Greene}, J.~A. {Gross}, M.~P. {Harrigan},
  S.~D. {Harrington}, J.~{Hilton}, S.~{Hong}, T.~{Huang}, W.~J. {Huggins},
  L.~B. {Ioffe}, S.~V. {Isakov}, E.~{Jeffrey}, Z.~{Jiang}, D.~{Kafri},
  K.~{Kechedzhi}, T.~{Khattar}, S.~{Kim}, P.~V. {Klimov}, A.~N. {Korotkov},
  F.~{Kostritsa}, D.~{Landhuis}, P.~{Laptev}, A.~{Locharla}, E.~{Lucero},
  O.~{Martin}, J.~R. {McClean}, M.~{McEwen}, K.~C. {Miao}, M.~{Mohseni},
  S.~{Montazeri}, W.~{Mruczkiewicz}, J.~{Mutus}, O.~{Naaman}, M.~{Neeley},
  C.~{Neill}, M.~Y. {Niu}, T.~E. {O{\textquoteright}Brien}, A.~{Opremcak},
  B.~{Pat{\'o}}, A.~{Petukhov}, N.~C. {Rubin}, D.~{Sank}, V.~{Shvarts},
  D.~{Strain}, M.~{Szalay}, B.~{Villalonga}, T.~C. {White}, Z.~{Yao}, P.~{Yeh},
  J.~{Yoo}, A.~{Zalcman}, H.~{Neven}, S.~{Boixo}, A.~{Megrant}, Y.~{Chen},
  J.~{Kelly}, V.~{Smelyanskiy}, A.~{Kitaev}, M.~{Knap}, F.~{Pollmann}, and
  P.~{Roushan}.
\newblock {Realizing topologically ordered states on a quantum processor}.
\newblock {\em Science}, 374(6572):1237--1241, December 2021.

\bibitem[SXZ{\etalchar{+}}18]{PhysRevLett.121.030502}
Chao Song, Da~Xu, Pengfei Zhang, Jianwen Wang, Qiujiang Guo, Wuxin Liu, Kai Xu,
  Hui Deng, Keqiang Huang, Dongning Zheng, Shi-Biao Zheng, H.~Wang, Xiaobo Zhu,
  Chao-Yang Lu, and Jian-Wei Pan.
\newblock Demonstration of topological robustness of anyonic braiding
  statistics with a superconducting quantum circuit.
\newblock {\em Phys. Rev. Lett.}, 121:030502, Jul 2018.

\bibitem[Wen07]{Wen}
Xiao-Gang Wen.
\newblock {\em {Quantum Field Theory of Many-Body Systems: From the Origin of
  Sound to an Origin of Light and Electrons}}.
\newblock Oxford University Press, 09 2007.

\bibitem[Wil82]{Wilczeck_flux_attachments}
Frank Wilczek.
\newblock Magnetic flux, angular momentum, and statistics.
\newblock {\em Phys. Rev. Lett.}, 48:1144--1146, Apr 1982.

\bibitem[ZXW{\etalchar{+}}16]{PhysRevLett.117.110501}
Y.~P. Zhong, D.~Xu, P.~Wang, C.~Song, Q.~J. Guo, W.~X. Liu, K.~Xu, B.~X. Xia,
  C.-Y. Lu, Siyuan Han, Jian-Wei Pan, and H.~Wang.
\newblock Emulating anyonic fractional statistical behavior in a
  superconducting quantum circuit.
\newblock {\em Phys. Rev. Lett.}, 117:110501, Sep 2016.

\end{thebibliography}

\appendix

\section{Quasitrangular Hopf algebra from a finite group}
\label{appendix::Quasi_Hopf_Group}
This section will provide some definitions of a quasitriangular Hopf algebra from a finite group.
 See \cite{Buerschaper2013,Kassel,Catherine1,DavydovQuasiStruct} and example 1.7, 1.8 in \cite{BraidedMajid} for further details.
Let $G$ be a finite, not necessarily commutative finite group, then we can construct $\mathbb{C}G$, the free module over $G$, also known as a group algebra.  
On this space, we have the following canonical Hopf algebra structure,
\begin{equation}
	\mu(g,h) = g \cdot h \hspace{20pt}  \Delta(g) = g \otimes g \hspace{20pt} S(g) = g^{-1} \hspace{20pt} \eta(k_{\mathbb{C}}) = k_{\mathbb{C}}e \hspace{20pt} \epsilon(g) = 1_{\mathbb{C}}
	\label{eq::CG_Hopf}
\end{equation}
where all of these maps are $\mathbb{C}$ linear. A Hopf algebra constructed in this way is cocommutative and semisimple.
The Haar integral is given by a sum over all of the elements of the group. So, for example, in  $\mathbb{C}\mathbb{Z}_N$ it is given by
\begin{equation}
    l = e+a+a^{2}+\ldots +a^{N-1}.
\end{equation}
Now we would like to describe the dual Hopf algebra, $K^{*}$, which is a Hopf algebra constructed on $G^{*}$ = Hom($\mathbb{C}G,\mathbb{C}$), this is a semisimple and commutative Hopf algebra. 
There are two standard definitions of a basis for this space which we will now describe. 
The basis of functions,
\begin{equation}
\begin{split}
	\mu(\delta_{g},\delta_{h}) = \delta_{g}(h)\delta_{h}  \quad  
	& \Delta(\delta_{g}) =  \sum_{u v =g} \delta_{u} \otimes \delta_{v} 
	\quad S(\delta_{g})= \delta_{g^{-1}} 
	\\  
  & \hspace{-40pt} \epsilon(\delta_{g}) = \delta_{g}(e)  \quad \eta(k_{\mathbb{C}}) = \delta_{e}
  \label{eq::Hopf_DeltaBasis}
\end{split}
\end{equation}
where $\delta_{g}(h)$ is equal to $1$ if $g=h$ and $0$ otherwise.
There is also another basis given by matrix elements of a representation of $G$,
\begin{equation}
\label{eq::Hopf_RepBasis}
\begin{split}
 \mu(\pi_{ij} , \pi_{kl})(g) =  \pi_{ij}(g) \pi_{kl}(g) \quad  & \Delta(\pi_{ij}) = \sum_{k} \pi_{ik} \otimes \pi_{kj}  \quad  	S(\pi_{ij}) = \pi_{ji}^{*}
  \\
  & \hspace{-60pt} \epsilon(\pi_{ij}) = \pi_{ij}(e) \quad \eta(\pi_{ij}) = \pi_{0}
\end{split}
\end{equation}
where $\pi_{0}$ is the matrix element of the trivial representation.
For $(\mathbb{C}\mathbb{Z}_{N})^{*}$ the characters are grouplike: $\Delta(\chi_{i})= \chi_{i} \otimes \chi_{i}$. 
The Haar integral for $(\mathbb{C}\mathbb{Z}_{N})^{*}$ is the sum over the character ring,
\begin{equation}
    \lambda = \chi_{0} + \chi_{1} + \ldots + \chi_{N-1}.
\end{equation}
The dual Hopf algebra $K^{*}$ is cocommutative if and only if the group is commutative. 
The definition of a quasitriangular structure is given by the R matrix, which satisfies the following axioms
\begin{equation}
	\begin{split}
		& R \Delta(x) R^{-1} = (T \circ \Delta) (x) \\
		& ( \Delta \otimes \text{id}) (R) = R_{13}R_{23} \\
		& ( \text{id} \otimes \Delta) (R) = R_{13}R_{12} \\
	\end{split}
\end{equation}
where $T: K \otimes K \rightarrow K \otimes K$ is the flip map, which exchanges tensor factors.
% We can see that the $R$ matrix is a weakening of cocommutativity.
The R matrix satisfies the Yang Baxter equation,
\begin{equation}
\label{eq::YangBaxter}
	R_{12} R_{13} R_{23} = R_{23}R_{13}R_{12}.
\end{equation}
The R matrix for a finite group can be defined as follows 
\begin{equation}
	R = \frac{1}{N} \sum_{g_{1},g_{2}} r(g_{1}, g_{2}) g_{1} \otimes g_{2}
	\end{equation}
This gives $\mathbb{C}G$ the structure of a quasitriangular Hopf algebra if $r(g_{1},g_{2})$ is a $g$-conjugation invariant bicharacter. 
This allows an R matrix defined on a normal subgroup  $H \unlhd G$ to be lifted to $G$. 
The R matrix defined as such satisfies the two cocycle condition, i.e,
\begin{equation}
		r(g_{1},g_{2}\cdot g_{3})\hspace{2pt} r^{*}(g_{1},g_{3})\hspace{2pt}
		r(g_{3},g_{2}) \hspace{2pt} r^{*}(g_{1} \cdot g_{3},g_{2}) = 1.
	\end{equation}
where $r^{*}$ is the complex conjugation of $r$.
For example, when $G = \mathbb{Z}_{N} $, we gave an explicit R matrix satisfying the two cocycle condition in Eq. (\ref{eq::ZnRmatrixAppendix}).

In \cite{DavydovQuasiStruct}, it is shown that this is the most general form of quasitriangular structure on a finite-dimensional cocommutative Hopf algebra.
The canonical Hopf algebra structure on a group algebra is equivalently; semisimple, cosemisimple and the antipode, is of order two, $S^{2} = \text{id}$ \cite{Kassel,Buerschaper2013,Catherine1}.

\section{Braided tensor product and plaquette operator supplementary material}
\label{appendix::AlgProofs}
In this appendix, we shall collect some proofs for the results stated in the main text. 
The proofs in this section are true for any quasitriangular semisimple finite dimensional Hopf algebra over an algebraically closed field. 
In particular, we will explain how all the different local algebra structures, as expressed in Equation (\ref{eq::BraidedTensProduct}), come together to form a Hopf algebra gauge theory. This procedure is 
 explained in \cite{Catherine1}, and we refer there for a more formal approach.

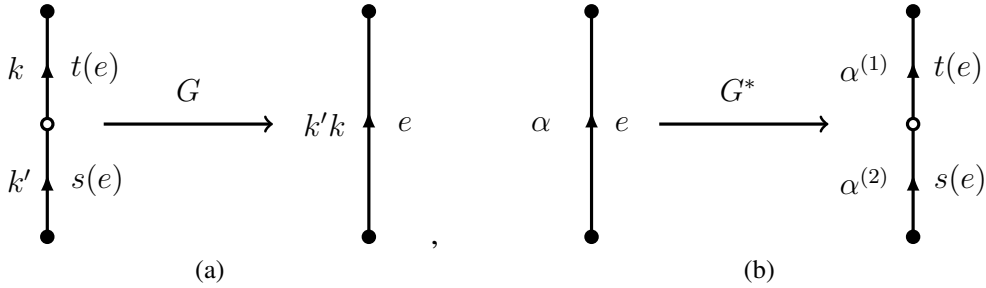
\begin{figure}[H]
\centering
    \subfloat[]
	{
	\begin{tikzpicture}[>=latex,scale=0.75,>=latex,very thick,decoration={
	    markings,
	    mark= at position 0.56 with {\arrow{>}}}]   
	    \draw[postaction={decorate}] ({-6.5},{0})--({-6.5},{2});
	    \draw[postaction={decorate}] ({-6.5},{2})--({-6.5},{4});
		\draw[fill] ({-6.5},{0}) circle [radius=0.1];
		\draw[fill] ({-6.5},{4}) circle [radius=0.1]; 
		\draw[fill=white] ({-6.5},{2}) circle [radius=0.1];	
		\node [left] at (-6.56,1) {$k'$};
		\node [left] at (-6.56,3) {$k\phantom{'}$};		
		\node [right] at (-6.3,1) {$s(e)$};
		\node [right] at (-6.3,3) {$t(e)$};	
		\draw [->,-{Classical TikZ Rightarrow}] (-5.5,2)--(-2.5,2);
		\node [above] at (-4.,2.2) {$G$};
		\draw[postaction={decorate}] ({-0.8},{0})--({-0.8},{4});
		\draw[fill] ({-0.8},{0}) circle [radius=0.1];
		\draw[fill] ({-0.8},{4}) circle [radius=0.1]; 
		\node [right] at (-0.5,2) {$e$};
		\node [left] at (-1,2) {$k'k$};
	\end{tikzpicture}
	\label{fig::GSplittingMap}
	},
    \qquad
    \subfloat[]{
	\begin{tikzpicture}[>=latex,scale=0.75,>=latex,very thick,decoration={
	    markings,
	    mark= at position 0.56 with {\arrow{>}}}]   
	    \draw[postaction={decorate}] ({-6.2},{0})--({-6.2},{4});
		\draw[fill] ({-6.2},{0}) circle [radius=0.1];
		\draw[fill] ({-6.2},{4}) circle [radius=0.1]; 
		\draw [->,-{Classical TikZ Rightarrow}] (-5,2)--(-2,2);
		\node [above] at (-3.6,2.2) {$G^{*}$};
		\node [left] at (-6.7,2) {$\alpha$};
		\node [right] at (-6,2) {$e$};
		\node [left] at (-0.7,1) {$\alpha^{(2)}$};
		\node [left] at (-0.7,3) {$\alpha^{(1)}$};
		\draw[postaction={decorate}] ({-0.5},{0})--({-0.5},{2});
		\draw[fill] ({-0.5},{0}) circle [radius=0.1];
		\draw[fill] ({-0.5},{4}) circle [radius=0.1]; 
		\draw[postaction={decorate}] ({-0.5},{2})--({-0.5},{4});	
		\draw[fill=white] ({-0.5},{2}) circle [radius=0.1];
		\node [right] at (-0.35,1) {$s(e)$};
		\node [right] at (-0.35,3) {$t(e)$};
	\end{tikzpicture}
	\label{fig::SplittingMap}
	}
    \caption{Here, we display the action of the splitting map given in Equation \ref{eq::GMaps}. The $G$ map sends the two elements belonging to the edges of two local vertex neighbourhoods that correspond to the same edge ($s(e)$ and $t(e)$), to a single element on that edge. Similarly the $G^{*}$ map splits a function $\alpha$ assigned to an edge $e$ to functions $\alpha^{(1)}$ and $\alpha^{(2)}$ assigned to edge ends $t(e)$ and $s(e)$ respectively.
	 We use these maps to move from the full lattice to the collection of local vertex neighbourhoods.
	}
	\label{fig::GMap}
\end{figure} 

The procedure of stitching together all the different local algebras is related to the existence of a map that connects the elements of the algebra associated with two edges to a single one. 
This map needs to be linear and has to preserve the local gauge structure ( $K^{V}$- module homomorphism). 
The most natural one, in this sense, is related to the algebra product. 
We will refer to this map as the $G$ map and the corresponding dual as $G^{*}$, their action is given as
\begin{equation}
	G:  (k\otimes k')_{s(e), t(e)}	\rightarrow (k'k)_{e} \qquad G^{*}:  (\alpha)_{e}	\rightarrow (\alpha^{(2)}  \otimes \alpha^{(1)})_{
	s(e), t(e)}\ ,
		\label{eq::GMaps}
\end{equation}
where, here and in the following, we will call $s(e)$ ($t(e)$) the starting edge (target edge) associated with the splitting of an edge $e$. 
We show the action of $G$ and $G^{*}$ pictorially in Figure \ref{fig::GMap}.
% Note that given two Hopf algebras $H$, $K$ and a homomorphism between them $F:H\to K$, then the  dual $F^*:K^{*}\to H^{*}$ to be given by
% \begin{equation}
% 	\langle F^{*}(\alpha), k \rangle = 	\langle \alpha, F(k) \rangle
% 	\hspace{10pt} \forall k \in K, \alpha \in K^{*}\ .
% 	\label{eq::Dualisation}
% \end{equation}
% Further details
%  % for the case in which $F$ is a module action,
%  can be found in Remark B.3 of \cite{Catherine1}.
The $G^{*}$ map is used to construct the braided tensor product on the whole lattice, extending it by linearity and imposing it to be 
a module homomorphism under the action given by the braided tensor product  itself
\begin{equation}    G^{*}\left((\alpha)_{e_{1}}*(\beta)_{e_{2}}\right)=G^{*}\left((\alpha)_{e_{1}}\right)*G^{*}\left((\beta)_{e_{2}}\right)
\end{equation}
for any collection of functions $\alpha,\ \beta$ on any collection of edges $e_{1},\ e_{2}$.

This set up the stage for to prove the main results of this section, which are given in (\ref{eq::plaquetteGeneralFormula}), (\ref{eq::plaquetteCommutations}) and (\ref{eq::commVertexPlaquette}) in the main text. 
Note that (\ref{eq::plaquetteCommutations}) is the same as Lemma 5.10 \cite{Catherine1} and (\ref{eq::commVertexPlaquette}) is essentially a refinement of the second formula in Theorem 5.7. 

We can start now with the following  

\begin{proposition}
\label{proposition::1}
Consider $ f \in K^{*}$, 
and $R $ 
the R-matrix of 
$K$, then the following is true
\begin{equation}
\begin{split}
\left\langle\phi^{(2-\tau)(1+\tau)}\otimes f^{(1)},\left(S^{\tau}\otimes S\right)(R)\right\rangle\left\langle\phi^{(1+\tau)}\otimes f^{(2)},\left(S^{\tau}\otimes \text{id}\right)(R)\right\rangle\phi^{(2-\tau)(2-\tau)}=\epsilon^{*}(f)\, \phi
	\end{split}
	\label{proposition1::eq1}
\end{equation}
where $\epsilon^*$ is the counit of $K^{*}$, $S$ is the antipode of $K$, $\tau= 0, 1$ denotes an incoming, outgoing edge respectively and the bracketed superscripts on $\phi$ and $f$ denote Sweedler indices.
\end{proposition}
\begin{proof}

Since $\tau=0,\ 1$ we can just check the formula directly.
\paragraph{Case 1: $\tau=0$}
In this case, the left side of the formula becomes
\begin{equation}
    \left\langle\phi^{(2)(1)}\otimes f^{(1)},\left(id\otimes S\right)(R)\right\rangle\left\langle\phi^{(1)}\otimes f^{(2)},R\right\rangle\phi^{(2)(2)}\ .
\end{equation}
Since $(S\otimes S)(R)=R$
 and dualising the action of the antipode, we get that the above function is the same as,
\begin{equation}
   \left\langle\phi^{(2)}\otimes f^{(1)},R^{-1}\right\rangle\left\langle\phi^{(1)}\otimes f^{(2)}, R\right\rangle\phi^{(3)}\ .
\end{equation}
Using $\left(S\otimes id\right)(R)=R^{-1}$ and that $\left(S\otimes S\right)(R)=R$ we have
\begin{equation}
    \left\langle\phi^{(1)}\otimes\phi^{(2)}\otimes f^{(1)}\otimes f^{(2)}, R_1\otimes S(R_1)\otimes R_2\otimes R_2\right\rangle\phi^{(3)} = \left\langle\phi^{(1)}\otimes S(f), R^{-1}R\right\rangle\phi^{(2)}\ 
\end{equation}
since $R^{-1}R=1$, the equality with the right-hand side 
follows from the properties of the counit.
\paragraph{Case 2: $\tau=1$}
In this case the left side of (\ref{proposition1::eq1}) becomes
\begin{equation}
    \left\langle\phi^{(1)(2)}\otimes f^{(1)},\left(S\otimes S\right)(R)\right\rangle\left\langle\phi^{(2)}\otimes f^{(2)},\left(S\otimes id\right)(R)\right\rangle\phi^{(1)(1)}
\end{equation}
By collecting the coefficients and relabelling the Sweedler indices accordingly, this is equal to,
\begin{equation}
	\begin{split}
		  \left\langle\phi^{(2)}\otimes\phi^{(3)}\otimes f^{(1)}\otimes f^{(2)}, R_1\otimes R_1\otimes S(R_2)\otimes R_2\right\rangle\phi^{(1)} &= \left\langle\phi^{(2)}\otimes f, R^{-1}R\right\rangle\phi^{(1)}\ , \\
		  & = \epsilon^{*} (f) \phi
	\end{split}  
\end{equation}
\noindent where we have used $R^{-1}R=1$ and $\epsilon^{*}(f)=\langle f, 1\rangle$ and $\epsilon^{*}(\phi^{(1)})\phi^{(2)}=\epsilon^{*}(\phi^{(2)})\phi^{(1)}=\phi$, so this is equal to
the right-hand side of  Proposition
  (\ref{proposition1::eq1}).
\end{proof}
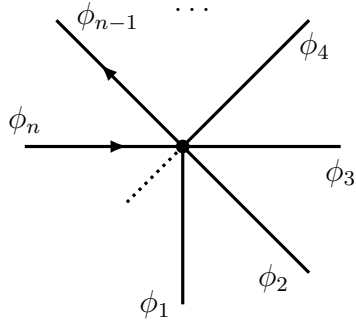
\begin{figure}[t]
    \centering
	\begin{tikzpicture}[>=latex,scale=0.7,>=latex,very thick,decoration={
	    markings,
	    mark= at position 0.65 with {\arrow{>}}}]
		\draw ({0},{0})--({0},{-3});
		\draw[fill] ({0},{0}) circle [radius=0.1];
		\draw[postaction={decorate}] ({-3},{0})--({0},{0});
		\draw ({0},{0})--({3},{0});
		\draw ({2.4)},-{2.4})--({0},{0});
		\draw ({2.4},{2.4})--({0},{0});
		\draw[postaction={decorate}] ({0},{0})--({-2.4},{2.4});
		\draw[very thick,dotted] ({-0.75*sqrt(2)},{-0.75*sqrt(2)}) -- ({0},{0});    
	    \node [left] at (0,-3) {$\phi_1$};
	    \node[left] at (2.25,-2.505) {$\phi_2$};
	    \node[below] at (3,0) {$\phi_3$};
	    \node[below] at (2.505,2.4) {$\phi_4$};
	    \node[right] at (-0.4,2.6) {$\ldots$};
	    \node[right] at (-2.25,2.475) {$\phi_{n-1}$};
	    \node[above] at (-3,0) {$\phi_{n}$};
	\end{tikzpicture}
    \caption{Here we display an example of an $n$-valent local vertex neighbourhood with the last two edges having consecutive directions. All the other edges can have an arbitrary orientation, which is why we do not provide any definite direction for them.
	 }
    \label{fig::propositionVertex}
\end{figure}
\noindent This proposition can be used to prove the next
\begin{proposition}
\label{proposition::2}
Consider a local vertex neighbourhood  in Figure \ref{fig::propositionVertex} with the braided tensor product defined in \ref{eq::BraidedTensProduct}.  
For $f\in K^{*}$  the following identity holds:
\begin{equation}
    \label{eq::proposition2}
    \begin{split}
    &(f^{(1)}\otimes f^{(2)})_{n-1, n}*(\phi_1 \otimes\ldots\otimes\phi_{n-1}\otimes\phi_{n})_{1,2,\ldots,n}=\\\
    &(\phi_1\otimes\ldots\otimes f^{(1)}\cdot\phi_{n-1}\otimes f^{(2)}\cdot\phi_{n})_{1,2,\ldots,n}
    \end{split}
\end{equation}
where $\cdot$ is the canonical product in $K^{*}$.
\end{proposition}
\begin{proof}
We will prove the proposition by induction over $n$, the number of edges around the local vertex neighbourhood.  
\paragraph{Case $n=3$:}
We can consider edge $(1)$ to have an arbitrary direction, edge $(2)$ to be outgoing and edge $(3)$ to be incoming. 
\\
Using the braided tensor product defined in Equation \ref{eq::BraidedTensProduct} we have that
\begin{equation}
\begin{split}
&(f^{(1)}\otimes f^{(2)})_{2, 3}*(\phi_1\otimes\phi_2\otimes\phi_3)_{1,2,3}=\\
&
\langle\phi_{2}^{(2)}\otimes f^{(2)}, R\rangle\langle\phi_{1}^{(2-\tau_1)(1+\tau_1)}\otimes f^{(3)},(S^{\tau_1}\otimes S)(R)\rangle\langle\phi_{1}^{(1+\tau_1)}\otimes f^{(4)},(S^{\tau_1}\otimes id)(R) \, \rangle\\%%%
&\langle\phi_{2}^{(3)}\otimes f^{(5)},(S\otimes id)(R)\rangle \left(\phi_{1}^{(2-\tau_1)(2-\tau_1)}\otimes f^{(1)}\phi_{2}^{(1)}\otimes f^{(6)}\phi_{3}\right)\ .
\end{split}
\end{equation}
\noindent Then using the result in Proposition \ref{proposition::1} on $\phi_{1}$, we get
\begin{equation}
\begin{split}
&(f^{(1)}\otimes f^{(2)})_{2, 3}\cdot(\phi_1\otimes\phi_2\otimes\phi_3)_{1,2,3} = \\
&\langle\phi_{2}^{(2)}\otimes f^{(2)}, R\rangle\langle\phi_{2}^{(3)}\otimes f^{(3)},(S\otimes id)(R)\rangle \left(\phi_{1}\otimes f^{(1)}\phi_{2}^{(1)}\otimes f^{(4)}\phi_{3}\right)
\end{split}
\end{equation}
now using that $(S\otimes id)(R)=R^{-1}$, we have
\begin{equation}
(f^{(1)}\otimes f^{(2)})_{2, 3}\cdot(\phi_1\otimes\phi_2\otimes\phi_3)_{1,2,3}=\left(\phi_{1}\otimes f^{(1)} \cdot \phi_{2}\otimes f^{(2)} \cdot \phi_{3}\right)\ ,
\end{equation}
which is what we needed to prove.

\paragraph{Inductive step:} Suppose equation (\ref{eq::proposition2}) holds for a local vertex neighbourhood with $n-1$ edges (whose last edge and second last edge orientations are still the same as the ones showed in Figure \ref{fig::propositionVertex}). With this assumption, we need now to prove the statement for a local vertex neighbourhood with $n$ edges. From direct computation we have
\begin{equation}
\begin{split}
&(f^{(1)}\otimes f^{(2)})_{n-1, n}*(\phi_1\otimes\phi_2\otimes\ldots\otimes\phi_{n-1}\otimes\phi_{n})_{1,2,\ldots, n}=\\
&\langle\phi_{1}^{(2-\tau_1)(1+\tau_1)}\otimes f^{(n)},(S^{\tau_1}\otimes S)(R)\rangle
\langle\phi_{2}^{(2-\tau_2)(1+\tau_2)}\otimes f^{(n-1)}, (S^{\tau_2}\otimes S)(R)\rangle\ldots\\
&\langle\phi_{n-2}^{(2-\tau_{n-2})(1+\tau_{n-2})}\otimes f^{(3)},(S^{\tau_{n-2}}\otimes S)(R)\rangle\langle\phi_{n-1}^{(2)}\otimes f^{(2)},R\rangle\\
&\langle\phi_{1}^{(1+\tau_1)}\otimes f^{(n+1)},(S^{\tau_1}\otimes id)(R)\rangle\langle\phi_{2}^{(1+\tau_2)}\otimes f^{(n+2)},(S^{\tau_2}\otimes id)(R)\rangle\ldots\\
&\langle\phi_{n-1}^{(3)}\otimes f^{(2n-1)},(S\otimes id)(R)\rangle
\\
&  \left(\phi_{1}^{(2-\tau_1)(2-\tau_1)}\otimes\ldots\otimes \phi_{n-2}^{(2-\tau_{n-2})(2-\tau_{n-2})}\otimes  f^{(1)}\phi_{n-1}^{(1)}\otimes f^{(2n)}\phi_{n}\right)_{1,2,\ldots, n}\ .
\end{split}
\end{equation}
Using  Proposition \ref{proposition::1} on $\phi_1$ we get
\begin{equation}
\begin{split}
&(f^{(1)}\otimes f^{(2)})_{n-1, n}\cdot(\phi_1\otimes\phi_2\otimes\ldots\otimes\phi_{n-1}\otimes\phi_{n})_{1,2,\ldots,n}=\\
&\langle\phi_{2}^{(2-\tau_2)(1+\tau_2)}\otimes f^{(n-1)}, (S^{\tau_2}\otimes S)(R)\rangle\ldots\\
&\langle\phi_{n-2}^{(2-\tau_{n-2})(1+\tau_{n-2})}\otimes f^{(3)},(S^{\tau_{n-2}}\otimes S)(R)\rangle\langle\phi_{n-1}^{(2)}\otimes f^{(2)},R\rangle\\
&\langle\phi_{2}^{(1+\tau_2)}\otimes f^{(n)},(S^{\tau_2}\otimes id)(R)\rangle\ldots\langle\phi_{n-1}^{(3)}\otimes f^{(2n-3)},(S\otimes id)(R)\rangle\\
&\left(\phi_{1}\otimes\phi_2^{(2-\tau_2)(2-\tau_2)}\otimes\ldots\otimes \phi_{n-2}^{(2-\tau_{n-2})(2-\tau_{n-2})}\otimes  f^{(1)}\phi_{n-1}^{(1)}\otimes f^{(2n-2)}\phi_{n}\right)_{1,2,\ldots,n}\ .
\end{split}
\end{equation}
 we can see that the equation above is actually equivalent to
\begin{equation}
    \begin{split}
       &(f^{(1)}\otimes f^{(2)})_{n-1, n}*(\phi_1\otimes\phi_2\otimes\ldots\otimes\phi_{n-1}\otimes\phi_{n})_{1,2,\ldots,n}=\\
       &(\phi_1)_1*(f^{(1)}\otimes f^{(2)})_{n-1, n}*(\phi_2\otimes\ldots\otimes\phi_{n-1}\otimes\phi_{n})_{2,3,\ldots,n}\ . 
    \end{split}
\end{equation}
We now recognise the same product on a local vertex neighbourhood with one edge less, and because of the inductive hypothesis and associativity we have
\begin{equation}
    \begin{split}
       &(f^{(1)}\otimes f^{(2)})_{n-1, n}*(\phi_1\otimes\phi_2\otimes\ldots\otimes\phi_{n-1}\otimes\phi_{n})_{2,3,\ldots,n}=\\
       &(\phi_1)_1*(\phi_2\otimes\ldots\otimes f^{(1)}\phi_{n-1}\otimes f^{(2)}\phi_{n})_{2,3,\ldots,n}\ ,
    \end{split}    
\end{equation}
which is the same as (\ref{eq::proposition2}).
\end{proof}
\noindent This proposition will be used repeatedly in the following. 
As already explained in Section \ref{sec::BraidedTensor}, the holonomy associated with a function $\alpha$ for an counterclockwise plaquette with edges $1,2, \ldots n$, counting from the cilia, acts on the algebra of functions on the ciliated graph as the product by the element
\begin{equation}
    (S(\alpha^{(1)})\otimes S(\alpha^{(2)})\otimes\cdots\otimes S(\alpha^{(n)}))_{1,2,\ldots,n}\ .
\end{equation}
Proposition \ref{proposition::2} then means that all the non-trivial aspects introduced by the $R$ matrix show up at the local vertex neighbourhood in the plaquette with the cilia pointing inwards.  At all the other vertexes the holonomy simply acts with the canonical product
 on $K^{*}$.
Because of this, when trying to find the general structure of the plaquette operator, it is sufficient to do it for the plaquette shown in Figure \ref{fig::simplest_plaquette}
\footnote{Note that this is not the simplest plaquette that can be used, but this case better shows the type of calculations involved in computing plaquette operators.}
\begin{figure}[t]
    \centering
	\begin{tikzpicture}[>=latex,scale=0.6,>=latex,very thick,decoration={
	    markings,
	    mark= at position 0.55 with {\arrow{>}}}]  
	    \draw[postaction={decorate}] (0, 0) to [out=45,in=-45] (0, 6);
	    \draw[postaction={decorate}] (0, 6) to [out=225,in=135] (0, 0); 	
		\draw[very thick,dotted] ({0},{-1}) -- ({0},{0});
	    \draw[fill] ({0},{0}) circle [radius=0.1];
	    \draw[fill] ({0},{6}) circle [radius=0.1];
		\draw[very thick,dotted] ({0},{6}) -- ({0},{5});
	    \draw ({0},{0})--({-2},{-2});
	    \draw ({0},{0})--({2},{-2});
	    \draw ({0},{6})--({-2},{8});
	    \draw ({0},{6})--({2},{8});    
	    \node[right] at (0,0) {w};
	    \node[right] at (0,6) {v};
	    \node[right] at (1.5,3) {$\phi_2$};
	    \node[left] at (-1.5,3) {$\phi_1$};	
		\node[right] at (1,7) {$\phi_3$};
		\node[right] at (-2.1,7) {$\phi_4$};
		\node[right] at (1,-0.5) {$\phi_5$};
		\node[right] at (-2.1,-0.5) {$\phi_6$};
		\node[right] at (-0.5,3) {$p$};
	\end{tikzpicture}
    \caption{Simple plaquette.}
    \label{fig::simplest_plaquette}
\end{figure}
We can now establish the following relation between the plaquette operator computed using the braided tensor product, which we denote by $\mathcal{B}_{p}^{f}$ and the plaquette operator from conventional Kitaev quantum double models, which we denote by $B_{p}^{f}$.
\begin{proposition}
	\label{eq::prop3}
Consider the ribbon graph shown in Figure \ref{fig::simplest_plaquette}. The plaquette operator using the braided tensor product is given by,
\begin{equation}
    \label{eq::plaquetteShape}
    \begin{split}
        & \mathcal{B}_{p}^{f}(\phi_{1},\ \phi_{2},\ \phi_{3},\ \phi_{4},\ \phi_{5},\ \phi_{6}) =\\
        &\langle\phi_{2}^{(1)}\otimes f^{(1)},R\rangle\langle\phi_{3}^{(1+\tau_3)}\otimes f^{(2)},(S^{\tau_3}\otimes id)(R)\rangle\langle\phi_{4}^{(1+\tau_4)}\otimes f^{(3)},(S^{\tau_4}\otimes id)(R)\rangle\\
        &\langle\phi_{1}^{(2)}\otimes f^{(4)},R^{-1}\rangle B_{p}^{f^{(5)}}\left(\phi_{1}^{(1)},\ \phi_{2}^{(2)},\ \phi_{3}^{(2-\tau_3)},\ \phi_{4}^{(2-\tau_4)},\ \phi_{5},\ \phi_{6}\right)\ 
    \end{split}
\end{equation}
where $\tau_{i}$ labels the directedness of the edge which has $\phi_{i}$ assigned to it. 
\end{proposition}
\begin{proof}
As mentioned the plaquette operator on plaquette $p$ for Hopf algebra gauge theory is given by
\begin{equation}
    \begin{split}
        &\mathcal{B}_{p}^{f}(\phi_{1},\ \phi_{2},\ \phi_{3},\ \phi_{4},\ \phi_{5},\ \phi_{6}) = \\
        &\big(S(f^{(1)})\otimes S(f^{(2)})\big)_{1,2}*(\phi_{1}\otimes\phi_{2}\otimes\phi_{2}\otimes\phi_{4}\otimes\phi_{5}\otimes\phi_{6})_{1,2,\ldots,6}
    \end{split}
\end{equation}
where the numbering of edges is the same as the subscript labelling on the functions.
In order to compute this product, we need to consider the splitting of the graph into local vertex neighbourhoods.  
Algebraically this translates into the use of $G^{*}$ map, given in Equation \ref{eq::GMaps}. 
This is represented pictorially in Figure \ref{fig::GMap}. 
We therefore need to consider two local vertex neighbourhoods associated with the vertices $v$ and $w$ as shown in the Figure \ref{fig::simplest_plaquette_local_vertex_neighborhood}. 
For the vertex $w$, we can use Proposition \ref{proposition::2} and the product reduces to the canonical product in the algebra of functions
\begin{equation}
    \begin{split}
    &(S(f^{(3)})\otimes S(f^{(2)}))_{w_2, w_3}*(\phi_5\otimes\phi_{2}^{(2)}\otimes\phi_{1}^{(1)}\otimes\phi_{6})_{w_1,w_2,w_3,w_4}=\\
    &(\phi_5\otimes S(f^{(3)})\phi_{2}^{(2)}\otimes S(f^{(2)})\phi_{1}^{(1)}\otimes\phi_{6})_{w_1,w_2,w_3,w_4}\ .
    \end{split}
\end{equation}
On the local vertex neighbourhood associated with $v$ we instead have
\begin{equation}
    \begin{split}
        &(S(f^{(4)})\otimes S(f^{(1)}))_{v_1, v_4}* (\phi_2^{(1)}\otimes\phi_{3}\otimes\phi_{4}\otimes\phi_{1}^{(2)})_{v_1,v_2,v_3,v_4}=\\
        &\langle\phi_{2}^{(1)(1)}\otimes S(f^{(1)})^{(2)},R^{-1}\rangle\langle\phi_{3}^{(1+\tau_3)}\otimes S(f^{(1)})^{(1)(2)},(S^{\tau_3}\otimes S)(R)\rangle\\
        &\langle\phi_{4}^{(1+\tau_4)}\otimes S(f^{(1)})^{(1)(1)(2)},(S^{\tau_4}\otimes S)(R)\rangle\langle\phi_{1}^{(2)(2)}\otimes S(f^{(1)})^{(1)(1)(1)(2)},R\rangle\\
        &(S(f^{(4)})\phi_{2}^{(1)(2)}\otimes\phi_{3}^{(2-\tau_3)}\otimes\phi_{4}^{(2-\tau_4)}\otimes S(f^{(1)})^{(1)(1)(1)(1)}\phi_{1}^{(2)(1)})_{v_1,v_2,v_3,v_4}\ ,
    \end{split}    
\end{equation}
\begin{figure}[H]
    \centering
    \begin{minipage}{0.6\textwidth}
		\begin{tikzpicture}[>=latex,scale=0.6,>=latex,very thick,decoration={
		    markings,
		    mark= at position 0.55 with {\arrow{>}}}]   
		    \draw[postaction={decorate}] (0, 0) to [out=45,in=270] (1.3,3);
		    \draw[postaction={decorate}] (1.3, 3) to [out=90,in=-45] (0,6);
		    \draw[postaction={decorate}] (0, 6) to [out=225,in=90] (-1.3, 3);
		    \draw[postaction={decorate}] (-1.3, 3) to [out=270,in=135] (0, 0);
			\draw[very thick,dotted] ({0},{-1}) -- ({0},{0});
		    \draw[fill] ({0},{0}) circle [radius=0.1];
		    \draw[fill] ({0},{6}) circle [radius=0.1];
			\draw[very thick,dotted] ({0},{6}) -- ({0},{5});
			 \node[right] at (0,0) {w};
		    \node[right] at (0,6) {v};
		    \node[right] at (-10,6) {$(S(f^{(4)})\otimes S(f^{(1)}))_{v_{1}, v_{4}}$};
		    \draw[fill] ({-2.5},{6}) circle [radius=0.03];
		    \node[right] at (1.2,2) {$\phi_2^{(2)}$};
		    \node[right] at (1.1,5) {$\phi_2^{(1)}$};
		    \node[left] at (-1.,5) {$\phi_1^{(2)}$};
		    \node[left] at (-1.2,2) {$\phi_1^{(1)}$};
		    \draw[fill=white] ({1.3},{3}) circle [radius=0.1];
		    \draw[fill=white] ({-1.3},{3}) circle [radius=0.1];    
		    \draw ({0},{0})--({-2},{-2});
		    \draw ({0},{0})--({2},{-2});
		    \draw ({0},{6})--({-2},{8});
		    \draw ({0},{6})--({2},{8});
			\node[right] at (1,7) {$\phi_3$};
			\node[right] at (-2.1,7) {$\phi_4$};
			\node[right] at (1,-0.5) {$\phi_5$};
			\node[right] at (-2.1,-0.5) {$\phi_6$};
			\node[right] at (-0.5,3) {$p$};
		    \node[right] at (-10,0) {$(S(f^{(3)})\otimes S(f^{(2)}))_{w_{2}, w_{3}}$};
		    \draw[fill] ({-2.5},{0}) circle [radius=0.03];
		\end{tikzpicture}
		
    \end{minipage}
    \begin{minipage}{0.3\textwidth}
    \end{minipage}
    \caption{Graphical representation of the product between the two local vertex neighbourhood in $v$ and $w$}
    \label{fig::simplest_plaquette_local_vertex_neighborhood}
\end{figure}
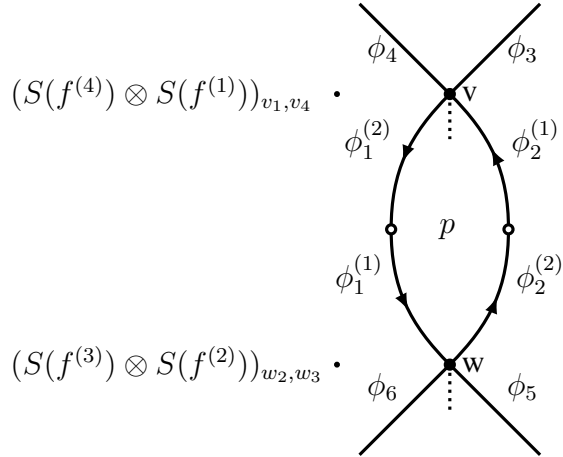
\noindent where we have used (\ref{eq::BraidedTensProduct}).
\\Since elements living at different local vertex neighbourhoods commute, we can see that overall the two products can be written as
\begin{equation*}
    \fontsize{11.3}{13.56}
    \begin{split}
        &G^{*}\left(S(f^{(4)})\otimes S(f^{(3)})\otimes S(f^{(2)})\otimes S(f^{(1)})\right)*\\
        &G^{*}\left( \phi_{1}^{(2)}\otimes\phi_{1}^{(1)}\otimes\phi_2^{(2)}\otimes\phi_2^{(1)}\otimes\phi_{3}\otimes\phi_{4}\otimes\phi_{5}\otimes\phi_{6}\right)=\\
        &\langle\phi_{2}^{(1)}\otimes f^{(1)},R\rangle\langle\phi_{3}^{(1+\tau_3)}\otimes f^{(2)},(S^{\tau_3}\otimes id)(R)\rangle\\
        &\langle\phi_{4}^{(1+\tau_4)}\otimes f^{(3)},(S^{\tau_4}\otimes id)(R)\rangle\langle\phi_{1}^{(3)}\otimes f^{(4)},R^{-1}\rangle\\
        &(S(f^{(5)})\phi_{1}^{(2)}\otimes S(f^{(6)})\phi_{1}^{(1)}\otimes S(f^{(7)})\phi_{2}^{(3)}\otimes(S(f^{(8)})\phi_{2}^{(2)}\otimes\phi_{3}^{(2-\tau_3)}\otimes\phi_{4}^{(2-\tau_4)}\otimes \phi_{5}\otimes\phi_{6})_{\mathbf{e}}\ .
    \end{split}    
\end{equation*}
where $\mathbf{e}= \{v_4,w_3,w_2,v_1,v_2,v_3,w_1,w_6\}$ 
and for compactness, on the first two lines, we have omitted the subscript labelling edges, since from the above discussion it should be clear which element is associated with which edge. In the last line, we can now recognise the $G^{*}$ map defined in
  Equation (\ref{eq::GMaps}) and we can therefore write
\begin{equation}
    \fontsize{11.3}{13.56}
    \begin{split}
        &G^{*}\left(S(f^{(4)})\otimes S(f^{(3)})\otimes S(f^{(2)})\otimes S(f^{(1)})\right)*\\
        &G^{*}\left( (\phi_{1}^{(2)}\otimes\phi_{1}^{(1)}\otimes\phi_2^{(2)}\otimes\phi_2^{(1)}\otimes\phi_{3}\otimes\phi_{4}\otimes\phi_{5}\otimes\phi_{6})\right)=\\
        &\langle\phi_{2}^{(1)}\otimes f^{(1)},R\rangle\langle\phi_{3}^{(1+\tau_3)}\otimes f^{(2)},(S^{\tau_3}\otimes id)(R)\rangle\\
        &\langle\phi_{4}^{(1+\tau_4)}\otimes f^{(3)},(S^{\tau_4}\otimes id)(R)\rangle\langle\phi_{1}^{(2)}\otimes f^{(4)},R^{-1}\rangle\\
        &G^{*}\left(B_{p}^{f^{(3)}}\left(\phi_{1}^{(1)},\ \phi_{2}^{(2)},\ \phi_{3}^{(2-\tau_3)},\ \phi_{4}^{(2-\tau_4)},\ \phi_{5},\ \phi_{6}\right)\right)\ ,
    \end{split}    
\end{equation}
So by linearity and from the fact that $G^{*}$ is an homomorphism, we can recognise Equation (\ref{eq::plaquetteShape}) of Proposition \ref{eq::prop3}.
\end{proof}
\noindent It is easy to generalise the above result for a plaquette with an arbitrary number of edges. In particular for the plaquette without any external edges, as the one shown in Figure \ref{fig::EvenSimplerPlaquette}, we get
\begin{equation}
    \label{eq::plaquetteShapeTwo}
    \begin{split}
        & \mathcal{B}_{p}^{f}(\phi_{1},\ \phi_{2}) =\langle\phi_{2}^{(1)}\otimes f^{(1)},R\rangle\langle\phi_{1}^{(2)}\otimes f^{(2)},R^{-1}\rangle B_{p}^{f^{(3)}}(\phi_{1}^{(1)},\ \phi_{2}^{(2)})\ .
    \end{split}
\end{equation}
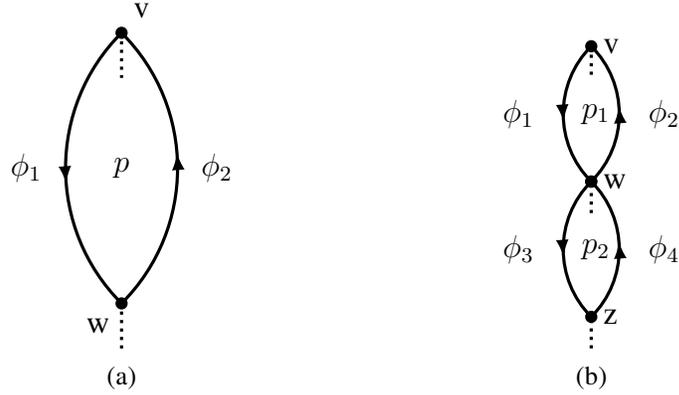
\begin{figure}[t]
\centering
   \subfloat[]
   % {\input{Paper_Pictures/even_simpler_plaquette}
 {  \begin{tikzpicture}[>=latex,scale=0.6,>=latex,very thick,decoration={
       markings,
       mark= at position 0.55 with {\arrow{>}}}]    
       \draw[postaction={decorate}] (0, 0) to [out=45,in=-45] (0, 6);
       \draw[postaction={decorate}] (0, 6) to [out=225,in=135] (0, 0);	
   	\draw[very thick,dotted] ({0},{-1}) -- ({0},{0});
       \draw[fill] ({0},{0}) circle [radius=0.1];
       \draw[fill] ({0},{6}) circle [radius=0.1];
   	\draw[very thick,dotted] ({0},{6}) -- ({0},{5});  
       \node[left] at (0,-0.5) {w};
       \node[right] at (0,6.5) {v};
       \node[right] at (1.5,3) {$\phi_2$};
       \node[left] at (-1.5,3) {$\phi_1$};
       \node at (0, 3) {$p$};	
   \end{tikzpicture}  
   \label{fig::EvenSimplerPlaquette} }
    \hspace{3cm}
    \subfloat[] {
	\begin{tikzpicture}[>=latex,scale=0.6,>=latex,very thick,decoration={
	    markings,
	    mark= at position 0.55 with {\arrow{>}}}]   
	    \draw[postaction={decorate}] (0, 0) to [out=45,in=-45] (0, 3);
	    \draw[postaction={decorate}] (0, 3) to [out=225,in=135] (0, 0);
		\draw[very thick,dotted] ({0},{-0.7}) -- ({0},{0});
	    \draw[fill] ({0},{0}) circle [radius=0.1];
	    \draw[fill] ({0},{3}) circle [radius=0.1];
		\draw[very thick,dotted] ({0},{3}) -- ({0},{2.3});   
	    \node[right] at (0,0) {w};
	    \node[right] at (0,3) {v};
	    \node[right] at (0,-3) {z};
	    \node[right] at (1,1.5) {$\phi_2$};
	    \node[left] at (-1,1.5) {$\phi_1$};
	    \node at (0.1,1.5) {$p_1$};
	    \draw[postaction={decorate}] (0, -3) to [out=45,in=-45] (0, 0);
	    \draw[postaction={decorate}] (0, 0) to [out=225,in=135] (0, -3);	
		\draw[very thick,dotted] ({0},{-3.7}) -- ({0},{-3});
	    \draw[fill] ({0},{-3}) circle [radius=0.1];
	    \draw[fill] ({0},{0}) circle [radius=0.1];
	    \node[right] at (1,-1.5) {$\phi_4$};
	    \node[left] at (-1,-1.5) {$\phi_3$};
		\node at (0.1,-1.5) {$p_2$};	
	\end{tikzpicture}
	\label{fig::CommutationPlaquettes}}
    \caption{In figure (a) we show a plaquette without any external edges attached. In figure (b) the relevant graph to compute the commutation between plaquette operators .}
\end{figure} 

We can now start considering the commutation relations of two plaquette operators. We can distinguish between two cases, one in which the two plaquettes have coinciding starting vertices and one in which the starting vertices differ. We will consider the latter case first. 
\\From the preceding results it is clear that the relevant graph to consider when enquiring about the commutation relation between plaquettes at different vertexes is the one given in Figure \ref{fig::CommutationPlaquettes}.
Similarly to Kitaev models, we can see that in this case plaquette operators commute. 
\begin{proposition}
Consider the graph shown in Figure \ref{fig::CommutationPlaquettes} and consider the two plaquette operators at vertexes $v$ and $w$, which act with functions $f,g \in K^{*}$, respectively. Then we have
\begin{equation}
    \label{eq::commutationPlaquette}
     \mathcal{B}_{p_{1}}^{f}  \mathcal{B}_{p_{2}}^{g}=  \mathcal{B}_{p_{2}}^{g} \mathcal{B}_{p_{1}}^{f}
\end{equation}
\end{proposition}
\begin{proof}
We start by considering the left hand of the equation (\ref{eq::commutationPlaquette}) when acting on the general elements $\phi_1, \phi_2, \phi_3,\ \phi_4 \in K^{*}$, as given in Figure \ref{fig::CommutationPlaquettes}. By using (\ref{eq::plaquetteShape}) and (\ref{eq::plaquetteShapeTwo}) we have
\begin{equation}
    \label{eq::proofPropositionCommutation}
    \begin{split}
        &\mathcal{B}_{p_{1}}^{f}  \mathcal{B}_{p_{2}}^{g}(\phi_{1}, \phi_{2}, \phi_{3}, \phi_{4})=\\
        &\langle\phi^{(1)}_{4}\otimes g^{(1)},R\rangle\langle\phi^{(3)}_{2}\otimes g^{(2)},R^{-1}\rangle\langle\phi^{(1)}_{1}\otimes g^{(3)},R\rangle\langle\phi^{(2)}_{3}\otimes g^{(4)},R^{-1}\rangle\\
        &\langle\phi^{(1)}_{2}\otimes f^{(1)},R\rangle\langle\phi^{(3)}_{1}\otimes f^{(2)},R^{-1}\rangle\\
        &(S(f^{(3)})\phi_{1}^{(2)}\otimes S(f^{(4)})\phi_2^{(2)}\otimes S(g^{(5)})\phi_{3}^{(1)}\otimes S(g^{(6)})\phi_4^{(2)})\ .
    \end{split}
\end{equation}
We will now show that the right hand side of (\ref{eq::commutationPlaquette}), when acting on $\phi_1, \phi_2, \phi_3,\ \phi_4 \in K^{*}$, turns out to be the same as (\ref{eq::proofPropositionCommutation}). We have in fact
\begin{equation}
    \begin{split}
        &\mathcal{B}_{p_{2}}^{g} \mathcal{B}_{p_{1}}^{f}(\phi_{1}, \phi_{2}, \phi_{3}, \phi_{4})=\\
        &\langle\phi_{2}^{(1)}\otimes f^{(1)},R\rangle\langle\phi_{1}^{(3)}\otimes f^{(2)},R^{-1}\rangle\\
        &\langle\phi_{4}^{(1)}\otimes g^{(1)},R\rangle\langle S(f^{(4)})\phi_{1}^{(1)}\otimes g^{(3)},R\rangle\langle S(f^{(5)})\, \phi_{2}^{(3)}\otimes g^{(2)},R^{-1}\rangle\\
        &\langle\phi_{3}^{(2)}\otimes g^{(4)},R^{-1}\rangle(S(f^{(3)})\, \phi_{1}^{(2)}\otimes S(f^{(6)})\phi_2^{(2)}\otimes S(g^{(5)})\, \phi_{3}^{(1)}\otimes S(g^{(6)})\, \phi_4^{(2)})\ .
    \end{split}
\end{equation}
Using that
\begin{equation}
    (\Delta\otimes id)(R) = R_{13}R_{23}, \qquad (\Delta\otimes id)(R^{-1}) = (id\otimes id\otimes S)(R_{13}R_{23}),
\end{equation}
we can dualize the products in $K^{*}$ into coproducts of $K$ and have as a result
\begin{equation}
    \begin{split}
        &\mathcal{B}_{p_{2}}^{g} \mathcal{B}_{p_{1}}^{f}=\\
        &\langle\phi_{2}^{(1)}\otimes f^{(1)},R\rangle\langle\phi_{1}^{(3)}\otimes f^{(2)},R^{-1}\rangle\langle\phi_{4}^{(1)}\otimes g^{(1)},R\rangle\\
        &\langle S(f^{(4)})\otimes\phi_{1}^{(1)}\otimes g^{(3)},R_{13}R_{23}\rangle\langle S(f^{(5)})\otimes\phi_{2}^{(3)}\otimes g^{(2)},(id\otimes id\otimes S)(R_{13}R_{23})\rangle\\
        &\langle\phi_{3}^{(2)}\otimes g^{(4)},R^{-1}\rangle(S(f^{(3)})\phi_{1}^{(2)}\otimes S(f^{(6)})\phi_2^{(2)}\otimes S(g^{(5)}) \, \phi_{3}^{(1)}\otimes S(g^{(6)}) \, \phi_4^{(2)})\ .
    \end{split}
\end{equation}
The $f$'s and $g$'s can now be regrouped, note in fact that we can write
\begin{equation}
    \begin{split}
        &\langle f^{(4)}\otimes f^{(5)}\otimes g^{(2)}\otimes g^{(3)}, S(R_1')\otimes S(R_3')\otimes S(R_4'')S(R_3'')\otimes R_1''R_2''\rangle=\\
        &\langle f^{(4)}\otimes g^{(2)}, S(R_3'R_1')\otimes S(R_4'')S(R_3'')R_1''R_2''\rangle=\\
        &\langle f^{(4)}\otimes g^{(2)}, 1\otimes S(R_4'')R_2''\rangle=\epsilon^{*}(f^{(4)})\langle g^{(2)}\otimes g^{(3)}, S(R_4'')\otimes R_2''\rangle\ , 
    \end{split}
\end{equation}
where the subscripts in the $R$ matrices are used to avoid confusion between the elements of different $R$ matrices and in the last step we have used that $(id\otimes S)(R)=R^{-1}$.
\\Thus, using this last result, by appropriately regrouping terms, we finally have
\begin{equation}
    \begin{split}
        &\mathcal{B}_{p_{2}}^{g} \mathcal{B}_{p_{1}}^{f}=\\
        &\langle\phi_{2}^{(1)}\otimes f^{(1)},R\rangle\langle\phi_{1}^{(3)}\otimes f^{(2)},R^{-1}\rangle\langle\phi_{4}^{(1)}\otimes g^{(1)},R\rangle\langle \phi_{1}^{(1)}\otimes g^{(3)},R\rangle\langle \phi_{2}^{(3)}\otimes g^{(2)},R^{-1}\rangle\\
        &\langle\phi_{3}^{(2)}\otimes g^{(4)},R^{-1}\rangle(S(f^{(3)})\phi_{1}^{(2)}\otimes S(f^{(6)})\phi_2^{(2)}\otimes S(g^{(5)})\phi_{3}^{(1)}\otimes S(g^{(6)})\phi_4^{(2)})\ ,
    \end{split}
\end{equation}
which is the same as (\ref{eq::proofPropositionCommutation}) and therefore proves the claim on Figure \ref{fig::CommutationPlaquettes}.
\end{proof}
\noindent As already mentioned the previous result holds for more general plaquettes and the proof is similar, even if more notationally involved.

Before continuing with the treatment of commutation relations between plaquettes, we need the following technical proposition:
\begin{proposition}
Consider $f,\ g,\ \phi\in K^*$ and the $R$ matrix of $K$. Then, given $\tau=0,\ 1$, the following equation holds
\begin{equation}
\label{eq::prooLemmaProdBf}
    \begin{split}
  &  \langle\phi^{(1+\tau)}\otimes g, (S^{\tau}\otimes id)(R)\rangle\langle\phi^{(2-\tau)(1+\tau)}\otimes f, (S^{\tau}\otimes id)(R)\rangle\phi^{(2-\tau)(2-\tau)}=\\ 
	&\langle\phi^{(1+\tau)}\otimes fg, (S^{\tau}\otimes id)(R)\rangle\phi^{(2-\tau)}
	\end{split}
	\end{equation}
\end{proposition}
\begin{proof}
    We can prove the formula for the two cases directly.
    \paragraph{Case $\tau=0$}
    In this instance we have
    \begin{equation}
        \begin{split}
        &\langle\phi^{(1)}\otimes g, R\rangle\langle\phi^{(2)(1)}\otimes f, R\rangle\phi^{(2)(2)}=\\
        &\langle\phi^{(1)}\otimes\phi^{(2)}\otimes f\otimes g, R_1'\otimes R_2'\otimes R_2''\otimes R_1''\rangle\phi^{(3)}=\langle \phi^{(1)}\otimes f\otimes g, R_{13}R_{12}\rangle\phi^{(2)}\ ,
        \end{split}
    \end{equation}
    where we have used the dualisation of the coproduct. Now we can use $(id\otimes\Delta)(R)=R_{13}R_{12}$ to get
    \begin{equation}
        \langle \phi^{(1)}\otimes f\otimes g, R_{13}R_{12}\rangle\phi^{(2)}=\langle \phi^{(1)}\otimes f\otimes g, (id\otimes\Delta)(R)\rangle\phi^{(2)}=\langle \phi^{(1)}\otimes fg, R\rangle\phi^{(2)}\ ,
    \end{equation}
    which proves (\ref{eq::prooLemmaProdBf}).
    \paragraph{Case $\tau=1$}
    In this case, we have
    \begin{equation}
        \begin{split}
        &\langle\phi^{(2)}\otimes g, R\rangle\langle\phi^{(1)(2)}\otimes f, (S\otimes id)(R)\rangle\phi^{(1)(1)}=\\
        &\langle\phi^{(2)}\otimes\phi^{(3)}\otimes f\otimes g, R_1'\otimes R_2'\otimes R_2''\otimes R_1''\rangle\phi^{(1)}=\langle \phi^{(2)}\otimes f\otimes g, R_{13}R_{12}\rangle\phi^{(1)}\ .
        \end{split}
    \end{equation}
    As before we can use the relation $(id\otimes\Delta)(R)=R_{13}R_{12}$ to simplify the result:
    \begin{equation}
        \langle \phi^{(2)}\otimes f\otimes g, R_{13}R_{12}\rangle\phi^{(1)}=\langle \phi^{(2)}\otimes f\otimes g, (id\otimes\Delta)(R)\rangle\phi^{(1)}=\langle \phi^{(2)}\otimes f \cdot g, R\rangle\phi^{(1)}\ ,
    \end{equation}
    which ends our proof.
\end{proof}
We now have everything we needed in order to prove one of the main results, that is the product between plaquette operators at coinciding starting vertexes. 
\begin{figure}[t]
\centering
\begin{tikzpicture}[>=latex,scale=0.6,>=latex,very thick,decoration={
    markings,
    mark= at position 0.55 with {\arrow{>}}}]   
    \draw[postaction={decorate}] (0, 0) to [out=45,in=-45] (0, 6);
    \draw[postaction={decorate}] (0, 6) to [out=225,in=135] (0, 0); 	
    \draw[fill] ({0},{0}) circle [radius=0.1];
    \draw[fill] ({0},{6}) circle [radius=0.1];
	\draw[very thick,dotted] ({0},{6}) -- ({0},{5});
    \draw ({0},{6})--({-2},{8});
    \draw ({0},{6})--({2},{8});  
    \node[right] at (0,6) {v};
    \node[right] at (1.5,3) {$\phi_2$};
    \node[left] at (-1.5,3) {$\phi_1$};	
	\node[right] at (1,7) {$\phi_3$};
	\node[right] at (-2.1,7) {$\phi_4$};
	\node[right] at (-0.5,3) {$p$};
	\node[right] at (0,0) {w};
\end{tikzpicture}

\caption{Relevant plaquette used to understand the product between plaquette operators on the same plaquette with the braided tensor product.}
    \label{fig::Tadpole}
\end{figure}
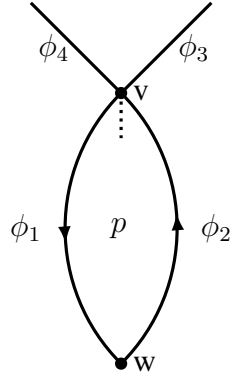
\begin{proposition}
Consider two plaquette operators acting on the same plaquette, by functions $f$ and $g$ respectively. Then the following relation holds
\begin{equation}
    \label{eq::propCommBfBg}
    \mathcal{B}_{p}^{f} \mathcal{B}_{p}^{g}=\langle g^{(1)}\otimes f^{(2)}, R\rangle\langle g^{(3)}\otimes f^{(1)}, R^{-1}\rangle \mathcal{B}_{p}^{f^{(3)}g^{(2)}}
\end{equation}
\end{proposition}
\begin{proof}
To prove the relation we will consider a ribbon graph as the one given in Figure \ref{fig::Tadpole}. As we have seen, vertexes in the plaquette without the plaquette cilia do not play any role in the plaquette operator, so proving the formula for a plaquette with only two vertexes and without external edges at $w$ is sufficient. 
\\The proof of the statement comes from a direct computation. Let's therefore write
\begin{equation}
\begin{split}
    &\mathcal{B}_{p}^{f}\mathcal{B}_{p}^{g}(\phi_{1},\phi_{2},\phi_{3},\phi_{4})=\\
    &\langle\phi_2^{(1)}\otimes g^{(1)}, R\rangle\langle\phi_3^{(1+\tau_3)}\otimes g^{(2)}, (S^{\tau_3}\otimes id)(R)\rangle\langle\phi_4^{(1+\tau_4)}\otimes g^{(3)}, (S^{\tau_4}\otimes id)(R)\rangle\\
    &\langle\phi_1^{(3)}\otimes g^{(4)},R^{-1}\rangle\langle S(g^{(5)})\phi_{1}^{(2)}\otimes f^{(4)}, R^{-1}\rangle\langle S(g^{(8)})\phi_2^{(2)}\otimes f^{(1)}, R\rangle\\
    &\langle\phi_3^{(2-\tau_3)(1+\tau_3)}\otimes f^{(2)}, (S^{\tau_3}\otimes id)(R)\rangle\langle\phi_4^{(2-\tau_4)(1+\tau_4)}\otimes f^{(3)}, (S^{\tau_4}\otimes id)(R)\rangle\\
    &(S(f^{(5)})S(g^{(6)})\phi_{1}^{(1)}\otimes S(f^{(6)})S(g^{(7)})\phi_{2}^{(3)}\otimes \phi_{3}^{(2-\tau_3)(2-\tau_3)}\otimes \phi_{4}^{(2-\tau_4)(2-\tau_4)})\ .
\end{split}
\end{equation}
Using Proposition \ref{eq::prop3} we can see that we can simplify some of the terms:
\begin{equation}
\label{eq::propositionPlaquettesEqToUse}
\begin{split}
    &\mathcal{B}_{p}^{f}\mathcal{B}_{p}^{g}(\phi_{1},\phi_{2},\phi_{3},\phi_{4})=\\
    &\langle\phi_2^{(1)}\otimes g^{(1)}, R\rangle\langle\phi_3^{(1+\tau_3)}\otimes f^{(2)}g^{(2)}, (S^{\tau_3}\otimes id)(R)\rangle\\
    &\langle\phi_4^{(1+\tau_4)}\otimes f^{(3)}g^{(3)}, (S^{\tau_4}\otimes id)(R)\rangle\langle\phi_1^{(3)}\otimes g^{(4)},R^{-1}\rangle\\
    &\langle S(g^{(5)})\phi_{1}^{(2)}\otimes f^{(4)}, R^{-1}\rangle\langle S(g^{(8)})\phi_2^{(2)}\otimes f^{(1)}, R\rangle\\
    &(S(f^{(5)})S(g^{(6)})\phi_{1}^{(1)}\otimes S(f^{(6)})S(g^{(7)})\phi_{2}^{(3)}\otimes \phi_{3}^{(2-\tau_3)}\otimes \phi_{4}^{(2-\tau_4)})\ .
\end{split}
\end{equation}
Using that $(\Delta\otimes id)(R)=R_{13}R_{23}$, we can write
\begin{equation*}
    \begin{split}
        &\langle S(g^{(5)})\phi_{1}^{(2)}\otimes f^{(4)}, R^{-1}\rangle=\langle S(g^{(5)})\otimes\phi_1^{(2)}\otimes S(f^{(4)(2)})\otimes S(f^{(4)(1)}), R'\otimes R_1'\otimes R''\otimes R_1''\rangle\\
    \end{split}
\end{equation*}
and
\begin{equation*}
    \begin{split}
        &\langle S(g^{(8)})\phi_2^{(2)}\otimes f^{(1)}, R\rangle=\langle S(g^{(8)})\otimes\phi_2^{(2)}\otimes f^{(1)(1)}\otimes f^{(1)(2)}, R'\otimes R_2'\otimes R''\otimes R_2''\rangle\ ,
    \end{split}
\end{equation*}
where the indices $1,\ 2$ are to keep track of the elements belonging to different $R$ matrixes. We can now regroup some of the coproducts in (\ref{eq::propositionPlaquettesEqToUse})
\begin{equation*}
    \begin{split}
        &\langle\phi_{1}^{(2)}\otimes\phi_{1}^{(3)}\otimes S(g^{(4)})\otimes S(f^{(4)(1)}), R_{2}'\otimes R' \otimes R''\otimes R_{2}''\rangle=\langle \phi_{1}^{(2)}\otimes S(f^{(4)(1)}g^{(4)}),R\rangle\\
        &\langle\phi_{2}^{(1)}\otimes\phi_{2}^{(2)}\otimes f^{(1)(2)}\otimes g^{(1)}, R'\otimes R_{4}' \otimes R_{4}''\otimes R''\rangle=\langle\phi_{2}^{(1)}\otimes f^{(1)(2)}g^{(1)},R\rangle\ ,
    \end{split}
\end{equation*}
where, again, we used $(id\otimes \Delta)(R)=R_{13}R_{12}$ . Therefore
\begin{equation}
    \begin{split}
       &\mathcal{B}_{p}^{f}\mathcal{B}_{p}^{g}(\phi_{1},\phi_{2},\phi_{3},\phi_{4})=\\
       &\langle g^{(5)}\otimes f^{(6)}, R\rangle\langle g^{(8)}\otimes f^{(1)}, R^{-1}\rangle\langle\phi_{2}^{(1)}\otimes f^{(2)}g^{(1)},R\rangle\\
       &\langle\phi_3^{(1+\tau_3)}\otimes f^{(3)}g^{(2)}, (S^{\tau_3}\otimes id)(R)\rangle\langle\phi_4^{(1+\tau_4)}\otimes f^{(4)}g^{(3)}, (S^{\tau_4}\otimes id)(R)\rangle\\
       &\langle\phi_{1}^{(2)}\otimes S(f^{(5)}g^{(4)}),R\rangle(S(g^{(6)}f^{(7)})\phi_{1}^{(1)}\otimes S(g^{(7)}f^{(8)})\phi_{2}^{(3)}\otimes \phi_{3}^{(2-\tau_3)}\otimes \phi_{4}^{(2-\tau_4)})
    \end{split}
\end{equation}
\noindent Note the given two functions $\alpha,\ \beta\in K^{*}$ the following holds
\begin{equation}
    \langle\alpha^{(1)}\otimes\beta^{(1)}, R\rangle \, \alpha^{(2)}\beta^{(2)}=\langle\alpha^{(2)}\otimes\beta^{(2)}, R\rangle\, \beta^{(1)}\alpha^{(1)}
\end{equation}
which follows from $R \cdot \Delta \cdot R = \Delta^{\text{op}}$ (i.e. see Proposition 3.8 in \cite{Catherine1}).
Therefore we can write
\begin{equation}
    \begin{split}
       &\mathcal{B}_{p}^{f}\mathcal{B}_{p}^{g}(\phi_{1},\phi_{2},\phi_{3},\phi_{4})=\\
       &\langle g^{(7)}\otimes f^{(8)}, R\rangle\langle g^{(8)}\otimes f^{(1)}, R^{-1}\rangle\langle\phi_{2}^{(1)}\otimes f^{(2)}g^{(1)},R\rangle\\
       &\langle\phi_3^{(1+\tau_3)}\otimes f^{(3)}g^{(2)}, (S^{\tau_3}\otimes id)(R)\rangle\langle\phi_4^{(1+\tau_4)}\otimes f^{(4)}g^{(3)}, (S^{\tau_4}\otimes id)(R)\rangle\\
       &\langle \phi_{1}^{(2)}\otimes S(f^{(5)}g^{(4)}),R\rangle(S(f^{(6)}g^{(5)})\phi_{1}^{(1)}\otimes S(f^{(7)}g^{(6)})\phi_{2}^{(3)}\otimes \phi_{3}^{(2-\tau_3)}\otimes \phi_{4}^{(2-\tau_4)})
    \end{split}
\end{equation}
which shows that (\ref{eq::propCommBfBg}) holds and this concludes the proof. 
\end{proof}
We can now start considering the relationship between plaquette operators and gauge transformations, as given by the action of the vertex operator. 
We will first prove that vertex operators acting at a vertex not coinciding with a plaquette's starting vertex, will commute with the plaquette operator associated with that plaquette. 
\begin{proposition}
Consider the ribbon graph in Figure \ref{fig::simplest_plaquette} then the plaquette operator at
$p$ and the gauge transformation at vertex $w$ commute
\begin{equation}
    \mathcal{B}^{f}_{p} A_{w}^{h} = A_{w}^{h} \, \mathcal{B}^{f}_{p}
\end{equation}
\end{proposition}
\begin{proof}
The proof of the statement comes from a direct computation. From (\ref{eq::gaugeTransformation}) we can write
\begin{equation}
	\begin{split}
		  	&  A^{h}_w(\phi_1\otimes\ldots\otimes\phi_6)= 
			 \\
			 & \langle S^{\tau_{5}}(\phi_{5}^{(1+\tau_{5})})S(\phi_{2}^{(2)})\phi_{1}^{(1)}S^{\tau_{6}}(\alpha_{6}^{(1+\tau_{6})}), h\rangle(\phi_1^{(2)}\otimes\phi_2^{(1)}\otimes\phi_{3}\otimes\phi_{4}\otimes\phi_{5}^{(2-\tau_5)}\otimes\phi_{6}^{(2-\tau_6)})
	\end{split}
\end{equation} 
and therefore, by applying the plaquette operator to this equation, we find
\begin{equation}
\label{eq::proofBACommute}
    \begin{split}
        &\mathcal{B}^{f}_{p}A^{h}_w(\phi_1\otimes\ldots\otimes\phi_6)=\\
        &\langle S^{\tau_{5}}(\phi_{5}^{(1+\tau_{5})})S(\phi_{2}^{(3)})\phi_{1}^{(1)}S^{\tau_{6}}(\phi_{6}^{(1+\tau_{6})}), h\rangle\langle\phi_{2}^{(1)}\otimes f^{(1)},R\rangle\\
        &\langle\phi_{3}^{(1+\tau_3)}\otimes f^{(2)},(S^{\tau_3}\otimes id)(R)\rangle\langle\phi_{4}^{(1+\tau_4)}\otimes f^{(3)},(S^{\tau_4}\otimes id)(R)\rangle\langle\phi_{1}^{(3)}\otimes f^{(4)},R^{-1}\rangle\\
        &(S(f^{(5)})\phi_1^{(2)}\otimes S(f^{(6)})\phi_2^{(2)}\otimes\phi_{3}^{(2-\tau_3)}\otimes\phi_{4}^{(2-\tau_4)}\otimes\phi_{5}^{(2-\tau_5)}\otimes\phi_{6}^{(2-\tau_6)})\ .
    \end{split}
\end{equation}
On the other hand we have
\begin{equation}
    \begin{split}
        &A^{h}_{w}\mathcal{B}_{p}^{f}(\phi_1\otimes\ldots\otimes\phi_6)=\\
        &\langle S^{\tau_{5}}(\phi_{5}^{(1+\tau_{5})})S(\phi_{2}^{(3)})S(f^{(6)}S(f^{(7)}))\phi_{1}^{(1)}S^{\tau_{6}}(\phi_{6}^{(1+\tau_{6})}), h\rangle\langle\phi_{2}^{(1)}\otimes f^{(1)},R\rangle\\
        &\langle\phi_{3}^{(1+\tau_3)}\otimes f^{(2)},(S^{\tau_3}\otimes id)(R)\rangle\langle\phi_{4}^{(1+\tau_4)}\otimes f^{(3)},(S^{\tau_4}\otimes id)(R)\rangle\langle\phi_{1}^{(2)}\otimes f^{(4)},R^{-1}\rangle\\
        &(S(f^{(5)})\phi_{1}^{(2)}\otimes S(f^{(8)})\phi_{2}^{(2)}\otimes\phi_{3}^{(2-\tau_3)}\otimes\phi_{4}^{(2-\tau_4)}\otimes\phi_{5}^{(2-\tau_5)}\otimes\phi_{6}^{(2-\tau_6)})\ ,
    \end{split}
\end{equation}
which is the same as (\ref{eq::proofBACommute}) once we note that $f^{(6)}S(f^{(7)})=\epsilon^{*}(f^{(6)})$. This completes the proof. Note that the general case with an arbitrary number of edges at vertex $w$ is essentially the same as the current one since the plaquette's edges are always adjacent with respect to the ordering imposed by the cilia.
\end{proof}
We can now consider the commutation relation between plaquette and vertexes at the same vertex. Before undertaking this task, we need to prove the one last auxiliary proposition
\begin{proposition}
\label{proposition::last_lemma}
Consider $\phi,\ f\in K^{*}$, then for $ \tau=0,\ 1$, the following equation holds
\begin{equation}
\begin{split}
    &\langle\phi^{(1+\tau)}\otimes f^{(1)}, (S^{\tau}\otimes id)(R)\rangle S^{\tau}(\phi^{(2-\tau)(1+\tau)})f^{(2)}=\\  &\langle\phi^{(2-\tau)(1+\tau)}\otimes f^{(2)}, (S^{\tau}\otimes id)(R)\rangle f^{(1)}S^{\tau}(\phi^{(1+\tau)})
\end{split}
\end{equation}
\end{proposition}
\begin{proof}
The above equation follows from the identity
\begin{equation} \langle\alpha^{(1)}\otimes\beta^{(1)},R\rangle\alpha^{(2)}\beta^{(2)}
	=\langle\alpha^{(2)}\otimes\beta^{(2)},R\rangle\beta^{(1)}\alpha^{(1)}\ ,
	\label{eq::Interchange_Sweedler}
\end{equation}
which holds for every $\alpha,\ \beta\in K^{*}$, as a result of $R\Delta(h)=\Delta(h)R$. We can prove the statement directly for the two cases $\tau=0,\ 1$.
\paragraph{Case 1: $\tau=0$}
\begin{equation}
\begin{split}
    &\langle\phi^{(1)}\otimes f^{(1)}, R\rangle \phi^{(2)(1)}f^{(2)}=    \langle\phi^{(1)(1)}\otimes f^{(1)}, R\rangle \phi^{(1)(2)}f^{(2)}=\\
    &\langle\phi^{(1)(2)}\otimes f^{(2)}, R\rangle f^{(2)}\phi^{(1)(1)}
\end{split}
\end{equation}
which proves the first of the two equations.
\paragraph{Case 2: $\tau=1$}
\begin{equation}
\begin{split}
    &\langle\phi^{(2)}\otimes f^{(1)}, (S\otimes id)(R)\rangle S(\phi^{(1)(2)})f^{(2)}=\langle S(\phi^{(2)})^{(1)}\otimes f^{(1)}, R\rangle S(\phi^{(2)})^{(2)}f^{(2)}=\\
    &\langle S(\phi^{(2)})^{(2)}\otimes f^{(2)}, R\rangle f^{(2)}S(\phi^{(2)})^{(1)}\ ,
\end{split}
\end{equation}
that ends our proof.
\end{proof}
With this technical proposition, we are finally in a position to prove the following
\begin{proposition}
Consider the ribbon graph in Figure \ref{fig::Tadpole}, then we have
\begin{equation}
    \label{eq::last_proposition}
    A_{v}^{h}\mathcal{B}_{p}^{f}=\langle f^{(3)}, S(h^{(1)})\rangle\langle f^{(1)}, h^{(2)}\rangle  \mathcal{B}^{f(2)} A_{v}^{h^{(3)}}
\end{equation}
\end{proposition}
\begin{proof}
We will prove the statement by direct computation. 
Firstly we can see that
\begin{equation}
    \begin{split}
    &\mathcal{B}_{p}^{f}A^{h}_{v}(\phi_1\otimes\phi_2\otimes\phi_3\otimes\phi_4)=\\
    &\langle\phi_{2}^{(1)}S^{\tau_3}(\phi_{3}^{(1+\tau_3)})S^{\tau_4}(\phi_{4}^{(1+\tau_4)})S(\phi_{1}^{(3)}), h\rangle\langle\phi_{2}^{(2)}\otimes f^{(1)}, R\rangle\\
    &\langle\phi_{3}^{(2-\tau_3)(1+\tau_3)}\otimes f^{(2)}, (S^{\tau_3}\otimes id)(R)\rangle\langle\phi_4^{(2-\tau_4)(1+\tau_4)}\otimes f^{(4)}, (S^{\tau_4}\otimes id)(R)\rangle\\
    &\langle\phi_1^{(2)}\otimes f^{(4)}, R^{-1}\rangle(S(f^{(5)})\phi_1^{(1)}\otimes S(f^{(6)})\phi_2^{(3)}\otimes \phi_3^{(2-\tau_3)(2-\tau_3)}\otimes\phi_4^{(2-\tau_4)(2-\tau_4)})\ .
    \end{split}
\end{equation}
For the product on the right hand side of (\ref{eq::last_proposition}), we instead have
\begin{equation}
    \begin{split}
    &A_{v}^{h}\mathcal{B}_{p}^{f}(\phi_1\otimes\phi_2\otimes\phi_3\otimes\phi_4)=\\
    &\langle S(f^{(8)})\phi_{2}^{(2)}S^{\tau_3}(\phi_{3}^{(2-\tau_3)(1+\tau_3)})S^{\tau_4}(\phi_{4}^{(2-\tau_3)(1+\tau_4)})S(\phi_{1}^{(2)})f^{(5)}, h\rangle\langle\phi_{2}^{(1)}\otimes f^{(1)},R\rangle\\
    &\langle\phi_3^{(1+\tau_3)}\otimes f^{(2)}, (S^{\tau_3}\otimes id)(R)\rangle\langle\phi_4^{(1+\tau_4)}\otimes f^{(3)}, (S^{\tau_4}\otimes id)(R)\rangle\langle\phi_{1}^{(3)}\otimes f^{(4)},R^{-1}\rangle\\
    &(S(f^{(6)})\phi_1^{(1)}\otimes S(f^{(7)})\phi_2^{(3)}\otimes \phi_3^{(2-\tau_3)(2-\tau_3)}\otimes\phi_4^{(2-\tau_4)(2-\tau_4)}).
    \end{split}
\end{equation}
Now note that
\begin{equation}
    \begin{split}
    &\langle\phi_{1}^{(3)}\otimes f^{(4)},R^{-1}\rangle S(\phi_{1}^{(2)})f^{(5)}=
    \langle\phi_{1}^{(2)}\otimes f^{(5)},R^{-1}\rangle f^{(4)}S(\phi_{1}^{(3)})
    \end{split}
\end{equation}
and therefore 
\begin{equation}
    \begin{split}
    &A_{v}^{h}\mathcal{B}_{p}^{f}(\phi_1\otimes\phi_2\otimes\phi_3\otimes\phi_4)=\\
    &\langle S(f^{(8)})\phi_{2}^{(2)}S^{\tau_3}(\phi_{3}^{(2-\tau_3)(1+\tau_3)})S^{\tau_4}(\phi_{4}^{(2-\tau_3)(1+\tau_4)})f^{(4)}S(\phi_{1}^{(3)}), h\rangle\langle\phi_{2}^{(1)}\otimes f^{(1)},R\rangle\\
    &\langle\phi_3^{(1+\tau_3)}\otimes f^{(2)}, (S^{\tau_3}\otimes id)(R)\rangle\langle\phi_4^{(1+\tau_4)}\otimes f^{(3)}, (S^{\tau_4}\otimes id)(R)\rangle\langle\phi_{1}^{(2)}\otimes f^{(5)},R^{-1}\rangle\\
    &(S(f^{(6)})\phi_1^{(1)}\otimes S(f^{(7)})\phi_2^{(3)}\otimes \phi_3^{(2-\tau_3)(2-\tau_3)}\otimes\phi_4^{(2-\tau_4)(2-\tau_4)})\ .
    \end{split}
\end{equation}
We can now repeatedly apply Proposition \ref{proposition::last_lemma}, to get 
\begin{equation}
    \begin{split}
    &A_{v}^{h}\mathcal{B}_{p}^{f}(\phi_1\otimes\phi_2\otimes\phi_3\otimes\phi_4)=\\
    &\langle S(f^{(8)})\phi_{2}^{(2)}f^{(2)}S^{\tau_3}(\phi_{3}^{(1+\tau_3)})S^{\tau_4}(\phi_{4}^{(1+\tau_4})S(\phi_{1}^{(3)}), h\rangle\langle\phi_{2}^{(1)}\otimes f^{(1)},R\rangle\\
    &\langle\phi_3^{(2-\tau_3)(1+\tau_3)}\otimes f^{(3)}, (S^{\tau_3}\otimes id)(R)\rangle\langle\phi_4^{(2-\tau_3)(1+\tau_4))}\otimes f^{(4)}, (S^{\tau_4}\otimes id)(R)\rangle\\
    &\langle\phi_{1}^{(2)}\otimes f^{(5)},R^{-1}\rangle(S(f^{(6)})\phi_1^{(1)}\otimes S(f^{(7)})\phi_2^{(3)}\otimes \phi_3^{(2-\tau_3)(2-\tau_3)}\otimes\phi_4^{(2-\tau_4)(2-\tau_4)})\ .
    \end{split}
\end{equation}
Finally, since $\langle\phi_{2}^{(1)}\otimes f^{(1)},R\rangle\phi_{2}^{(2)}f^{(2)}=\langle\phi_{2}^{(2)}\otimes f^{(2)},R\rangle f^{(1)}\phi_{2}^{(1)}$, we can write
\begin{equation}
    \begin{split}
    &A_{v}^{h}\mathcal{B}_{p}^{f}(\phi_1\otimes\phi_2\otimes\phi_3\otimes\phi_4)=\\
    &\langle S(f^{(8)}), h^{(1)}\rangle\langle f^{(1)}, h^{(2)}\rangle\langle\phi_{2}^{(1)}S^{\tau_3}(\phi_{3}^{(1+\tau_3)})S^{\tau_4}(\phi_{4}^{(1+\tau_4})S(\phi_{1}^{(3)}), h^{(3)}\rangle\langle\phi_{2}^{(2)}\otimes f^{(2)},R\rangle\\
    &\langle\phi_3^{(2-\tau_3)(1+\tau_3)}\otimes f^{(3)}, (S^{\tau_3}\otimes id)(R)\rangle\langle\phi_4^{(2-\tau_3)(1+\tau_4))}\otimes f^{(4)}, (S^{\tau_4}\otimes id)(R)\rangle\\
    &\langle\phi_{1}^{(2)}\otimes f^{(5)},R^{-1}\rangle(S(f^{(6)})\phi_1^{(1)}\otimes S(f^{(7)})\phi_2^{(3)}\otimes \phi_3^{(2-\tau_3)(2-\tau_3)}\otimes\phi_4^{(2-\tau_4)(2-\tau_4)})\,
    \end{split}
\end{equation}
where (\ref{eq::last_proposition}) can be recognised. Note that as for the other propositions of this appendix, this property can be easily generalised to ribbon graphs with an arbitrary number of edges at vertex $v$.
\end{proof}

\end{document}